\documentclass[letterpaper, 11pt]{article}
\usepackage{fullpage,amsmath,amsthm,amssymb,dsfont} 
\usepackage{algorithm2e}
\usepackage{hyperref}
\usepackage{booktabs}
\usepackage{multirow}
\usepackage{cleveref}
\usepackage{xcolor}
\usepackage{comment}
\usepackage{graphicx}
\usepackage{tikz}
\usepackage{thm-restate}
\usepackage{algorithm2e}
\usepackage{makecell}
\usepackage[margin=1in]{geometry}
\usetikzlibrary{positioning, arrows.meta, shapes.geometric}
\RestyleAlgo{ruled}

\newcommand{\eps}{\varepsilon}
\newcommand{\T}{\mathcal{T}}

\newcommand{\E}{\mathbb{E}}
\newcommand{\rz}{\mathbf{z}}
\newcommand{\rx}{\mathbf{x}}
\newcommand{\one}{\mathds{1}}

\newcommand{\rs}{\mathbf{s}}
\newcommand{\R}{\mathbb{R}}

\newcommand{\cJ}{\mathcal{J}}
\newcommand{\I}{\mathcal{I}}

\newcommand{\rZ}{\mathbf{Z}}

\newcommand{\rc}{\mathbf{c}}
\newcommand{\raq}{\mathbf{q}}

\newcommand{\bq}{\widehat{\mathbf{m}}}

\newcommand{\rr}{\mathbf{r}}

\newcommand{\rM}{\mathbf{M}}

\newcommand{\ra}{\mathbf{a}}
\newcommand{\Vecd}{\mathcal{V}}
\newcommand{\rbeta}{\mathbf{b}}
\DeclareMathOperator{\mx}{max}
\DeclareMathOperator{\mn}{min}
\newcommand{\cL}{\mathcal{L}}
\newcommand{\ri}{\mathbf{i}}

\DeclareMathOperator*{\median}{median}

\DeclareMathOperator{\err}{err}

\DeclareMathOperator{\herdisc}{herdisc}
\DeclareMathOperator{\disc}{disc}
\DeclareMathOperator{\Anc}{Anc}
\DeclareMathOperator{\pAnc}{pAnc}

\newtheorem{theorem}{Theorem}

\newtheorem{lemma}{Lemma}
\newtheorem{definition}{Definition}
\newtheorem{corollary}{Corollary}

\newtheorem{innerrestatethm}{Restatement of Theorem}
\newenvironment{restatethm}[1]
  {\renewcommand\theinnerrestatethm{#1}\innerrestatethm}
  {\endinnerrestatethm}

\newcommand{\todo}[1]{\textcolor{red}{[TODO: #1]}}

\title{The Binary Tree Mechanism is Optimal for Differentially Private Continual Counting}

\author{Konstantina Bairaktari\footnotemark[1] \and Markus Engelund Dahl\footnotemark[1] \and Kasper Green Larsen\footnotemark[1]}

\date{}

\begin{document}

\maketitle
\thispagestyle{empty}

\footnotetext[1]{Department of Computer Science, Aarhus University. Supported by the European Union (ERC, TUCLA, 101125203). Views and opinions expressed are however those of the author(s) only and do not necessarily reflect those of the European Union or the European Research Council. Neither the European Union nor the granting authority can be held responsible for them.}

\begin{abstract}
Private continual counting is a fundamental problem in differential privacy: given a binary stream of length $n$, where each $1$ corresponds to the contribution of one individual, the goal is to release all running counts while protecting the privacy of each individual. For fixed privacy parameters, the standard binary tree mechanism achieves expected $\ell_\infty$ error $O(\log^{3/2} n)$ under approximate differential privacy and $O(\log^2 n)$ under pure differential privacy. Whether these dependences on the stream length are necessary has remained a central open problem.

For fixed $\eps\in(0,1)$, we prove a lower bound of $\Omega(\log^{3/2} n)$ under approximate DP with sufficiently small fixed $\delta>0$, and a lower bound of $\Omega(\log^2 n)$ under pure DP. These bounds establish the optimality of the binary tree mechanism in both settings. The bounds hold for arbitrary mechanisms, even when the entire stream is available in advance. Both proofs use the same decomposition and accumulation of residual noise along a tree.

As a consequence of the approximate-DP bound, we also obtain a largest-possible separation between hereditary discrepancy and private $\ell_\infty$ error for linear queries, showing that the known general upper bound in terms of hereditary discrepancy has the optimal dependence on the number of queries.
\end{abstract}

\section{Introduction}

How can an organization track the occurrence of sensitive events over time without compromising the privacy of individuals? Such scenarios arise naturally in many settings, including monitoring daily disease cases, counting user interactions on a platform, or measuring responses to a public poll. This motivates the study of \textit{private continual counting}, the focus of this paper. The standard formulation of the problem models time as discrete steps $1, \ldots, n$, where at each step a bit indicates whether an event occurred. Each $1$ corresponds to the contribution of one individual and each individual participates in at most one event. The goal is to release an accurate running count at every time step without revealing any individual's contribution.

The study of private continual counting was initiated by \cite{UpperBound1, UpperBound2}, who established the foundations for differentially private computations under continual observation. Differential privacy (DP) \cite{PureDPDefinition} is the standard framework for formal privacy guarantees, ensuring that the output of a mechanism reveals little about any individual's data. Private continual counting has found applications in a broad range of tasks including private online learning \cite{DPOnline, DPFTRL}, convex optimization \cite{PrivateSCOl1} and federated learning \cite{PrivateLearningwDPFTRL}. In many of these applications, it appears as a subroutine, and improving its accuracy directly translates to improvements in the downstream task. The appropriate notion of error, however, depends on the application: $\ell_\infty$ error (the expected worst-case error over all time steps) is a natural measure for monitoring tasks, while $\ell_2$ error (the expected root mean squared error over all time steps) is more relevant for learning applications \cite{l2LowerBound}. In this work, we focus on the $\ell_\infty$ error.

The most well-known method for private continual counting is the binary tree mechanism \cite{UpperBound1, UpperBound2}, which computes carefully structured subset counts in a differentially private manner and combines them to recover the running count at any given time step. 

For pure differential privacy, the binary tree mechanism uses Laplace noise to make the subset counts $\eps$-differentially private. Here, and throughout the paper, we focus on the regime where $\eps \in (0,1)$. While the original analysis achieves $\ell_\infty$ error $O(\log^{5/2} (n)/\varepsilon)$ \cite{UpperBound1,UpperBound2}, a more careful analysis of the same algorithm yields a bound of $O(\log^2 (n)/\eps)$. 
Prior to this work, the best known lower bound was $\Omega(\log (n)/\eps)$ \cite{UpperBound1}, leaving a multiplicative gap of $\log n$.

For approximate DP, the binary tree mechanism uses Gaussian noise instead of Laplace, yielding an $(\varepsilon, \delta)$-differentially private mechanism with $\ell_\infty$ error $O(\log^{3/2} (n) \sqrt{\log(1/\delta)}/\varepsilon)$ \cite{UpperBoundApprox}. Prior to this work, the best known lower bound related to the $\ell_\infty$ error under approximate DP was an $\Omega(\log (n)/\eps)$ lower bound on the square root of the expected value of the \emph{square} of the worst-case error over all $n$ time steps ($\sqrt{\ell_{\infty}^2}$ error)~\cite{ConstantUB}. 

Whether the $\ell_\infty$ error upper bounds are tight in their dependence on the stream length $n$ has remained an important open problem. In this work, we prove a lower bound of $\Omega(\log^{3/2}(n)/\max\{\eps,\delta\})$ for approximate DP and a lower bound of $\Omega(\log^2(n)/\eps)$ for pure DP. These bounds establish the optimality of the binary tree mechanism for pure DP and its optimal dependence on $n$ for approximate DP with fixed $\eps$ and sufficiently small fixed $\delta>0$. \Cref{tab:dp_bounds} compares our lower bounds with the binary tree upper bounds and the earlier lower bound for squared $\ell_\infty$ error.

\begin{table}[ht]
    \centering
    \begin{tabular}{llcc}
        \toprule
        \textbf{Setting} & \textbf{Error Notion} & \textbf{Upper Bound} & \textbf{Lower Bound} \\
        \midrule
        Pure DP
            & $\ell_\infty$ error
            & $\log^{2}(n)/\varepsilon$ & $\log^{2}(n)/\varepsilon$ [Our result; \cite{HarrisonLeeman2026}] \\
        \midrule
        \multirow{2}{*}{Approx. DP}
            & $\ell_\infty$ error
            & $\log^{3/2}(n)\log^{1/2}(1/\delta)/\varepsilon$ \cite{UpperBoundApprox} & $\log^{3/2}(n)/\max\{\varepsilon,\delta\}$ [Our result]\\
        \cmidrule{2-4}
            & $\sqrt{\ell_\infty^2}$ error
            & & $\log(n)/\varepsilon$ \cite{ConstantUB} \\
        \bottomrule
    \end{tabular}
    \caption{Comparison of our lower bounds with the binary tree upper bounds, together with the earlier lower bound for squared $\ell_\infty$ error. Our bounds hold for $n^{-1+\Omega(1)}<\varepsilon<1$, and the approximate-DP bound assumes $0<\delta<C$ for a sufficiently small constant $C>0$.}
    \label{tab:dp_bounds}
\end{table}

Before stating our results, we formally define the problem of private continual counting and the privacy notion we work with. Let $x \in \{0,1\}^n$ be the binary input vector, where $x_i = 1$ indicates that an event occurred at time step $i$ and since the mechanism operates in an online fashion, at time step $i$ only the entries $x_1, \ldots, x_i$ are visible to it. The goal is to approximate, at each time step $i$, the running count $\sum_{j\leq i}x_j$ in a differentially private manner. These $n$ running counts can be written compactly as $Ax$, where $A$ is the $n \times n$ lower triangular matrix of all ones, known as the prefix-sum matrix, with $(Ax)_i=\sum_{j \leq i} x_j$. This problem can be seen as a specific instance of linear queries (mechanisms that answer queries of the form $Ax$ for a fixed matrix $A$), where the query matrix is the prefix-sum matrix $A$. A private mechanism $\rM$ must output a vector $\rM(x) \in \R^n$, where $\rM(x)_i$ is the estimate released at time step $i$, so that $\|\rM(x)-Ax\|_\infty$ is small while ensuring privacy of any individual's entry in $x$. With this notation, the $\ell_\infty$ error is $\E[\|\rM(x)-Ax\|_\infty]$, whereas the $\ell_2$ error is $n^{-1/2}\E[\|\rM(x)-Ax\|_2]$.

We say that two binary vectors $x,x' \in \{0,1\}^n$ are neighboring if they differ in exactly one coordinate. The mechanism $\rM$ is required to satisfy differential privacy, defined as follows.
\begin{definition}[Differential Privacy \cite{PureDPDefinition}]
Fix $\eps > 0$ and $\delta \in [0,1)$. A mechanism $\rM$ is $(\eps,\delta)$-DP if for any measurable set $S \subseteq \R^n$,
\[\Pr[\rM(x) \in S] \leq e^\eps \Pr[\rM(x') \in S]+\delta\] for all pairs of neighboring inputs $x,x' \in \{0,1\}^n$. 
\end{definition}
The case where $\delta = 0$ is referred to as pure DP and the case where $\delta > 0$ is referred to as approximate DP. 

Our first result establishes a lower bound on the $\ell_\infty$ error of any $(\eps,\delta)$-DP mechanism for continual counting.

\begin{theorem}
    \label{thm:main}
    For any $n^{-1+\Omega(1)}  < \eps < 1$, $0 < \delta < C$ for a sufficiently small constant $C>0$, it holds that any $(\eps,\delta)$-DP mechanism $\rM$ for continual counting has error
    \[\max_{x \in \{0,1\}^n }\E[\|\rM(x)-Ax\|_\infty] = \Omega\left(\frac{\log^{3/2} n }{\max\{\eps,\delta\}}\right),\]
    where $A$ is the $n \times n$ prefix-sum matrix. 
\end{theorem}

For fixed $\eps$ and sufficiently small fixed $\delta>0$, this matches the dependence on $n$ of the Gaussian binary tree mechanism. For pure DP, we prove the following stronger bound.

\begin{theorem}
\label{thm:pure-dp}
For any $n^{-1+\Omega(1)}<\eps<1$, it holds that any $\eps$-DP
mechanism $\rM$ for continual counting has error
\[
  \max_{x\in\{0,1\}^n}\E[\|\rM(x)-Ax\|_\infty]
  =\Omega\left(\frac{\log^2 n}{\eps}\right),
\]
where $A$ is the $n\times n$ prefix-sum matrix.
\end{theorem}
This matches the $O(\log^2(n)/\eps)$ upper bound of the Laplace binary tree mechanism.

Both lower bounds hold in the offline observation model~\cite{SparseLowerBound}, where the mechanism receives the entire input stream at once. Since a mechanism for continual observation also yields a mechanism for offline observation, the bounds apply in particular to continual counting. Thus the difficulty already arises from releasing all prefix sums privately, without imposing the requirement that they be released online.

\paragraph{Concurrent Work.}
Concurrent with our extension to pure DP, Harrison and Leeman~\cite{HarrisonLeeman2026} also obtain tight lower bounds for pure-DP continual counting, building on our earlier approximate-DP manuscript. They further establish tight bounds for polynomially small $\delta$ (for fixed $\eps$) and for mean and maximum per-coordinate squared error. Both approaches use the tree geometry to infer the location of an input change and exploit group privacy over several changes. Our proof retains the tree basis, conditional noise predictions, restriction to a subtree, and reduction to a maximum residual noise value from the approximate-DP proof, replacing only the final linear-measurement attack by a randomized decoder. Their proof instead works directly with a distribution of binary streams and analyzes an exponential score on a four-ary tree, relating the score to the error through a telescoping potential.

\paragraph{Maximum Separation Between Privacy and Hereditary Discrepancy.}
As an interesting corollary of Theorem~\ref{thm:main}, we also derive the largest possible separation between hereditary discrepancy and the error of $(\eps,\delta)$-DP mechanisms for \emph{linear queries}.

Recall that continual counting, or prefix sums, can be seen as a special case where $A$ is the $n\times n$ prefix-sum matrix. More generally, for an arbitrary query matrix $A \in \R^{m \times n}$, a well-studied task is to design a private mechanism $\rM(x)$ for releasing an approximation of $Ax$ for $x \in \{0,1\}^n$. Let $\err^{\eps,\delta}_\infty(A)$ denote the smallest achievable $\ell_\infty$ error of an $(\eps,\delta)$-DP mechanism $\rM$:
\[
\err^{\eps,\delta}_\infty(A) := \inf_{\rM : \rM\textrm{ is }(\eps,\delta)\textrm{-DP}} \max_{x \in \{0,1\}^n} \E[\|\rM(x)-Ax\|_\infty].
\]

The quantity $\err^{\eps,\delta}_\infty(A)$ is known to be closely related to the notion of \emph{hereditary discrepancy}. The discrepancy of an $m \times n$ matrix $A$, denoted $\disc_\infty(A)$, is defined as $\disc_\infty(A) = \min_{x \in \{-1,1\}^n} \|Ax\|_\infty$. If $S \subseteq [n]$ denotes a subset of column indices and $A_{|S}$ the  matrix $A$ restricted to the columns indexed by $S$ then
\[
\herdisc_\infty(A) := \max_{S \subseteq [n]} \disc(A_{|S}).
\]
The ellipsoid mechanism $\rM$ of Nikolov, Talwar and Zhang~\cite{GeometrySparse} can be analyzed through the factorization norm framework of Matousek, Nikolov and Talwar~\cite{HereditaryDiscrepancy} to guarantee
\[
\max_{x \in \{0,1\}^n} \E[\|\rM(x)-Ax\|_\infty] = O\left(\eps^{-1} \sqrt{\log m \cdot \log(1/\delta)} \cdot \gamma_2(A)\right),
\]
where $\gamma_2(A)$ is the $\gamma_2$ factorization norm of $A$. It was further shown in~\cite{HereditaryDiscrepancy} that $\gamma_2(A) = O(\log m \cdot \herdisc_\infty(A))$. Combining it all gives
\[
\err^{\eps,\delta}_\infty(A) = O(\eps^{-1} \log^{3/2} m \cdot \sqrt{\log(1/\delta)}\cdot  \herdisc_\infty(A)).
\]
Ignoring the $\eps$ and $\delta$ dependencies, this is $\err^{\eps,\delta}_\infty(A) = O_{\eps,\delta}(\log^{3/2} m \cdot \herdisc_\infty(A))$. At the same time,~\cite{GeometrySparse} showed that $\err_\infty^{\eps,\delta}(A) = \Omega_{\eps,\delta}(\herdisc_\infty(A)/\log m)$ (with a slight twist in the lower bound requiring the definition of $\herdisc_\infty(A)$ to be only over $S \subseteq [n]$ of linear size). In summary, there is a constant $C_{\eps,\delta}>0$ (depending on $\eps$ and $\delta$) such that
\[
 \herdisc_\infty(A)/(C_{\eps,\delta} \log m) \leq \err_\infty^{\eps,\delta}(A) \leq C_{\eps,\delta}\herdisc_\infty(A) \log^{3/2} m.
\]
If $A$ denotes the prefix-sum matrix, we have $\herdisc_\infty(A) = 1$. Our lower bound in Theorem~\ref{thm:main} thus shows that the general upper bound $\err_\infty^{\eps,\delta}(A) = O_{\eps,\delta}(\herdisc_\infty(A) \log^{3/2} m)$ cannot be improved. Prior to our work, the largest separation was a factor $\log m$, obtainable from prior lower bounds for prefix sums.

\paragraph{Overview.}
Both proofs start from the same reduction inspired by prior work \cite{ObliviousLinear} that allows us to assume that the mechanism $\rM$ is \emph{oblivious}, meaning that $\rM(x)=Ax+\rz$ for a noise vector $\rz$ independent of $x$, but possibly with correlations between its coordinates. We thus need to lower bound $\E[\|\rz\|_\infty]$. We express $\rz$ in a basis induced by a binary tree and show that local residual noise accumulates along root-to-leaf paths. This reduces the lower bound on $\E[\|\rz\|_\infty]$ to showing that at least one residual noise value in each suitable selection of nodes is large.

For approximate DP, we obtain this local lower bound by using linear measurements to estimate one input bit. If all the selected residual noise values were too small, the estimate would violate privacy. For pure DP, we strengthen this last step via a packing argument where we place a larger change at an unknown leaf and use the noisy counts to guide a randomized search for its location. This decoding attack is inspired by the randomized tree decoder of Ghazi et al.~\cite{GhaziEtAl2023Trees} for private tree aggregation. The decoder need only improve on uniform guessing by an exponential factor in the tree height. Group privacy then forces a larger residual noise value, and the same accumulation argument gives the additional factor of $\sqrt{\log n}$ in the final lower bound. We believe the techniques developed in this proof introduce several ideas that may find further applications in lower bound proofs for differential privacy. We give a detailed presentation of the main ideas in Section~\ref{sec:overview}, prove the approximate-DP bound in Section~\ref{sec:formal-proof}, and give the modified attack for pure DP in Section~\ref{sec:pure-dp}.

\section{Related Work}

\paragraph{Private Continual Counting.} Private continual counting was introduced independently by \cite{UpperBound1} and \cite{UpperBound2}, who proposed the binary tree mechanism and established the foundations of differentially private computation under continual observation. While the original analysis achieves $\ell_\infty$ error $O_{\eps,\delta}(\log^{5/2} n)$ for pure DP and $O_{\eps,\delta}(\log^{2}n)$ for approximate DP, a more careful analysis of the same algorithm yields improved bounds of $O_{\eps,\delta}(\log^2n)$ and $O_{\eps,\delta}(\log^{3/2}n)$ \cite{UpperBoundApprox} respectively. Additionally, \cite{l2LowerBound} introduced a mechanism achieving $\ell_2^2$ error (that is, $\max_{x\in \{0,1\}^n} \E[n^{-1}\|\rM(x)-Ax\|_2^2]$) of $O_{\eps,\delta}(\log^2 n)$ and a matching lower bound in terms of the dependence on $n$. This implies an $\Omega_{\eps,\delta}(\log n)$ lower bound for the $\sqrt{\ell_\infty^2}$ error. \cite{EfficientBinary} and \cite{ImprovedPure} proposed variants of tree mechanisms achieving better error in terms of the constant factors compared to the standard binary tree mechanism for the offline setting and the pure DP continual setting respectively. Since the binary tree mechanism is a special case of matrix factorization mechanisms, other approaches to improving the error include finding alternative matrix factorizations \cite{ConstantUB, l2LowerBound}. The best previously known lower bounds were obtained via packing arguments or by lower bounding the error of matrix factorization mechanisms \cite{UpperBound1, l2LowerBound}. Our work complements this line of research by proving matching lower bounds for $\ell_\infty$ error under pure and approximate DP with fixed privacy parameters.

\paragraph{The Sparse Case.} A natural variant of private continual counting assumes that the total number of events (number of $1$'s in $x$) is bounded by $n_e\ll n$. In this sparse setting, \cite{RectangleQueries} gave an upper bound of $O_{\eps,\delta}(\min\{\log n + \log^2 n_e, n_e\})$, which can significantly improve over the binary tree mechanism when $n_e$ is small. \cite{SparseLowerBound} proved a lower bound of $\Omega_{\eps,\delta}(\min\{\log n,n_e \})$, showing that this improvement is tight. In our setting, we do not assume a bound on $n_e$, so the results apply with $n_e \leq n$, recovering the suboptimal $O_{\eps,\delta}(\log^2 n)$ upper bound and $\Omega_{\eps,\delta}(\log n)$ lower bound.

\paragraph{Applications.} Private continual counting has found applications across a range of machine learning and optimization tasks. The work of \cite{DPOnline} was the first to use it for differentially private online learning via a follow-the-leader approach combined with the binary tree mechanism. Subsequent works extended this to the full-information and bandit setting \cite{DPFTRL,DPFTRL2}. These techniques were applied to federated learning \cite{PrivateLearningwDPFTRL}, enabling practical private deep learning without sampling or shuffling.  More recently, \cite{CorrelatedNoise} showed that DP-FTRL, a correlated noise mechanism based on continual counting, provably outperforms DP-SGD for differentially private learning. In the context of private stochastic convex optimization, variations of the binary tree mechanism have been used to extend the Frank Wolfe algorithm to differential privacy \cite{PrivateSCOl1, PrivateSCO}. In all these applications, improving the accuracy of continual counting translates to improved accuracy in the downstream task.

\paragraph{Related Continual Observation Problems.} Several other problems in the continual observation model are closely related to private continual counting. The problem of privately counting distinct elements in the case where items can be both inserted and deleted was studied by \cite{CountingDistinct, DistinctExtension, EfficientDistinct}, where continual counting on a difference stream serves as a subroutine. \cite{FrameworkPrivateSums} studied continual decaying sums, a generalization of continual counting where older observations are down-weighted over time. Continual counting with privacy expiration, where the privacy guarantee degrades gracefully over time, was studied by \cite{GradualPrivacy}. Graph problems under continual observation, including triangle counting and other graph statistics, were studied by \cite{ContinualGraphs, Graphs}, with continual counting appearing as a key subroutine in several of these results.

\section{Proof Overview and Main Ideas}
\label{sec:overview}
In this section, we present the high level ideas and overall structure of our lower bound proof for private continual counting. Along the way, we present the intuition behind each step and explain the barriers we overcome. We first explain the approximate-DP proof and then show how to strengthen its final privacy attack for pure DP. Throughout, we assume that the mechanism $\rM$ is oblivious, i.e.\ $\rM(x)=Ax+\rz$ for a noise vector $\rz$ independent of $x$ (see~\Cref{sec:formal-proof} for the justification).

\paragraph{Bounding Noise via Linear Measurements.}
In order to best motivate our approach to lower bounding the noise $\E[\|\rz\|_{\infty}]$, let us start by sketching an approach for lower bounding $\E[\|\rz\|_2^2]$, and then argue how we adapt it to a local notion of noise we call \textit{residual noise}, which will ultimately yield a lower bound on $\E[\|\rz\|_{\infty}]$. 

Fix a coordinate $i \in [n]$. If $\rM(x) = Ax + \rz$ is $(\eps,\delta)$-DP for small enough constants $\eps,\delta>0$, then it must be the case that for any neighboring pair of inputs $x',x'' \in \{0,1\}^n$ that differ in the $i$'th input bit, if $\rx$ is chosen uniformly among $\{x',x''\}$, then it is not possible for an adversary to guess $\rx$ from $\rM(\rx)=A\rx + \rz$ with high probability. Importantly, this must be the case even when given knowledge of the pair $x',x''$. Since $x',x''$ differ only in the $i$'th input bit, the adversary just needs to guess the value of $\rx_i$.

Without loss of generality we consider that $x'_i=0$ and $x''_i=1$, whereas $x'_j = x''_j$ for all $j \neq i$. Let $\rx$ be uniform in $\{x',x''\}$. We consider an adversary that estimates $\rx_i$ in the following way. They first compute $Ax'$ and subtract it from $\rM(\rx)$ to obtain $A\rx - Ax' + \rz = \rx_iA e_i + \rz$. The idea is now to define a \emph{linear measurement} $\psi^i \in \R^n$ for $i$ (that is, a vector defining a linear form $y \to \langle \psi^i, y\rangle$ applied to a vector $y \in \R^n$), and let the adversary compute the estimate
\[
\hat{\rx}_i = \frac{\langle\psi^{i},  \rx_{i} Ae_{i} + \rz\rangle}{\langle \psi^{i}, A e_{i}  \rangle} = \rx_{i} + \frac{\langle\psi^{i}, \rz \rangle}{\langle \psi^{i},A e_{i} \rangle}.
\]

Since $\hat \rx_i$ is a function of $\rM(\rx)-Ax'$, it is $(\eps,\delta)$-differentially private by the postprocessing property of differential privacy.  We can apply the following lemma to $\hat \rx_i$ because it is an estimate of the binary variable $\rx_i$ with additive noise $\langle\psi^{i}, \rz \rangle/\langle \psi^{i},A e_{i} \rangle$.

\begin{lemma}[\cite{GengViswanath2016ApproxDPNoise}]
\label{lem:releaseVar}
Any $(\eps,\delta)$-DP mechanism $\rM$ releasing a private binary variable $x \in \{0,1\}$ as $\rM(x) = x + \rz$ for a noise variable $\rz$ independent of $x$ must have $\E[\rz^2] = \Omega(\min\{\eps^{-2}, \delta^{-2}\})$.
\end{lemma}

We conclude that the error of the adversary is
\[
\frac{\E[\langle \psi^{i}, \rz \rangle^2]} {\langle \psi^{i}, Ae_{i} \rangle^2}= \Omega(\min\{\eps^{-2},\delta^{-2}\}).
\]
Intuitively, this says that the noise $\rz$, when measured in the direction of $\psi^{i}$, must be large relative to the signal $\langle \psi^i,Ae_i\rangle$. Otherwise, the adversary could distinguish $\rx_i = 0$ from $\rx_i=1$, violating privacy. Overall, it implies that for all $i \in [n]$, $\E[\langle \psi^{i}, \rz \rangle^2]=\Omega(\langle \psi^{i}, Ae_{i} \rangle^2\min\{\eps^{-2},\delta^{-2}\})$.

The idea is now to carefully choose the linear measurements $\psi^i$ as orthogonal and unit length. In this way, we have $\sum_i \E[\langle \psi^i, \rz\rangle^2] \leq \E[\|\rz\|_2^2]$. 
Therefore,
\[
\E[\|\rz\|^2_2] \geq \sum_{i}\E[\langle \psi^{i}, \rz \rangle^2] = \Omega\left(\sum_i\langle \psi^{i}, Ae_{i} \rangle^2 \cdot \min\{\eps^{-2}, \delta^{-2}\}\right).
\]
Variations of this idea can be used to prove the known $\E[\|\rz\|_2^2] = \Omega(n \log^2 n)$ lower bound by choosing $\psi^i = UV^te_i$ where $A=U\Sigma V^t$ is the singular value decomposition of $A$. Unfortunately, this only yields a lower bound of $\Omega(\log^2 n)$ on $\E[\|\rz\|_\infty^2]$.  

\paragraph{Tree Basis.} Inspired by the approach just outlined, we next attempt to adapt these ideas to lower bounding $\E[\|\rz\|_\infty]$. The main barrier is that if an adversary again computes $\rx_i Ae_i + \rz$ and takes linear measurements $\psi^i$, then we inevitably lower bound $\E[\|\rz\|_2^2]$. Our first task is thus to find a different decomposition of the mechanism output so that the adversary can instead compute vectors  of the form $\rx_i \nu^i + \raq$, where $\raq$ is a noise vector whose entries reflect a local rather than global contribution of $\rz$, and $\nu^i$ is chosen so that $\sum_i \langle  \psi^i, \nu^i \rangle^2$ is large, for some linear measurements $\psi^i$.

Recall that the proof has two main parts. We now describe the first: rewriting the noise in a basis more closely resembling the binary tree mechanism, which converts the task of lower bounding $\E[\|\rz\|_\infty]$ into lower bounding a weighted sum of residual noise contributions. The basis is constructed as follows. Assuming for simplicity that $n$ is a power of two, consider a perfect binary tree $\T$ of height $h = \log_2 n$, i.e.\ $\T$ has $n$ leaves. We think of the $n$ leaves as being associated with the coordinates of $Ax+\rz$ so that the $i$-th leaf corresponds to $(Ax+\rz)_i$.

We now define a set of vectors from $\T$ that together form a basis for $\R^n$. For every node $u \in \T$ that is
either the root or the right child of its parent, we define a vector $\chi^u$
which is the indicator vector of the leaves in the subtree rooted at $u$. 
There are exactly $n$ such vectors, and they form a basis of $\R^n$ (though not an orthonormal basis). See \Cref{fig:tree} for an example of the tree basis for $n=8$ leaves.

\begin{figure}[ht]
\centering
\scalebox{0.75}{
\begin{tikzpicture}[
    every node/.style={font=\small},
    filled/.style={circle, draw, fill=black!80, text=white,
                   minimum size=1.6em, inner sep=2pt},
    empty/.style={circle, draw, fill=white,
                  minimum size=1.6em, inner sep=2pt},
    edge/.style={-},
    lbl/.style={font=\footnotesize, align=center},
    depth/.style={font=\scriptsize, text=gray},
]

\node[depth] at (-7.0, 0)    {$d=0$};
\node[depth] at (-7.0, -1.6) {$d=1$};
\node[depth] at (-7.0, -3.2) {$d=2$};
\node[depth] at (-7.0, -4.8) {$d=3$};
\draw[dashed, gray!40, thin] (-6.6, 0.6) -- (-6.6, -5.4);

\node[lbl] at (0, 0.9)  {$\chi^{1\text{-}8}$};
\node[lbl, text=gray!70] at (0, 0.55) {\scriptsize \texttt{11111111}};
\node[filled] (r) at (0, 0) {};

\node[empty]  (l0) at (-3, -1.6) {};
\node[filled] (l1) at ( 3, -1.6) {};
\node[lbl] at (3, -0.7)  {$\chi^{5\text{-}8}$};
\node[lbl, text=gray!70] at (3, -1.05) {\scriptsize \texttt{00001111}};

\node[empty]  (l00) at (-4.5, -3.2) {};
\node[filled] (l01) at (-1.5, -3.2) {};
\node[empty]  (l10) at ( 1.5, -3.2) {};
\node[filled] (l11) at ( 4.5, -3.2) {};
\node[lbl] at (-1.5, -2.3)  {$\chi^{3\text{-}4}$};
\node[lbl, text=gray!70] at (-1.5, -2.65) {\scriptsize \texttt{00110000}};
\node[lbl] at (4.5, -2.3)   {$\chi^{7\text{-}8}$};
\node[lbl, text=gray!70] at (4.5, -2.65) {\scriptsize \texttt{00000011}};

\node[empty]  (l000) at (-5.25, -4.8) {};
\node[filled] (l001) at (-3.75, -4.8) {};
\node[empty]  (l010) at (-2.25, -4.8) {};
\node[filled] (l011) at (-0.75, -4.8) {};
\node[empty]  (l100) at ( 0.75, -4.8) {};
\node[filled] (l101) at ( 2.25, -4.8) {};
\node[empty]  (l110) at ( 3.75, -4.8) {};
\node[filled] (l111) at ( 5.25, -4.8) {};

\node[lbl] at (-3.75, -3.9) {$\chi^{2}$};
\node[lbl, text=gray!70] at (-3.75, -4.25) {\scriptsize \texttt{01000000}};
\node[lbl] at (-0.75, -3.9) {$\chi^{4}$};
\node[lbl, text=gray!70] at (-0.75, -4.25) {\scriptsize \texttt{00010000}};
\node[lbl] at ( 2.25, -3.9) {$\chi^{6}$};
\node[lbl, text=gray!70] at ( 2.25, -4.25) {\scriptsize \texttt{00000100}};
\node[lbl] at ( 5.25, -3.9) {$\chi^{8}$};
\node[lbl, text=gray!70] at ( 5.25, -4.25) {\scriptsize \texttt{00000001}};

\draw[edge] (r)   -- (l0);
\draw[edge] (r)   -- (l1);
\draw[edge] (l0)  -- (l00);
\draw[edge] (l0)  -- (l01);
\draw[edge] (l1)  -- (l10);
\draw[edge] (l1)  -- (l11);
\draw[edge] (l00) -- (l000);
\draw[edge] (l00) -- (l001);
\draw[edge] (l01) -- (l010);
\draw[edge] (l01) -- (l011);
\draw[edge] (l10) -- (l100);
\draw[edge] (l10) -- (l101);
\draw[edge] (l11) -- (l110);
\draw[edge] (l11) -- (l111);

\draw[gray!40, thin] (-5.8, -5.4) -- (5.8, -5.4);
\foreach \x/\i in {-5.25/1, -3.75/2, -2.25/3, -0.75/4,
                    0.75/5,  2.25/6,  3.75/7,  5.25/8}{
    \node[font=\footnotesize, text=gray] at (\x, -5.7) {$\i$};
}
\node[font=\scriptsize, text=gray] at (0, -6.0) {leaf index $i$};

\node[filled, minimum size=1.0em] (leg1) at (-2.5, -6.8) {};
\node[font=\footnotesize, right=0.15cm of leg1, text=black]
    {root or right child ($\chi^u$ defined)};
\node[empty, minimum size=1.0em] (leg2) at (-2.5, -7.4) {};
\node[font=\footnotesize, right=0.15cm of leg2, text=black]
    {left child (no $\chi^u$)};

\end{tikzpicture}
}
\caption{Binary tree basis for $n=8$. Filled nodes are the root or right
children and have a $\chi^u$ vector defined, shown with its binary
indicator string. Empty nodes are left children and have no $\chi^u$ vector. For clarity, each vector $\chi^u$ in this figure uses the range of leaves in the subtree rooted at $u$ to denote $u$, e.g. the root has $\chi^{1-8}$. The depth of the tree level is denoted by $d$.}
\label{fig:tree}
\end{figure}
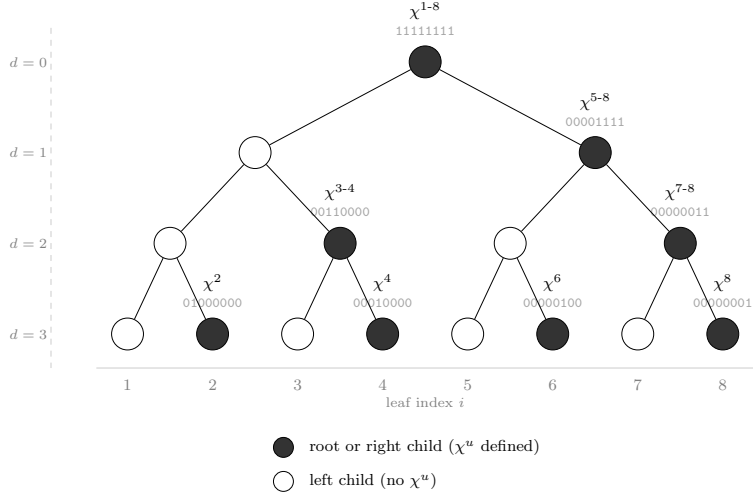

The reason for using this basis is that every suffix vector $Ae_i$ has a
simple representation in it. More precisely, if $Ae_i$ is the vector
with ones in positions $j \geq i$ and zeros elsewhere, then
\[
    Ae_i = \sum_{u \in \I(i)} \chi^u
\]
for the set $\I(i)$ of nodes that are right children of the path to the $(i-1)$'th
leaf from the root (we refer the reader to~\Cref{sec:formal-proof} for the precise statement). Letting $\Vecd(\T) \subset \T$ be the set of nodes $u$ with a $\chi^u$ vector defined (right children and the root), we see that for every input $x$, we may write
\[
    Ax = \sum_{u \in \Vecd(\T)} a_u \chi^u,
    \qquad
    a_u = \sum_{i : u \in \I(i)} x_i.
\]
Similarly, since the $\chi^u$ vectors form a basis, the noise vector also has a unique expansion
\[
    \rz = \sum_{u \in \Vecd(\T)} \rbeta_u \chi^u.
\]
For a leaf $i$, the released noise in coordinate $i$ is therefore the sum of
the coefficients $\rbeta_u$ over the nodes $u$ on the root-to-leaf path to
$i$ for which $\chi^u$ is defined:
\[
    \rz_i = \sum_{u \in \Anc(i) \cap \Vecd(\T)} \rbeta_u,
\]
where $\Anc(i)$ are the ancestors of the $i$'th leaf (including the $i$'th leaf). Thus the $\ell_\infty$ error is controlled by these path sums.

\paragraph{Noise Accumulation.} Having established the tree basis, we now analyze how noise accumulates along its paths. To define the vectors $\nu^i$ and $\raq$, we need to understand how the noise $\rz$ is distributed across the tree. Specifically, we argue that noise accumulates up the tree when it has ``residual uncertainties'', a notion we now formalize.

Consider a node $u \in \T$, let $\T^u$ be the subtree rooted at $u$, and let $\Vecd(\T^u) = \T^u \cap \Vecd(\T)$ be the set of nodes $v$ in $\T^u$ with a $\chi^v$ vector defined. We now consider the accumulation of the noise $\rz_i$ up to node $u$ defined as
\[
\rz_i^{< u} := \sum_{v \in (\Vecd(\T^u) \setminus \{u\}) \cap \Anc(i)} \rbeta_v.
\]
Thus $\rz_i^{<u}$ gives the contribution to $\rz_i$ from nodes below $u$. Our next key idea is to lower bound $\E[\max_i \rz_i - \min_i \rz_i]$ rather than $\E[\|\rz\|_\infty]$ directly, avoiding working with absolute values. This is sufficient since $\|\rz\|_\infty\geq (\max_i \rz_i - \min_i \rz_i)/2$. If $\cL(\T^u)$ denotes the indices of the leaves of $\T^u$, we may now consider
\[
\mx^{< u}(\rz) := \max_{i \in \cL(\T^u)} \rz_i^{< u}
\]
and
\[
\mn^{< u}(\rz) := \min_{i \in \cL(\T^u)} \rz^{< u}_i.
\]
Our main inductive argument is to show that $\E[\mx^{<u}(\rz)-\mn^{<u}(\rz)]$ increases as we move up the tree towards the root. Finally, for the root $\rho$, we have 
\[
\mx^{<\rho}(\rz)-\mn^{<\rho}(\rz) = (\max_i \rz_i - \rbeta_\rho)-(\min_i \rz_i - \rbeta_\rho) = \max_i \rz_i - \min_i \rz_i,
\]
and thus a lower bound on $\E[\mx^{<\rho}(\rz)-\mn^{<\rho}(\rz) ]$ directly translates to a lower bound on $\E[\|\rz\|_\infty]$.

Consider now a non-leaf node $u$ and the range $\mx^{<u}(\rz)-\mn^{<u}(\rz)$. Let $\ell(u)$ denote the left child of $u$ and $r(u)$ the right child. We see that
\[
\mx^{<u}(\rz)-\mn^{<u}(\rz) = \max\{\mx^{<\ell(u)}(\rz), \mx^{<r(u)}(\rz) + \rbeta_{r(u)} \} - \min\{\mn^{<\ell(u)}(\rz), \mn^{<r(u)}(\rz) + \rbeta_{r(u)} \}.
\]
This holds since every $i \in \T^{r(u)}$ has $r(u)$ as an ancestor (only the right child of $u$ has a $\chi$-vector). Using that $\max\{a,b\}-\min\{c,d\} = (a-c)/2+(b-d)/2 + |a-b|/2+|c-d|/2$ we get
\begin{align}
\label{eq:recurse}
\mx^{<u}(\rz)-\mn^{<u}(\rz) &=  \frac{\mx^{<\ell(u)}(\rz)-\mn^{<\ell(u)}(\rz)}{2}+\frac{\mx^{<r(u)}(\rz)-\mn^{<r(u)}(\rz)}{2} \nonumber \\
&+\frac{\left|\mx^{<\ell(u)}(\rz)-\mx^{<r(u)}(\rz)-\rbeta_{r(u)} \right|}{2}+\frac{\left|\mn^{<\ell(u)}(\rz)-\mn^{<r(u)}(\rz) - \rbeta_{r(u)} \right|}{2}.
\end{align}
The first line is very convenient as it expresses $\mx^{<u}(\rz)-\mn^{<u}(\rz)$ recursively as the average of the same quantity in the two children's subtrees. What remains is thus to argue that the second line introduces a noticeable growth in $\max^{<u}(\rz)-\min^{<u}(\rz)$ as we move up the tree.

Taking expectation on both sides of~\eqref{eq:recurse}, we thus need to lower bound 
\begin{align}
\E\left[\left|\mx^{<\ell(u)}(\rz)-\mx^{<r(u)}(\rz)-\rbeta_{r(u)} \right|\right], \label{eq:topredict}
\end{align}
and symmetrically for $\min$. Here our main observation is that $\mx^{<\ell(u)}(\rz)$ and $\mx^{<r(u)}(\rz)$ depend only on the noise \emph{below} $\ell(u)$ and $r(u)$. We thus define a type of ``residual noise''. Let $\Delta^{<u}(\rz)$ denote the random variable giving the noise $\rbeta_v$ for every $v$ in the subtree below $u$ (excluding $u$). We now define $\sigma_{u}$ to capture how well $\rbeta_{r(u)}$ can be predicted from the noise $\Delta^{<\ell(u)}(\rz), \Delta^{<r(u)}(\rz)$:
\[
\sigma_{u}(\Delta^{<\ell(u)}(\rz), \Delta^{<r(u)}(\rz)) := \inf_m \E[|\rbeta_{r(u)}-m| \mid \Delta^{<\ell(u)}(\rz), \Delta^{<r(u)}(\rz)].
\]

Returning to~\eqref{eq:topredict} and recalling that $\mx^{<\ell(u)}(\rz)$ and $\mx^{<r(u)}(\rz)$ depend only on $\Delta^{<\ell(u)}(\rz), \Delta^{<r(u)}(\rz)$ we thus conclude that~\eqref{eq:topredict} is at least
\begin{align*}
    \sigma_{u} := \E[\sigma_{u}(\Delta^{<\ell(u)}(\rz), \Delta^{<r(u)}(\rz))].
\end{align*}
We thus get an increase in $\E[\max^{<u}(\rz)-\min^{<u}(\rz)]$ if there is still randomness left in $\rbeta_{r(u)}$ conditioned on the noise in the subtrees below $\ell(u)$ and $r(u)$. In particular, recursively unfolding~\eqref{eq:recurse} and letting $d(u)$ denote the depth of a node $u \in \T$ with the root $\rho$ at depth $0$, we get
\begin{align}
\label{eq:unfolded}
    \E[\|\rz\|_\infty] &\geq \E[\max_i \rz_i - \min_i \rz_i]/2 \nonumber\\
    &= \E[\mx^{<\rho}(\rz) - \mn^{<\rho}(\rz)]/2 \nonumber\\
    &\geq \frac{1}{4} \sum_{d=0}^{h-1} \sum_{u \in \T : d(u)=d} 2^{-d} \sigma_{u}.
\end{align}
Observe that $\sum_{d=0}^{h-1} \sum_{u \in \T : d(u)=d} 2^{-d} \sigma_{u}$ equals the expected sum of $\sigma_{u}$ values along a uniform random root-to-leaf path in $\T$.

Before proceeding to explain how we lower bound this sum of $\sigma_{u}$ values, let us compare to the classic binary tree mechanism upper bound. There, every node in the binary tree adds independent noise of variance $\Theta(\log n)$. This means that even conditioned on the noise in the subtrees below $\ell(u)$ and $r(u)$, we still have $\E[\inf_m \E[|\rbeta_{r(u)}-m| \mid \Delta^{<\ell(u)}(\rz), \Delta^{<r(u)}(\rz)]] = \Omega(\sqrt{\log n})$. Summing $\Omega(\sqrt{\log n})$ for each of the $\log_2 n$ nodes along a root-to-leaf path leads to the claimed $\Omega(\log^{3/2} n)$ lower bound.

By~\eqref{eq:unfolded}, our goal reduces to showing that no private mechanism can make do with small residual noise $\sigma_u$. 

\paragraph{Obstacles for Bounding Residual Noise via Linear Measurements.}

The main obstacle is that the adversary cannot directly compute the subtree noise $\Delta^{<\ell(u)}(\rz), \Delta^{<r(u)}(\rz)$ needed to predict $\rbeta_{r(u)}$. To see why, consider the following attempt: let the adversary compute $\rx_i Ae_i + \rz$, write it (uniquely) in the basis $\chi^u$ and thereby obtain, for every node $u \in \Vecd(\T)$, the value 
\[
\rbeta_u + \rx_i \one\{u \in \I(i)\}.
\]
Next, for every non-leaf node $u\in\T$, suppose the adversary could compute $\Delta^{<\ell(u)}(\rz),\Delta^{<r(u)}(\rz)$. They could then subtract the conditional median $m^u(\Delta^{<\ell(u)}(\rz),\Delta^{<r(u)}(\rz))$, chosen as the smallest minimizer of
\[
  \E[|\rbeta_{r(u)}-m|\mid\Delta^{<\ell(u)}(\rz),\Delta^{<r(u)}(\rz)]
\]
over $m$, from $\rbeta_{r(u)}+\rx_i\one\{r(u)\in\I(i)\}$. The resulting residual noise
\[
  \raq_u:=\rbeta_{r(u)}-m^u(\Delta^{<\ell(u)}(\rz),\Delta^{<r(u)}(\rz))
\]
satisfies $\E[|\raq_u|]=\sigma_u$ (see \Cref{fig:adversary-tree} for an illustration). This would produce the desired vector $\rx_i\nu^i+\raq$, where both $\nu^i$ and $\raq$ are indexed over non-leaf nodes $u\in\T$, with $\nu^i$ having entries $\one\{r(u)\in\I(i)\}$.

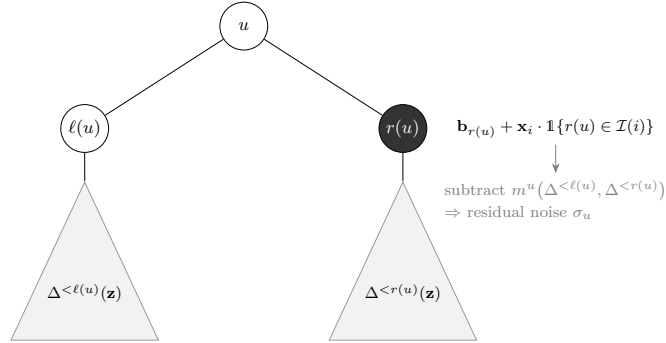
\begin{figure}[ht]
\centering
\scalebox{0.75}{
\begin{tikzpicture}[
    every node/.style={font=\small},
    filled/.style={circle, draw, fill=black!80, text=white,
                   minimum size=2.4em, inner sep=0pt},
    empty/.style={circle, draw, fill=white,
                  minimum size=2.4em, inner sep=0pt},
    subtree/.style={isosceles triangle, draw=black!40,
                    fill=gray!10,
                    shape border rotate=90, minimum height=2.2cm,
                    minimum width=2.6cm, inner sep=0pt,
                    isosceles triangle apex angle=50},
    edge/.style={-},
    ann/.style={->, >=Stealth, gray, thin},
]

\node[empty] (u)  {$u$};
\node[empty,  below left=1.2cm and 2.2cm of u]  (lu) {\small$\ell(u)$};
\node[filled, below right=1.2cm and 2.2cm of u] (ru) {\small$r(u)$};

\draw[edge] (u) -- (lu);
\draw[edge] (u) -- (ru);

\node[subtree, below=0.5cm of lu] (Tl) {};
\node[subtree, below=0.5cm of ru] (Tr) {};

\draw[edge] (lu.south) -- (Tl.north);
\draw[edge] (ru.south) -- (Tr.north);

\node[above=0.55cm of Tl.south, align=center]
    {\footnotesize $\Delta^{<\ell(u)}(\rz)$};
\node[above=0.55cm of Tr.south, align=center]
    {\footnotesize $\Delta^{<r(u)}(\rz)$};

\node[right=0.4cm of ru, align=left] (rulbl)
    {\footnotesize $\rbeta_{r(u)} + \rx_i \cdot \one\{r(u) \in \mathcal{I}(i)\}$};

\node[below=0.5cm of rulbl, align=left, text=gray] (ann)
    {\footnotesize subtract $m^u\!\left(\Delta^{<\ell(u)}, \Delta^{<r(u)}\right)$ \\
     \footnotesize $\Rightarrow$ residual noise $\sigma_u$};
\draw[ann] (rulbl.south) -- (ann.north);

\end{tikzpicture}
}
\caption{Adversary at node $u \in \T$ (which may or may not be in $\Vecd(\T)$): given the noise $\Delta^{<\ell(u)}(\rz)$ and
$\Delta^{<r(u)}(\rz)$ in the two subtrees, the adversary computes
$m^u$ and subtracts it from $\rbeta_{r(u)} + \rx_i \cdot \one\{r(u) \in \mathcal{I}(i)\}$
to obtain a residual $\sigma_u$.}
\label{fig:adversary-tree}
\end{figure}

Unfortunately the adversary cannot compute $\Delta^{<\ell(u)}(\rz), \Delta^{<r(u)}(\rz)$ from $\rx_i A e_i + \rz$. The obstacle is that for any node $v \in \I(i)$ in the subtrees below $\ell(u)$ or $r(u)$, the adversary observes $\rbeta_v+\rx_i$ rather than $\rbeta_v$ alone. Since $\rx_i$ is unknown, the adversary cannot extract the noise $\rbeta_v$. To remedy this, we take a step back and slightly redefine $\sigma_{u}$.

\paragraph{Noise Prediction via Grandchildren.}
Consider again a non-leaf node $u$ and assume the children of $u$ are also internal nodes. Instead of arguing that $\mx^{<u}(\rz)-\mn^{<u}(\rz)$ grows compared to $(\mx^{<\ell(u)}(\rz)-\mn^{<\ell(u)}(\rz))/2 + (\mx^{<r(u)}(\rz)-\mn^{<r(u)}(\rz))/2$, we can also relate it to the grandchildren.

In more detail, consider first the two grandchildren $\ell(\ell(u))$ and $\ell(r(u))$ who are both left children of their parent. Since $\Anc(i)$ for $i \in \cL(\T^{\ell(\ell(u))})$ contains none of $r(u), r(\ell(u))$ and $r(r(u))$ (the only nodes among $u$, the children of $u$, and the grandchildren, who have a $\chi$-vector), we get $\rz^{<u}_i = \rz^{<\ell(\ell(u))}_i$. Likewise, $\Anc(i)$ for $i \in \cL(\T^{\ell(r(u))})$ contains precisely $r(u)$ and hence, similarly to~\eqref{eq:recurse}, we get
\begin{align}
&\mx^{<u}(\rz)-\mn^{<u}(\rz) \geq \nonumber\\
&\frac{\mx^{<\ell(\ell(u))}(\rz)-\mn^{<\ell(\ell(u))}(\rz)}{2}+\frac{\mx^{<\ell(r(u))}(\rz)-\mn^{<\ell(r(u))}(\rz)}{2} +\nonumber\\
&\frac{\left|\mx^{<\ell(\ell(u))}(\rz)-\mx^{<\ell(r(u))}(\rz)-\rbeta_{r(u)} \right|}{2}+\frac{\left|\mn^{<\ell(\ell(u))}(\rz)-\mn^{<\ell(r(u))}(\rz) - \rbeta_{r(u)} \right|}{2}.\label{eq:grand1}
\end{align}
We can thus define $\sigma_{u,\ell}$ to measure the residual noise in $\rbeta_{r(u)}$ conditioned on the grandchildren $\ell(\ell(u))$ and $\ell(r(u))$ (see \Cref{fig:adversary-ell}):
\[
\sigma_{u,\ell} := \E[\inf_m \E[|\rbeta_{r(u)}-m| \mid \Delta^{<\ell(\ell(u))}(\rz), \Delta^{<\ell(r(u))}(\rz)]].
\]
The symmetric case where we consider the grandchildren $r(\ell(u))$ and $r(r(u))$ is slightly more complicated as $\Anc(i)$ for $i \in \cL(\T^{r(\ell(u))})$ contains $r(\ell(u))$ and for $i \in \cL(\T^{r(r(u))})$ contains both $r(u)$ and $r(r(u))$. This gives 
\begin{align}
&\mx^{<u}(\rz)-\mn^{<u}(\rz) \geq \nonumber\\
&\frac{\mx^{<r(\ell(u))}(\rz)-\mn^{<r(\ell(u))}(\rz)}{2}+\frac{\mx^{<r(r(u))}(\rz)-\mn^{<r(r(u))}(\rz)}{2} + \nonumber\\
&\frac{\left|\mx^{<r(\ell(u))}(\rz)-\mx^{<r(r(u))}(\rz)+\rbeta_{r(\ell(u))}-\rbeta_{r(u)}-\rbeta_{r(r(u))} \right|}{2}+\nonumber \\
&\frac{\left|\mn^{<r(\ell(u))}(\rz)-\mn^{<r(r(u))}(\rz)+\rbeta_{r(\ell(u))}-\rbeta_{r(u)}-\rbeta_{r(r(u))}\right|}{2}. \label{eq:grand2}
\end{align}
We handle this by measuring the residual noise in the sum $\rbeta_{r(\ell(u))}-\rbeta_{r(u)}-\rbeta_{r(r(u))}$ instead (see \Cref{fig:adversary-r}):
\[
\sigma_{u,r} := \E[\inf_m \E[|\rbeta_{r(\ell(u))}-\rbeta_{r(u)}-\rbeta_{r(r(u))}-m| \mid \Delta^{<r(\ell(u))}(\rz), \Delta^{<r(r(u))}(\rz)]].
\]

\begin{figure}[ht]
\centering
\scalebox{0.75}{
\begin{tikzpicture}[
    every node/.style={font=\small},
    filled/.style={circle, draw, fill=black!80, text=white,
                   minimum size=2.4em, inner sep=0pt},
    empty/.style={circle, draw, fill=white,
                  minimum size=2.4em, inner sep=0pt},
    subtree/.style={isosceles triangle, draw=black!60,
                    fill=gray!10,
                    shape border rotate=90, minimum height=1.8cm,
                    minimum width=2.2cm, inner sep=0pt,
                    isosceles triangle apex angle=50},
    subtree-unk/.style={isosceles triangle, draw=black!60, dashed,
                    fill=white,
                    shape border rotate=90, minimum height=1.8cm,
                    minimum width=2.2cm, inner sep=0pt,
                    isosceles triangle apex angle=50},
    edge/.style={-},
    ann/.style={->, >=Stealth, black!50, thin},
]

\node[empty] (u) at (0, 0) {$u$};

\node[empty,  below left=1.2cm and 3.5cm of u]  (lu) {$\ell(u)$};
\node[filled, below right=1.2cm and 3.5cm of u] (ru) {$r(u)$};
\draw[edge] (u) -- (lu);
\draw[edge] (u) -- (ru);

\node[empty,  below left=1.2cm and 1.5cm of lu]  (llu) {$\ell(\ell(u))$};
\node[filled, below right=1.2cm and 1.5cm of lu] (rlu) {$r(\ell(u))$};
\node[empty,  below left=1.2cm and 1.5cm of ru]  (lru) {$\ell(r(u))$};
\node[filled, below right=1.2cm and 1.5cm of ru] (rru) {$r(r(u))$};
\draw[edge] (lu) -- (llu);
\draw[edge] (lu) -- (rlu);
\draw[edge] (ru) -- (lru);
\draw[edge] (ru) -- (rru);

\node[subtree,     below=0.4cm of llu] (Tllu) {};
\node[subtree-unk, below=0.4cm of rlu] (Trlu) {};
\node[subtree,     below=0.4cm of lru] (Tlru) {};
\node[subtree-unk, below=0.4cm of rru] (Trru) {};

\draw[edge] (llu.south) -- (Tllu.north);
\draw[edge] (rlu.south) -- (Trlu.north);
\draw[edge] (lru.south) -- (Tlru.north);
\draw[edge] (rru.south) -- (Trru.north);

\node[above=0.25cm of Tllu.south, align=center]
    {\footnotesize $\Delta^{<\ell(\ell(u))}(\rz)$};
\node[above=0.25cm of Trlu.south, align=center]
    {\footnotesize $\Delta^{<r(\ell(u))}(\rz)$};
\node[above=0.25cm of Tlru.south, align=center]
    {\footnotesize $\Delta^{<\ell(r(u))}(\rz)$};
\node[above=0.25cm of Trru.south, align=center]
    {\footnotesize $\Delta^{<r(r(u))}(\rz)$};

\node[below=0.05cm of Tllu, font=\scriptsize] {known};
\node[below=0.05cm of Trlu, font=\scriptsize] {unknown};
\node[below=0.05cm of Tlru, font=\scriptsize] {known};
\node[below=0.05cm of Trru, font=\scriptsize] {unknown};

\node[right=0.5cm of ru, align=left] (ann)
    {\footnotesize predict $\rbeta_{r(u)}$\\
     \footnotesize with residual noise $\sigma_{u,\ell}$};
\draw[ann] (ru.east) to[out=0, in=180] (ann.west);

\end{tikzpicture}
}
\caption{Case $\ell$ at a node $u$ of even depth: given the noise $\Delta^{<\ell(\ell(u))}(\rz)$ and $\Delta^{<\ell(r(u))}(\rz)$ in the two left grandchildren's subtrees (solid), the adversary predicts $\rbeta_{r(u)}$ with residual noise $\sigma_{u,\ell}$. The right grandchildren's subtrees (dashed) need not be known.}
\label{fig:adversary-ell}
\end{figure}
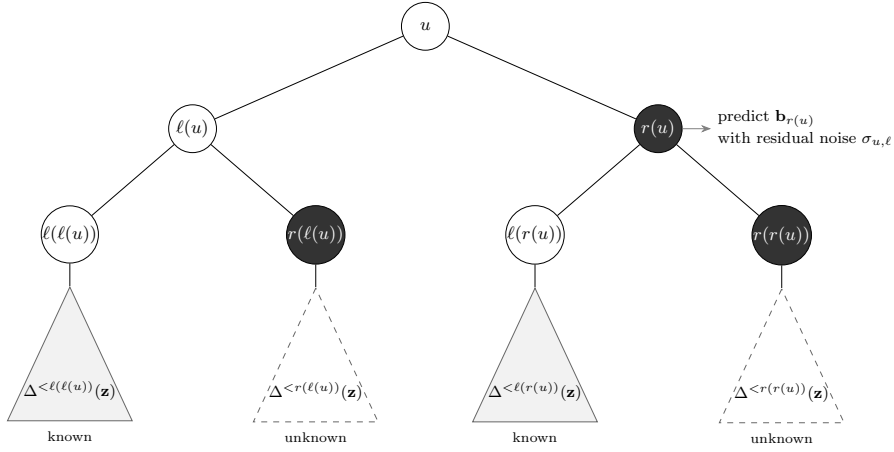

\begin{figure}[ht]
\centering
\scalebox{0.75}{
\begin{tikzpicture}[
    every node/.style={font=\small},
    filled/.style={circle, draw, fill=black!80, text=white,
                   minimum size=2.4em, inner sep=0pt},
    empty/.style={circle, draw, fill=white,
                  minimum size=2.4em, inner sep=0pt},
    subtree/.style={isosceles triangle, draw=black!60,
                    fill=gray!10,
                    shape border rotate=90, minimum height=1.8cm,
                    minimum width=2.2cm, inner sep=0pt,
                    isosceles triangle apex angle=50},
    subtree-unk/.style={isosceles triangle, draw=black!60, dashed,
                    fill=white,
                    shape border rotate=90, minimum height=1.8cm,
                    minimum width=2.2cm, inner sep=0pt,
                    isosceles triangle apex angle=50},
    edge/.style={-},
    ann/.style={->, >=Stealth, black!50, thin},
]

\node[empty] (u) at (0, 0) {$u$};

\node[empty,  below left=1.2cm and 3.5cm of u]  (lu) {$\ell(u)$};
\node[filled, below right=1.2cm and 3.5cm of u] (ru) {$r(u)$};
\draw[edge] (u) -- (lu);
\draw[edge] (u) -- (ru);

\node[empty,  below left=1.2cm and 1.5cm of lu]  (llu) {$\ell(\ell(u))$};
\node[filled, below right=1.2cm and 1.5cm of lu] (rlu) {$r(\ell(u))$};
\node[empty,  below left=1.2cm and 1.5cm of ru]  (lru) {$\ell(r(u))$};
\node[filled, below right=1.2cm and 1.5cm of ru] (rru) {$r(r(u))$};
\draw[edge] (lu) -- (llu);
\draw[edge] (lu) -- (rlu);
\draw[edge] (ru) -- (lru);
\draw[edge] (ru) -- (rru);

\node[subtree-unk, below=0.4cm of llu] (Tllu) {};
\node[subtree,     below=0.4cm of rlu] (Trlu) {};
\node[subtree-unk, below=0.4cm of lru] (Tlru) {};
\node[subtree,     below=0.4cm of rru] (Trru) {};

\draw[edge] (llu.south) -- (Tllu.north);
\draw[edge] (rlu.south) -- (Trlu.north);
\draw[edge] (lru.south) -- (Tlru.north);
\draw[edge] (rru.south) -- (Trru.north);

\node[above=0.25cm of Tllu.south, align=center]
    {\footnotesize $\Delta^{<\ell(\ell(u))}(\rz)$};
\node[above=0.25cm of Trlu.south, align=center]
    {\footnotesize $\Delta^{<r(\ell(u))}(\rz)$};
\node[above=0.25cm of Tlru.south, align=center]
    {\footnotesize $\Delta^{<\ell(r(u))}(\rz)$};
\node[above=0.25cm of Trru.south, align=center]
    {\footnotesize $\Delta^{<r(r(u))}(\rz)$};

\node[below=0.05cm of Tllu, font=\scriptsize] {unknown};
\node[below=0.05cm of Trlu, font=\scriptsize] {known};
\node[below=0.05cm of Tlru, font=\scriptsize] {unknown};
\node[below=0.05cm of Trru, font=\scriptsize] {known};

\node[right=0.5cm of ru, align=left] (ann)
    {\footnotesize predict $\rbeta_{r(\ell(u))}-\rbeta_{r(u)}-\rbeta_{r(r(u))}$\\
     \footnotesize with residual noise $\sigma_{u,r}$};
\draw[ann] (ru.east) to[out=0, in=180] (ann.west);

\end{tikzpicture}
}
\caption{Case $r$ at a node $u$ of even depth: given the noise $\Delta^{<r(\ell(u))}(\rz)$ and $\Delta^{<r(r(u))}(\rz)$ in the two right grandchildren's subtrees (solid), the adversary predicts $\rbeta_{r(\ell(u))}-\rbeta_{r(u)}-\rbeta_{r(r(u))}$ with residual noise $\sigma_{u,r}$. The left grandchildren's subtrees (dashed) need not be known.}
\label{fig:adversary-r}
\end{figure}
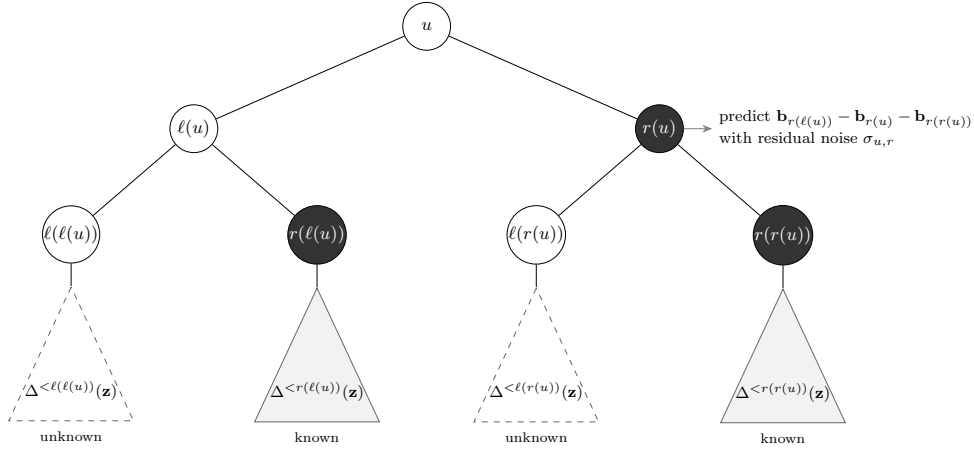

Averaging~\eqref{eq:grand1} and~\eqref{eq:grand2}, we conclude
\begin{align*}
&\E[\mx^{<u}(\rz)-\mn^{<u}(\rz)] \geq \sigma_{u,\ell}/2 + \sigma_{u,r}/2 + \sum_{f,g \in \{\ell,r\}} \frac{\E\left[ \mx^{<f(g(u))}(\rz)-\mn^{<f(g(u))}(\rz)\right]}{4}.
\end{align*}
Starting from the root, which has depth $0$, we can recursively unfold this sum and, similarly to~\eqref{eq:unfolded}, we conclude
\begin{align}
    \E[\|\rz\|_\infty] = \Omega\left(\sum_{d=0}^{h/2-1} \sum_{u \in \T : d(u)=2d} 2^{-2d} (\sigma_{u,\ell} + \sigma_{u,r}) \right). \label{eq:evensum}
\end{align}
Note that the sum is again the expectation over a random root-to-leaf path, of the sum of $\sigma_{u,\ell}$ and $\sigma_{u,r}$ for nodes $u$ at even levels of $\T$.

Now, let us look carefully at what we have gained from predicting the noise from the grandchildren. The key observation is that for an adversary trying to distinguish neighboring $x',x''$ with $x'_i \neq x''_i$, it holds for any node $u$ that if $i \in \cL(\T^{r(\ell(u))}) \cup \cL(\T^{r(r(u))})$, then $v \notin \I(i)$ for all $v \in \T^{\ell(\ell(u))} \setminus \{\ell(\ell(u))\}$ and all $v \in \T^{\ell(r(u))} \setminus \{\ell(r(u))\}$. The symmetric thing holds for $i \in \cL(\T^{\ell(\ell(u))}) \cup \cL(\T^{\ell(r(u))})$. In a nutshell, if an adversary ``attacks'' a variable $\rx_i$ with $i$ in the subtree of one of the two grandchildren $r(\ell(u))$ and $r(r(u))$, then the adversary can directly compute $\Delta^{<\ell(\ell(u))}(\rz), \Delta^{<\ell(r(u))}(\rz)$ and predict $\rbeta_{r(u)}$ up to an expected absolute error of $\sigma_{u,\ell}$. A symmetric argument allows the adversary to predict $\rbeta_{r(\ell(u))}-\rbeta_{r(u)}-\rbeta_{r(r(u))}$  up to expected error $\sigma_{u,r}$ if $i$ is in the subtree of either $\ell(\ell(u))$ or $\ell(r(u))$. Let us ignore this slightly more tedious case in the proof overview.

Unfortunately, a new obstacle arises from this approach. One central idea in the linear measurements approach for bounding the noise, is to define orthogonal unit length vectors $\psi^i$ so that $\sum_i \E[\langle \psi^i, \rz \rangle^2] \leq \E[\|\rz\|_2^2]$. Now, if we follow the outlined approach and use the noise in the pair of grandchildren not containing the $i$'th leaf, then the distribution of the residual noise vector $\raq$ depends on $i$. That is, if $i$ falls in one of the subtrees rooted at the grandchildren $\ell(\ell(u)), \ell(r(u))$, we get one distribution, and if it falls in the subtrees rooted at $r(\ell(u)), r(r(u))$, we get another. If we have different residual noise vectors $\raq^i$ for different $i$, then the argument $\sum_i \E[\langle \psi^i, \raq^i \rangle^2] \leq \E[\|\raq\|_2^2]$ breaks, since $\raq$ is not a fixed vector as it depends on $i$.

\paragraph{Restricting to a Subtree.}
Our next idea is to restrict attention to a collection of leaf indices $i$ that all result in the same noise predictions across the tree $\T$. In more detail, consider a string $s \in \{\ell,r\}^{h/2}$ and let $\T^s$ be the subset of $\T$ defined as follows. Starting from the root, from a node at an even depth $2d$ we descend into both children, and from a node at an odd depth $2d-1$ we descend only to the left child if $s_d=\ell$ and only to the right child if $s_d=r$. Let $\cL(\T^s)$ be the indices of the leaves of $\T^s$.

The critical observation is that if $u \in \T^s$ is a node of even depth $2d$, then every $i \in \cL(\T^s) \cap \cL(\T^u)$ falls in one of the two subtrees $\T^{s_{d+1}(\ell(u))}$ and  $\T^{s_{d+1}(r(u))}$. This implies that an adversary trying to infer $\rx_i$ via a linear measurement can compute $\Delta^{<\widetilde{s_{d+1}}(\ell(u))}(\rz), \Delta^{<\widetilde{s_{d+1}}(r(u))}(\rz)$ regardless of which $i \in \cL(\T^s)$ the adversary attacks, where $\widetilde{\ell} = r$ and $\widetilde{r} = \ell$ denotes the opposite direction. The adversary can thus predict $\rbeta_{r(u)}$ up to expected absolute error $\sigma_{u,\ell}$ if $s_{d+1}=r$, and predict $\rbeta_{r(\ell(u))}-\rbeta_{r(u)}-\rbeta_{r(r(u))}$ up to error $\sigma_{u,r}$ if $s_{d+1}=\ell$. 

Assume for simplicity that $s = r r \cdots r$ is the all-r string, so we do not have to deal with the $\rbeta_{r(\ell(u))}-\rbeta_{r(u)}-\rbeta_{r(r(u))}$ case. For a neighboring pair $x',x''$ with $x'_i \neq x''_i$ and $\rx$ uniform in $\{x',x''\}$, an adversary can now compute 
\[
\rbeta_{r(u)}+ \rx_i \one \{r(u) \in \I(i)\} - m^u_\ell(\Delta^{<\ell(\ell(u))}(\rz), \Delta^{<\ell(r(u))}(\rz))= \raq_u + \rx_i \one \{r(u) \in \I(i)\}
\]
for every $u \in \T^s$ of even depth, where $\E[|\raq_u|] = \sigma_{u,\ell}$. Here $m^u_\ell$ is a conditional median of $\rbeta_{r(u)}$.

Letting $\nu^i$ be the vector with entries $\one\{r(u) \in \I(i)\}$ for every $u \in \T^s$ of even depth, the adversary has thus obtained a vector $\rx_i \nu^i + \raq$ with $\E[|\raq_u|] = \sigma_{u,\ell}$. Moreover, if we draw $\rs$ uniformly in $\{\ell,r\}^{h/2}$, then a uniform random root-to-leaf path in $\T^{\rs}$, is uniform random in $\T$. Thus if we can lower bound the expected sum of $\sigma_{u,\ell} + \sigma_{u,r}$ along a uniform random root-to-leaf path in \emph{every} $\T^s$, then we get the same lower bound on~\eqref{eq:evensum}. We thus fix an arbitrary $s \in \{\ell,r\}^{h/2}$ and lower bound
\begin{align}
\label{eq:ts}
\sum_{d=0}^{h/2-1} \sum_{u \in \T^s : d(u)=2d} 2^{-d} \sigma_{u,\widetilde{s_{d+1}}}.
\end{align}
For simplicity of the proof overview, let us merely consider $s=rr \cdots r$ and $\sum_{d=0}^{h/2-1} \sum_{u \in \T^s : d(u)=2d} 2^{-d} \sigma_{u,\ell}$. See \Cref{fig:Ts} for an example for $n=16$ and $s=rr$.

\begin{figure}[ht]
\centering
\scalebox{0.75}{
\begin{tikzpicture}[
    every node/.style={font=\small},
    highlighted-V/.style={circle, draw=black, fill=black!80, text=white,
                   minimum size=1.6em, inner sep=2pt},
    highlighted/.style={circle, draw=black, fill=white,
                   minimum size=1.6em, inner sep=2pt},
    normal/.style={circle, draw=black!40, dashed, fill=white,
                  minimum size=1.6em, inner sep=2pt},
    edge/.style={-, black!40, dashed},
    hedge/.style={-, thick},
    lbl/.style={font=\scriptsize, align=center},
    depth/.style={font=\scriptsize, text=gray},
]

\node[depth] at (-7.5, 0)    {$d=0$};
\node[depth] at (-7.5, -1.4) {$d=1$};
\node[depth] at (-7.5, -2.8) {$d=2$};
\node[depth] at (-7.5, -4.2) {$d=3$};
\node[depth] at (-7.5, -5.6) {$d=4$};
\draw[dashed, gray!40, thin] (-7.1, 0.5) -- (-7.1, -6.0);

\node[highlighted-V] (r) at (0, 0) {};

\node[highlighted]   (l0) at (-3, -1.4) {};
\node[highlighted-V] (l1) at ( 3, -1.4) {};

\node[normal]        (l00) at (-4.5, -2.8) {};
\node[highlighted-V] (l01) at (-1.5, -2.8) {};
\node[normal]        (l10) at ( 1.5, -2.8) {};
\node[highlighted-V] (l11) at ( 4.5, -2.8) {};

\node[normal]        (l000) at (-5.25, -4.2) {};
\node[normal]        (l001) at (-3.75, -4.2) {};
\node[highlighted]   (l010) at (-2.25, -4.2) {};
\node[highlighted-V] (l011) at (-0.75, -4.2) {};
\node[normal]        (l100) at ( 0.75, -4.2) {};
\node[normal]        (l101) at ( 2.25, -4.2) {};
\node[highlighted]   (l110) at ( 3.75, -4.2) {};
\node[highlighted-V] (l111) at ( 5.25, -4.2) {};

\node[normal]        (l0000) at (-5.625, -5.6) {};
\node[normal]        (l0001) at (-4.875, -5.6) {};
\node[normal]        (l0010) at (-4.125, -5.6) {};
\node[normal]        (l0011) at (-3.375, -5.6) {};
\node[normal]        (l0100) at (-2.625, -5.6) {};
\node[highlighted-V] (l0101) at (-1.875, -5.6) {};
\node[normal]        (l0110) at (-1.125, -5.6) {};
\node[highlighted-V] (l0111) at (-0.375, -5.6) {};
\node[normal]        (l1000) at ( 0.375, -5.6) {};
\node[normal]        (l1001) at ( 1.125, -5.6) {};
\node[normal]        (l1010) at ( 1.875, -5.6) {};
\node[normal]        (l1011) at ( 2.625, -5.6) {};
\node[normal]        (l1100) at ( 3.375, -5.6) {};
\node[highlighted-V] (l1101) at ( 4.125, -5.6) {};
\node[normal]        (l1110) at ( 4.875, -5.6) {};
\node[highlighted-V] (l1111) at ( 5.625, -5.6) {};

\draw[hedge] (r)    -- (l0);
\draw[hedge] (r)    -- (l1);
\draw[edge]  (l0)   -- (l00);
\draw[hedge] (l0)   -- (l01);
\draw[edge]  (l1)   -- (l10);
\draw[hedge] (l1)   -- (l11);
\draw[edge]  (l00)  -- (l000);
\draw[edge]  (l00)  -- (l001);
\draw[hedge] (l01)  -- (l010);
\draw[hedge] (l01)  -- (l011);
\draw[edge]  (l10)  -- (l100);
\draw[edge]  (l10)  -- (l101);
\draw[hedge] (l11)  -- (l110);
\draw[hedge] (l11)  -- (l111);
\draw[edge]  (l000) -- (l0000);
\draw[edge]  (l000) -- (l0001);
\draw[edge]  (l001) -- (l0010);
\draw[edge]  (l001) -- (l0011);
\draw[edge]  (l010) -- (l0100);
\draw[hedge] (l010) -- (l0101);
\draw[edge]  (l011) -- (l0110);
\draw[hedge] (l011) -- (l0111);
\draw[edge]  (l100) -- (l1000);
\draw[edge]  (l100) -- (l1001);
\draw[edge]  (l101) -- (l1010);
\draw[edge]  (l101) -- (l1011);
\draw[edge]  (l110) -- (l1100);
\draw[hedge] (l110) -- (l1101);
\draw[edge]  (l111) -- (l1110);
\draw[hedge] (l111) -- (l1111);

\draw[gray!40, thin] (-5.9, -6.2) -- (5.9, -6.2);
\foreach \x/\i in {
    -5.625/1, -4.875/2, -4.125/3, -3.375/4,
    -2.625/5, -1.875/6, -1.125/7, -0.375/8,
     0.375/9,  1.125/10, 1.875/11, 2.625/12,
     3.375/13, 4.125/14, 4.875/15, 5.625/16}{
    \node[font=\scriptsize, text=gray] at (\x, -6.5) {$\i$};
}
\node[font=\scriptsize, text=gray] at (0, -6.8) {leaf index $i$};

\node[highlighted-V, minimum size=1.0em] (leg1) at (-3.5, -7.5) {};
\node[font=\footnotesize, right=0.15cm of leg1, text=black]
    {node in $\mathcal{T}^s \cap \mathcal{V}(\mathcal{T})$};
\node[highlighted, minimum size=1.0em] (leg2) at (-3.5, -8.1) {};
\node[font=\footnotesize, right=0.15cm of leg2, text=black]
    {node in $\mathcal{T}^s \setminus \mathcal{V}(\mathcal{T})$};
\node[normal, minimum size=1.0em] (leg3) at (-3.5, -8.7) {};
\node[font=\footnotesize, right=0.15cm of leg3, text=black]
    {node not in $\mathcal{T}^s$};

\end{tikzpicture}
}
\caption{The subtree $\mathcal{T}^s$ for $s=rr$ in a binary tree of height $h=4$.
Nodes in $\mathcal{T}^s \cap \mathcal{V}(\mathcal{T})$ are filled black, nodes in
$\mathcal{T}^s \setminus \mathcal{V}(\mathcal{T})$ have a solid white border, and
nodes not in $\mathcal{T}^s$ have a dashed border. For each node $u \in \T^s$ at even depth ($d=0$ and $d=2$), both children of $u$ are included in $\mathcal{T}^s$. For each node $u$ at odd depth ($d=1$ and $d=3$), only the right child of $u$ is included. The highlighted leaves give 
$\mathcal{L}(\mathcal{T}^s) = \{6, 8, 14, 16\}$.}
\label{fig:Ts}
\end{figure}
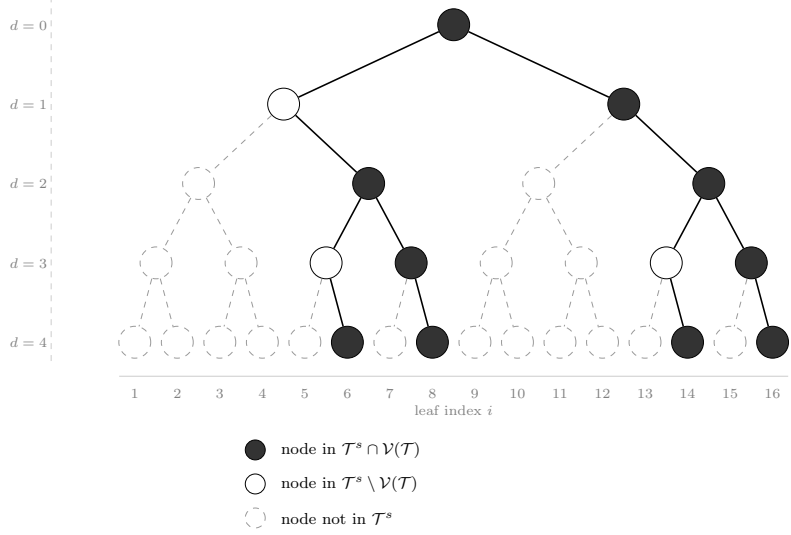

\paragraph{Median-of-Three Trick.}
We now have an adversary for every $i \in \cL(\T^s)$ that can compute $\rx_i \nu^i + \raq$. Returning to the strategy of using linear measurements, the natural next step would be to define orthogonal unit length vectors $\psi^i$ for which $\sum_i \langle \psi^i, \nu^i\rangle^2$ is large. By orthogonality, we would simultaneously get $\sum_i \E[\langle \psi^i, \raq \rangle^2] \leq \E[\|\raq\|_2^2] = \sum_{u \in \T^s: d(u) \textrm{ even}} \E[\raq_u^2]$. Unfortunately, our guarantee $\E[|\raq_u|] = \sigma_{u,\ell}$ is not sufficient to give any upper bound on $\E[ \raq_u^2]$ and thereby $\E[\|\raq\|_2^2]$. To overcome this, we use a trick from Larsen, Pagh, Tetek~\cite{mediantrick}. An immediate corollary of their work is the following.
\begin{corollary}[\cite{mediantrick}]
\label{cor:median}
Let $\rZ_1,\rZ_2,\rZ_3$ be i.i.d.\ real-valued random variables. Then 
\[
\E[\median\{\rZ_1,\rZ_2,\rZ_3\}^2] \leq 3 \cdot \E[|\rZ_1|]^2.
\]
\end{corollary}
Now using the composition property of differential privacy, we have that the mechanism $\rM'$ releasing $(Ax+\rz^{(1)}, Ax+\rz^{(2)}, Ax+\rz^{(3)})$ for three i.i.d.\ draws of the noise $\rz$ of the mechanism $\rM$ is $(3\eps, 3\delta)$-DP. Thus even if we give an adversary access to three independent copies of the noise, privacy still guarantees that one cannot distinguish $x', x''$. Crucially, the adversary trying to infer $\rx_i$ can now compute, for every $u \in \T^s$ of even depth, the value
\begin{align*}
&\median\{\rx_i \one\{r(u) \in \I(i)\} + \raq_u^{(1)},\rx_i \one\{r(u) \in \I(i)\} + \raq_u^{(2)},\rx_i \one\{r(u) \in \I(i)\} + \raq_u^{(3)}\} =\\
& \rx_i \one\{r(u) \in \I(i)\} + \median\{ \raq_u^{(1)}, \raq_u^{(2)}, \raq_u^{(3)}\}.
\end{align*}
From Corollary~\ref{cor:median}, this results in a vector $\rx_i \nu^i + \raq$ where 
\[
\E[(\raq_u)^2] = \E[\median\{ \raq_u^{(1)}, \raq_u^{(2)}, \raq_u^{(3)}\}^2] \leq 3 \E[|\raq_u^{(1)}|]^2 = 3 \sigma_{u,\ell}^2.\]
We thus have 
\begin{align}
\E[\|\raq\|_2^2] \leq 3 \sum_{d=0}^{h/2-1} \sum_{u \in \T^s : d(u)=2d} \sigma^2_{u,\ell}. \label{eq:sumsquare}
\end{align}
As suggested earlier, the next step is to define the linear measurements $\psi^i$ and then use privacy to argue that $\E[\|\raq\|_2^2] = \Omega(\sum_i \langle \psi^i, \nu^i\rangle^2)$. Indeed this would lower bound the sum of $\sigma^2_{u,\ell}$. Sadly, this is not enough as we need to bound the weighted sum of $\sigma_{u,\ell}$ given in~\eqref{eq:ts}. We thus have a mismatch, both in terms of weights, but more critically,~\eqref{eq:sumsquare} is a sum of squares whereas~\eqref{eq:ts} is a sum.

Let us mention that our full proof handles the above in a slightly different manner. Rather than attacking different $\rx_i$ variables directly, we separate the noise reduction step from the linear measurement attacks, making the proof more modular. In essence, we define a derived mechanism $\widehat{\rM}_s$ that is $(3\eps, 3\delta)$-DP when restricted to inputs with $x_i = 0$ for $i \notin \cL(\T^s)$, and apply the linear measurements argument to $\widehat{\rM}_s$. Since $|\cL(\T^s)| = \sqrt{n}$ and the lower bound is polylogarithmic, the restriction to a smaller input domain is harmless.

\paragraph{From Weighted Sum to Max Bound.}
We now address the mismatch between the sum of squares in~\eqref{eq:sumsquare} and the weighted sum~\eqref{eq:ts} that we need to lower bound. We will show that for at least half of the $h/2$ choices of depth $2d$, at least half of the nodes $u \in \T^s$ must have $\sigma_{u,\ell} = \Omega(\sqrt{\log n})$. This still lower bounds~\eqref{eq:ts} by the desired $\Omega(\log^{3/2} n)$.

The way we approach it is as follows: for every even depth $2d$, consider the median
\[
m_d := \median_{u \in \T^s : d(u)=2d} \{\sigma_{u,\ell}\}.
\]
Now let $\cJ \subseteq \{0,\dots,h/2-1\}$ with $|\cJ|=h/4$ be so that the medians $m_d$ with $d \in \cJ$ are the $h/4$ smallest of the $h/2$ medians. Furthermore, for every $d \in \cJ$, let $\cJ_d$ consist of the nodes $u \in \T^s$ of depth $2d$ having the smallest $\sigma_{u,\ell}$ values. In particular, each such $\sigma_{u,\ell}$ satisfies $\sigma_{u,\ell} \leq m_d$. Let us now define 
\[
\sigma := \max_{d \in \cJ} \max_{u \in \cJ_d} \sigma_{u,\ell}.
\]
Then by definition of $\cJ$, for each of the $h/4$ values of $d$ not in $\cJ$, we have $m_d \geq \sigma$. Hence, at least half the nodes of depth $2d$ have $\sigma_{u,\ell} \geq \sigma$. Summing across the $h/4$ depths not in $\cJ$, we lower bound~\eqref{eq:ts} by
\begin{align}
\label{eq:fullsum}
    &\sum_{d=0}^{h/2-1} \sum_{u \in \T^s : d(u)=2d} 2^{-d} \sigma_{u,\ell} \geq 
    \sum_{d \notin \cJ} \sum_{u \in \T^s : d(u)=2d \wedge \sigma_{u,\ell} \geq m_d} 2^{-d} m_d \geq 
    \sum_{d \notin \cJ} 2^{d-1} 2^{-d} \sigma = 
    \sum_{d \notin \cJ} \sigma/2 =
    h\sigma /8.
\end{align}
With this in mind, we now set out to prove that for \emph{any} choice of $\cJ \subseteq \{0,\dots,h/2-1\}$ with $|\cJ|=h/4$ and \emph{any} choice of $\{\cJ_d\}_{d \in \cJ}$ with $\cJ_d$ containing half the nodes in $\T^s$ of depth $2d$, there is \emph{at least one} node $u \in \cup_{d \in \cJ} \cJ_d$ with $\sigma_{u,\ell} = \Omega(\sqrt{\log n})$.

What we have achieved is that we now need to handle the simpler task of showing that just one residual noise must be large, rather than showing that an intricate weighted sum is large. Furthermore, showing that just one $\sigma_{u,\ell}$ is large is also possible when given a lower bound on the sum of squared values, as in~\eqref{eq:sumsquare}, making the attacks via linear measurements amenable.

\paragraph{Privacy Attacks via Linear Measurements.} We now describe the second part of the proof: constructing the adversary using linear measurements. Having analyzed the noise structure in the tree, we have shown that the noise of mechanism $\rM$ has residual contributions $\sigma_{u,\ell}$ along the tree. It remains to show that if these residual contributions are small, then privacy is violated. Recall that for simplicity we consider $s=rr\ldots r$ and an adversary trying to distinguish neighboring inputs $x', x''$ with $x_i'\neq x_i''$ for some $i \in \cL(\T^s)$. As shown above, this adversary can compute a vector $\rx_i \nu^i + \raq$ with an entry for each non-leaf $u \in \T^s$ of even depth so that $\E[\raq_u^2] \leq 3 \sigma_{u,\ell}^2$ and $\nu^i_u = \one\{r(u) \in \I(i) \}$. We now show that if the $\sigma_{u,\ell}$ are small, the adversary can use linear measurements to distinguish $x'$ from $x''$, contradicting privacy.

Recall from above the arbitrary sets $\cJ \subseteq \{0,\dots,h/2-1\}$ with $|\cJ|=h/4$ and $\{\cJ_d\}_{d \in \cJ}$ with $\cJ_d$ containing half the nodes in $\T^s$ of depth $2d$. We need to define $\psi^i$ (specified shortly) with the goal of showing that at least one $\sigma_{u,\ell}$ for $u \in \cup_{d \in \cJ} \cJ_d$ is large. Since post-processing preserves privacy, the adversary may set all coordinates corresponding to $u \notin \cup_{d \in \cJ} \cJ_d$ to $0$, obtaining a shorter vector. It now suffices to show that $\E[\raq_u^2]$ is large for at least one remaining coordinate $u$, which by $\E[\raq_u^2]\leq 3 \sigma^2_{u,\ell}$ implies that $\sigma_{u,\ell}$ is large. Let $\sigma = \max_{d \in \cJ} \max_{u \in \cJ_d} \sigma_{u,\ell}$. Then by definition, $\E[\|\raq\|_2^2] \leq \sum_{d \in \cJ} \sum_{u \in \cJ_d} 3 \sigma^2$. Since $\T^s$ has $2^d$ nodes at even depth $2d$ and $\cJ_d$ contains half of them, we have $|\cJ_d|=2^{d-1}$, meaning deeper levels contribute more to the sum. To balance this across depths, the adversary scales each coordinate of $\rx_i \nu^i + \raq$ by $2^{-d/2}$ for nodes $u$ at even depth $2d$. Let $\rr$ be the resulting noise vector with $\rr_u = 2^{-d/2} \raq_u$ for $u \in \cup_{d \in \cJ} \cJ_d$ and let $\mu^i$ be the vector with coordinates $\mu^i_u = 2^{-d/2} \nu^i_u$. The adversary has thus computed the vector $\rx_i \mu^i  + \rr$. With $\psi^i$ as the linear measurement for leaf index $i$, the adversary computes the estimate 
\[
\hat \rx_i = \frac{\langle \psi^i, \rx_i\mu^i+\rr\rangle}{\langle \psi^i, \mu^i\rangle} = \rx_i+\frac{\langle \psi^i, \rr \rangle}{\langle \psi^i, \mu^i\rangle}.
\]

Since $\hat \rx_i$ is a private estimate of the binary variable $\rx_i$, applying~\Cref{lem:releaseVar} the error of the adversary for coordinate $i$ is 
\[
\frac{\E[\langle \psi^i, \rr \rangle^2]}{\langle \psi^i, \mu^i\rangle^2} = \Omega(\min\{\eps^{-2}, \delta^{-2}\}).
\]

If we choose orthogonal unit length vectors $\psi^i$ for $i \in \cL(\T^s)$, then $\sum_{i \in \cL(\T^s)} \E[\langle \psi^i, \rr\rangle^2 ] \leq \E[\|\rr\|_2^2]$ and
\begin{align*}
\E[\|\rr\|_2^2] &= \sum_{d \in \cJ} \sum_{u \in \cJ_d} \E[\rr_u^2]= \sum_{d \in \cJ} \sum_{u \in \cJ_d} 2^{-d} \E[\raq_u^2] \\
&\leq \sum_{d \in \cJ} \sum_{u \in \cJ_d} 2^{-d} 3 \sigma^2 = O(h \sigma^2)
\end{align*}
 since $\sum_{d\in \cJ} |\cJ_d|\cdot 2^{-d} \leq |\cJ| = h/4$, including the singleton root if $0\in\cJ$.

 The final step of the proof is now to choose appropriate vectors $\psi^i$ so that $\sum_i \langle \psi^i , \mu^i \rangle^2$ is large. The choice of $\psi^i$ is inspired by the Haar wavelet basis. Given a perfect binary tree over $n$ leaves, the Haar wavelet basis has a vector $h^u$ for each internal node $u \in \T$. If the subtree rooted at $u$ has $n_u$ leaves, then the vector $h^u$ has coordinates $-1/\sqrt{n_u}$ for leaves in the subtree rooted at $\ell(u)$ and it has coordinates $1/\sqrt{n_u}$ for leaves in the subtree rooted at $r(u)$. Finally, the basis also contains the vector with all coordinates $1/\sqrt{n}$. The Haar wavelet basis is an orthonormal basis of $\R^n$. If we now arrange the vectors $h^u$ as columns of a matrix $H$ and compute its transpose $H^t$, then the resulting matrix is also orthonormal. More precisely, we use as $\psi^i$ the columns of $H^t$, where $H$ is the matrix of the Haar wavelet basis restricted to $\T^s$.

 Concretely, we let $\psi^i$ be the vector with one coordinate for every $u \in \T^s$. For a node $u$ of depth $2d$, we let $\psi^i_u = 0$ if $i \notin \cL(\T^u)$, $\psi^i_u = -2^{-(h/2-d)/2}$ if $i \in \cL(\T^{\ell(u)})$ and $\psi^i_u = 2^{-(h/2-d)/2}$ if $i \in \cL(\T^{r(u)})$. Since $2^{h/2-d}$ is the number of leaves in $\cL(\T^s) \cap \cL(\T^u)$, this coincides with the transpose of the matrix $H$ resulting from the Haar wavelet basis on $\T^s$. We expand $\sum_i \langle \psi^i, \mu^i \rangle^2$ and show that the sum is lower bounded by $\Omega(h^2)$. 
 
 Combining the bound $\sum_i \E[\langle \psi^i, \rr\rangle^2] = O(h \sigma^2)$ from above with $\sum_i \langle \psi^i, \mu^i \rangle^2 = \Omega(h^2)$, implies that  $\sigma^2 = \Omega(h \min\{\eps^{-2},\delta^{-2}\})$. Substituting into~\eqref{eq:fullsum} finally gives us
\[
\E[\|\rz\|_\infty] = \Omega\left(\sum_{d=0}^{h/2-1} \sum_{u \in \T^s : d(u)=2d} 2^{-d} \sigma_{u,\ell}\right) = \Omega(h \sigma) = \Omega\left(\frac{h^{3/2}}{\max\{\eps,\delta\}}\right) = \Omega\left(\frac{\log^{3/2} n}{\max\{\eps, \delta\}}\right).
\]

\paragraph{The Pure DP Case.}
The argument above gives a lower bound of $\Omega(\log^{3/2}(n)/\eps)$
also under pure DP. To obtain the stronger $\Omega(\log^2(n)/\eps)$
bound, we return to the last step of the proof. The tree basis, the
prediction of noise from grandchildren, the restriction to $\T^s$, and
the median-of-three trick remain the same. So does the reduction from
the weighted sum to the maximum
\[
  \sigma=\max_{d\in\cJ}\max_{u\in\cJ_d}
      \sigma_{u,\widetilde{s_{d+1}}}.
\]
We now want to show that this maximum is $\Omega(h/\eps)$, improving
the $\Omega(\sqrt h/\eps)$ bound obtained from linear measurements.
The earlier accumulation argument then gives
$\E[\|\rz\|_\infty]=\Omega(h\sigma)=\Omega(h^2/\eps)$.

What additional information does pure DP give us? The linear-measurement
attack used a lower bound on the second moment of the noise in an
estimate of a single input bit. For pure DP, we can also compare inputs
that differ by a larger amount: group privacy bounds the ratio of the
probabilities of any output event by $e^{O(\eps D)}$ when the inputs
differ by $D$ unit changes. We will use this to hide a change of size
$D$ at an unknown leaf and argue that small residual noise would allow
an adversary to locate that leaf too well.

\paragraph{Using Residual Noise to Locate a Change.} The change we hide is an input of the form $D e_i$ for an unkown leaf $i \in \cL(\T^s)$, which we compare against the all-zero input.
Since $De_i$ is not a binary input, it is used only for an auxiliary oblivious
mechanism. As explained in Appendix~\ref{app:oblivious}, the reduction
samples its noise using the original mechanism on Boolean encodings.
Once this noise distribution is fixed, the same additive-noise formula
defines the auxiliary mechanism on integer inputs. Its privacy
inequality for $0,e_i$ transfers to $x,x+e_i$ by translating both output
distributions by the same vector. We can therefore apply group privacy
along $0,e_i,\dots,De_i$, while the noise remains bounded by the error
of the original mechanism on Boolean inputs.

We now explain how small residual noise reveals the location of such a change. As above, let us first consider $s=rr\cdots r$. Recall the derived
mechanism $\widehat{\rM}_s$ obtained after the noise reduction. Its
output $\bq=\widehat{\rM}_s(x)$ has, for a node $u\in\T^s$ at even
depth, the form
\[
  \bq_u=a_{r(u)}+\raq_u
  =\sum_{j\in\cL(\T^s)\cap\cL(\T^{r(\ell(u))})}x_j+\raq_u,
\]
where $\E[\raq_u^2]\leq3\sigma^2$ whenever
$u\in\bigcup_{d\in\cJ}\cJ_d$. This gives a useful interpretation of
the output: it estimates the number of events in one of the two
subtrees in which a path through $\T^s$ can continue. For example, on
the input $De_i$, this count is $D$ if the path to $i$ continues into
$r(\ell(u))$, and it is $0$ if the path continues into $r(r(u))$.
If $D$ is a sufficiently large constant multiple of $\sigma$, comparing
$\bq_u$ to $D/2$ therefore gives a reliable guess for the next step on
the path to $i$.

The case of a general string $s$ has the same interpretation. When
$s_{d+1}=\ell$, the signal is the three-term expression
\[
  a_{r(\ell(u))}-a_{r(u)}-a_{r(r(u))}
  =-\sum_{j\in\cL(\T^s)\cap\cL(\T^{\ell(r(u))})}x_j.
\]
The contributions outside the indicated subtree cancel. Negating
$\bq_u$ therefore again gives an estimate of the count in one of the
two possible subtrees, with the same bound on its residual noise.
For the rest of this overview, we continue with $s=rr\cdots r$.

\paragraph{A Randomized Search Down the Tree.}
Our decoding attack is inspired by the randomized tree decoder of
Ghazi et al.~\cite[Section~5.1]{GhaziEtAl2023Trees}, used to prove
lower bounds for private tree aggregation. We adapt this idea to the
residual counts on $\T^s$, where noise is controlled only at the
selected nodes.

There are two obstacles to turning these local guesses into an attack.
First, we control the residual noise only at nodes in the sets
$\cJ_d$. Second, the noise at different nodes can be arbitrarily
correlated. Even if each local guess is correct with high probability,
we cannot conclude that all guesses along a path are simultaneously
correct with high probability.

To handle the first obstacle, we use the guesses only at selected nodes.
At any other node we simply choose uniformly between the two possible
continuations. A uniformly random leaf has a selected ancestor with
probability at least $1/2$ at each level in $\cJ$. Since
$|\cJ|=h/4$, Markov's inequality shows that at least a quarter of the
leaves have at least $h/16$ selected ancestors. We only use decoding
events for those leaves. Thus, although some of the steps carry no
information, each leaf we use has a linear number of opportunities
for the decoder to improve its guess.

To handle the second obstacle, we randomize the choices even at
selected nodes. Let $D=\max\{1,\lceil16\sigma\rceil\}$. Starting at
the root, the decoder descends through $\T^s$ two levels at a time.
At a selected node $u$, if $\bq_u\geq D/2$, it descends to
$r(\ell(u))$ with probability $3/4$ and to $r(r(u))$ with probability
$1/4$. If $\bq_u<D/2$, it reverses these probabilities. At unselected
nodes, both choices have probability $1/2$. It returns the leaf reached.

The reason for this randomization is that an incorrect local guess
should reduce the probability of finding $i$ by a constant factor,
without ruling out $i$ altogether. We do not need the decoder to find
$i$ with constant probability. There are $2^{h/2}$ possible leaves, so
uniform guessing succeeds with probability only $2^{-h/2}$. It will
suffice to beat this probability by a factor exponential in $h$.

\paragraph{Decoding Despite Correlated Noise.}
Fix one of the leaves $i$ with at least $h/16$ selected ancestors and
consider the input $De_i$. At each selected ancestor, the count is
either $D$ or $0$. Thus the noisy count can favor the wrong direction
only when $|\raq_u|\geq D/2$. By Markov's inequality,
\[
  \Pr[|\raq_u|\geq D/2]
  \leq\frac{4\E[\raq_u^2]}{D^2}
  \leq\frac{12\sigma^2}{D^2}\leq\frac1{16}.
\]

We measure how well the adversary locates $i$ by the logarithm of the ratio between the probability it assigns to the path to $i$ and the probability assigned by uniform guessing. Relative to uniform guessing, a correct local guess multiplies the
probability of following the path to $i$ by $3/2$, whereas an incorrect
guess multiplies it by $1/2$. Taking logarithms makes these multiplicative gains
and losses additive. There are at least $h/16$ selected
ancestors, and the expected fraction of incorrect guesses among them
is at most $1/16$. Therefore, each selected ancestor contributes a positive constant in expectation, so the expected log-ratio is $\Omega(h)$.
Jensen's inequality converts this bound back into a lower bound on
the probability of reaching $i$:
\begin{equation}\label{eq:decode_prob_lower}
  \Pr[\operatorname{Dec}(\widehat{\rM}_s(De_i))=i]
  \geq2^{-h/2}e^{\Omega(h)}.
\end{equation}
This argument only uses the expected number of incorrect local guesses;
it needs no independence between the noise variables. The decoder's
coins are independent after the noisy output has been fixed.

\paragraph{Comparing with the Zero Input.}
For each of the leaves under consideration, pure group privacy gives
\[
  \Pr[\operatorname{Dec}(\widehat{\rM}_s(De_i))=i]
  \leq e^{O(\eps D)}
       \Pr[\operatorname{Dec}(\widehat{\rM}_s(0))=i].
\]
Now sum over these leaves. There are at least $2^{h/2}/4$ of them,
so the sum on the left is at least $e^{\Omega(h)}$ due to \eqref{eq:decode_prob_lower}. On the right, all
probabilities refer to the same zero input and the decoder returns
only one leaf, so they sum to at most $1$. It follows that
\[
  e^{\Omega(h)}\leq e^{O(\eps D)}.
\]
Consequently $D=\Omega(h/\eps)$, and since $D\leq16\sigma+1$, we
obtain the desired $\sigma=\Omega(h/\eps)$. Substituting this into
the earlier bound~\eqref{eq:fullsum} gives
\[
  \E[\|\rz\|_\infty]=\Omega(h\sigma)
  =\Omega\left(\frac{h^2}{\eps}\right)
  =\Omega\left(\frac{\log^2 n}{\eps}\right).
\]

Finally, the argument also allows for the small positive $\delta'$ that
arises in the reduction to oblivious mechanisms. We choose
$\delta'=1/n^2$ before the reduction. The additional term from group
privacy can then be absorbed in the same comparison with the zero
input. We give the details in Section~\ref{sec:pure-dp}.

\section{Formal Lower Bound Proof for Approximate Differential Privacy}
\label{sec:formal-proof}

We restate our main result for convenience.
\begin{restatethm}{\ref{thm:main}}
    For any $n^{-1+\Omega(1)}  < \eps < 1$, $0 < \delta < C$ for a sufficiently small constant $C>0$, it holds that any $(\eps,\delta)$-DP mechanism $\rM$ for continual counting has error
    \[\max_{x \in \{0,1\}^n }\E[\|\rM(x)-Ax\|_\infty] = \Omega\left(\frac{\log^{3/2} n }{\max\{\eps,\delta\}}\right),\]
    where $A$ is the $n \times n$ prefix-sum matrix. 
\end{restatethm}

We now prove~\Cref{thm:main}. Our goal is to lower bound the noise needed to release a private version of $Ax$ for a binary vector $x \in \{0,1\}^n$. In Appendix~\ref{app:oblivious} we prove the following reduction using inspiration from prior works~\cite{ObliviousLinear}.
\begin{lemma}
\label{lem:oblivious}
If there is an $(\eps,\delta)$-DP mechanism $\rM$ (possibly with data dependent noise) for prefix sums on inputs $x \in \{0,1\}^n$ with $\err_\infty(\rM) =\max_{x \in \{0,1\}^n} \E[\| \rM(x) -Ax\|_\infty]$, then for any choice of $0 < \delta' < 1$, there is a $(2\eps,e^\eps\delta+\delta')$-DP mechanism $\rM'$ with \emph{oblivious noise} for prefix sums on inputs $x \in \{0,1\}^m$ with 
\[
m = \left\lfloor \frac{n}{2 \lceil \ln(1/\delta')/\eps \rceil} \right\rfloor
\]
and $\max_{x \in \{0,1\}^m} \E[\|\rM'(x)-Ax\|_\infty] \leq \err_\infty(\rM)$. Here oblivious noise means that $\rM'(x)=Ax + \rz$ for $\rz$ independent of $x$.
\end{lemma}
Note that the prior work~\cite{ObliviousLinear} gives a general reduction for obtaining oblivious mechanisms for linear queries, but only for mechanisms that support arbitrary real-valued input vectors $x \in \R^n$ and only for neighboring inputs defined as $\|x-x'\|_1 \leq 1$. We modify their reduction, using the special structure of the prefix sums problem, to give the above reduction for binary inputs.

Starting from an $(\eps,\delta)$-DP mechanism $\rM$, we may choose $\delta' = \max\{\eps,\delta\}$ in Lemma~\ref{lem:oblivious} to obtain an $(2\eps, e^\eps \delta +\max\{\eps,\delta\})$-DP oblivious mechanism $\rM_o$ for binary inputs of size $m$ (with $m$ as stated in Lemma~\ref{lem:oblivious}). For $\eps \geq n^{-1+\Omega(1)}$, we have $m = n^{\Omega(1)}$. We then prove an
\[
\Omega\left(\frac{\log^{3/2} m}{\max\{\eps,\delta\}}\right)
\]
for oblivious mechanisms. Since $m = n^{\Omega(1)}$ and $\max\{\eps,\delta\}$ changes only by a constant factor, this gives the result in Theorem~\ref{thm:main}.

With this established, we return to using $n$ as the length of the input vector $x \in \{0,1\}^n$ and prove a lower bound for an $(\eps,\delta)$-DP oblivious mechanism $\rM_o$.

\paragraph{Tree Structure.}
We assume $n$ is a power of $16$ and let $h = \log_2 n$. This only affects the lower bound by a constant factor as we can use the length $16^{\lfloor \log_{16} n \rfloor}$ prefix of $x$ to embed the hard instance. To define a basis for $\R^{n}$, consider a binary tree $\T$ of height $h$, i.e.\ $\T$ has $n$ leaves, where the $i$-th leaf corresponds to the $i$-th coordinate of $Ax+\rz$. Let $\Vecd(\T)$ be the set of nodes that are either the root or the right child of their parent. We use the notation $\ell(u)$ to denote the left child and $r(u)$ to denote the right child of a non-leaf node $u$. For each node $u \in \Vecd(\T)$, we define $\chi^u$ as the indicator vector having a $1$ in coordinates corresponding to the leaves in the subtree $\T^u$ rooted at $u$ and $0$ elsewhere. We also define $\Vecd(\T^u) := \Vecd(\T) \cap \T^u$ as the nodes $v \in \T^u$ with $\chi^v$ defined. 

We observe that the number of vectors $|\Vecd(\T)| = 1 + \sum_{d =1}^{h} 2^{d-1}  = n$. Furthermore, we claim that every \emph{suffix} vector $Ae_i \in \R^n$ having a $1$ in positions $j \geq i$ and $0$ elsewhere can be written as a linear combination of the $\chi^u$ vectors. Since the suffix vectors span $\R^n$, we conclude that the $\chi^u$ vectors span $\R^n$. Moreover, since there are $n$ of them, they form a basis (although neither orthogonal nor unit length). 

To write $Ae_i$ as a linear combination of the $\chi^u$ vectors, initialize $\hat{y}=0$ and traverse the path from the root to the $i$'th leaf. For each node $u$ visited, if $i$ is the left-most leaf in the subtree $\T^u$, add $\chi^u$ to $\hat{y}$ and terminate. Otherwise, if $u$ is an internal node and the path descends to $\ell(u)$, add $\chi^{r(u)}$ to $\hat{y}$. Upon termination, we have $\hat{y} = Ae_i$. Define $\I(i) \subseteq \Vecd(\T)$ as the nodes such that $Ae_i=\sum_{u \in \I(i)} \chi^{u}$.

We observe that 
\begin{align*}
Ax = \sum_{i=1}^n x_i Ae_ i
= \sum_{i=1}^n \sum_{u \in \I(i)} x_i\chi^{u} 
= \sum_{u \in \Vecd(\T)} \sum_{i : u \in \I(i)} x_i\chi^{u}.  
\end{align*}
Thus, if we define 
\[
a_u := \sum_{i : u \in \I(i)} x_i,
\]
then $Ax = \sum_{u \in \Vecd(\T)} a_u \chi^{u}$. We note that since the $\chi^u$ vectors form a basis, this linear combination is unique. 

Since the $\chi^u$ vectors form a basis, there is also a unique linear combination so that $\rz = \sum_{u \in \Vecd(\T)}  \rbeta_u \chi^{u}$ with $\rbeta_u \in \R$. From this it also follows that
\[
\rz_i = \sum_{u \in \Vecd(\T) }  \rbeta_u \chi^{u}_i = \sum_{u \in \Vecd(\T)  : \chi^{u}_i=1}  \rbeta_u.
\]
Note that the nodes $u$ with $\chi^{u}_i=1$ are precisely $\Anc(i) \cap \Vecd(\T)$, where $\Anc(i)$ is the set of all ancestors of the $i$'th leaf, including the leaf itself. Thus, to lower bound $\|\rz\|_\infty$ we need to show that there is an index $i$ for which $\sum_{u \in \Anc(i) \cap \Vecd(\T)} \rbeta_u$ is large in absolute value.

\paragraph{Noise Accumulation.}
We next argue that noise accumulates up the tree $\T$ in a subtle manner. First observe that $\|\rz\|_\infty \geq (\max_i \rz_i - \min_i \rz_i)/2$. We thus lower bound $\E[\max_i \rz_i - \min_i \rz_i]$ instead of directly lower bounding $\E[\|\rz\|_\infty]$, circumventing the absolute value. We do this inductively up the tree $\T$. For any node $u \in \T$, let $\T^u$ be the subtree rooted at $u$ and define $\cL(\T^u)$ as the indices of the leaves in the subtree rooted at $u$. For each $i \in \cL(\T^u)$, define
\[
\rz^{< u}_i := \sum_{v \in (\Vecd(\T^u )\setminus \{u\}) \cap \Anc(i)} \rbeta_v.
\]
That is, $\rz^{< u}_i$ sums only the contributions to $\rz_i$ from nodes strictly below $u$ on the path from the root to the $i$'th leaf. Also define 
\[
\mx^{< u}(\rz) := \max_{i \in \cL(\T^u)} \rz_i^{< u}.
\]
and
\[
\mn^{< u}(\rz) := \min_{i \in \cL(\T^u)} \rz^{< u}_i.
\]
We will show that $\mx^{< u}(\rz)-\mn^{< u}(\rz)$ grows as we move up the tree.

For any node $u \in \T$, define $\Delta^{<u}(\rz)$ as the random variable giving $\rbeta_v$ for every $v \in \Vecd(\T^u) \setminus \{u\}$. That is, $\Delta^{<u}(\rz)$ gives the contributions to the noise vector $\rz$ from $\chi^v$ vectors in the subtree rooted at $u$, excluding $u$ itself. $\Delta^{<u}(\rz)$ is thus a function of $\rz$, $\Delta^{<u}(\rz): \Vecd(\T^u) \setminus \{u\} \to \R$ with $\Delta^{<u}(\rz)(v) = \rbeta_v$ for $v \in \Vecd(\T^u) \setminus \{u\}$.

We now consider how well $\rbeta_{r(u)}$ can be predicted  from the \emph{noise} $\Delta^{<\ell(\ell(u))}(\rz), \Delta^{<\ell(r(u))}(\rz)$ for $u \in \T$ that is not a leaf:
\[
\sigma_{u,\ell}(\Delta^{<\ell(\ell(u))}(\rz), \Delta^{<\ell(r(u))}(\rz)) := \inf_m \E\left[\left|\rbeta_{r(u)} - m\right| \mid \Delta^{<\ell(\ell(u))}(\rz), \Delta^{<\ell(r(u))}(\rz)\right].
\]
Also let
\[
\sigma_{u,\ell} =  \E[\sigma_{u,\ell}(\Delta^{<\ell(\ell(u))}(\rz), \Delta^{<\ell(r(u))}(\rz))].
\]
For later use, we also define $m_\ell^u(\Delta^{<\ell(\ell(u))}(\rz), \Delta^{<\ell(r(u))}(\rz))$ as a measurable minimizer $m$ of 
\[
\inf_m \E\left[\left|\rbeta_{r(u)} - m\right| \mid \Delta^{<\ell(\ell(u))}(\rz), \Delta^{<\ell(r(u))}(\rz)\right],
\]
chosen as the smallest conditional median.

Similarly, we also consider predicting the sum $\rbeta_{r(\ell(u))} -\rbeta_{r(u)} - \rbeta_{r(r(u))}$ from $\Delta^{<r(\ell(u))}(\rz), \Delta^{<r(r(u))}(\rz)$ for internal $u \in \T$ whose children are not leaves:
\[
\sigma_{u,r}(\Delta^{<r(\ell(u))}(\rz), \Delta^{<r(r(u))}(\rz)) := \inf_m \E\left[\left|\rbeta_{r(\ell(u))} - \rbeta_{r(u)} - \rbeta_{r(r(u))} - m\right| \mid \Delta^{<r(\ell(u))}(\rz), \Delta^{<r(r(u))}(\rz)\right],
\]
and let
\[
\sigma_{u,r} =  \E[\sigma_{u,r}(\Delta^{<r(\ell(u))}(\rz), \Delta^{<r(r(u))}(\rz))].
\]
Again, we also define $m^u_r(\Delta^{<r(\ell(u))}(\rz), \Delta^{<r(r(u))}(\rz))$ as a measurable minimizer of 
\[
\inf_m \E\left[\left|\rbeta_{r(\ell(u))} - \rbeta_{r(u)} - \rbeta_{r(r(u))} - m\right| \mid \Delta^{<r(\ell(u))}(\rz), \Delta^{<r(r(u))}(\rz)\right],
\]
chosen as the smallest conditional median. Our goal is to show that $\mx^{< u}(\rz) - \mn^{< u}(\rz)$ grows proportionally to $\sigma_{u,r}$ and $\sigma_{u,\ell}$. Intuitively, this means that if there is still randomness left in $\rbeta_{r(u)}$ after revealing $\Delta^{<\ell(\ell(u))}(\rz), \Delta^{<\ell(r(u))}(\rz)$, or in $\rbeta_{r(\ell(u))} - \rbeta_{r(u)} - \rbeta_{r(r(u))}$ after revealing $\Delta^{<r(\ell(u))}(\rz), \Delta^{<r(r(u))}(\rz)$, then the noise accumulates up the tree to create a large $\mx^{< u}(\rz)-\mn^{< u}(\rz)$. This is captured in the following lemma

\begin{lemma}
\label{lem:noiseAcc}
For any internal node $u \in \T$ whose children are not leaves, we have
\begin{align*}
&\E[\mx^{< u}(\rz) - \mn^{< u}(\rz)] \geq \sigma_{u,\ell}/2 + \sigma_{u,r}/2 + \sum_{f,g \in \{\ell,r\}^2} \E[\mx^{< f(g(u))}(\rz) - \mn^{<f(g(u))}(\rz)]/4.
\end{align*}
\end{lemma}
\begin{proof}
Let $u$ be an internal node whose children are not leaves. We have
\begin{align*}
    &\mx^{< u}(\rz) - \mn^{< u}(\rz) \geq\\
    &\max\{\max_{i \in \cL(\T^{\ell(\ell(u))})} \rz_i^{< u},\max_{i \in \cL(\T^{\ell(r(u))})} \rz_i^{< u} \} - \min\{\min_{i \in \cL(\T^{\ell(\ell(u))})} \rz_i^{< u},\min_{i \in \cL(\T^{\ell(r(u))})} \rz_i^{< u} \}.
\end{align*}
Notice that $\max\{a,b\} = (a+b)/2 + |a-b|/2$ and $\min\{a,b\} = (a+b)/2-|a-b|/2$ for any $a,b$. Hence $\max\{a,b\}-\min\{c,d\} = (a+b)/2 - (c+d)/2 + |a-b|/2  + |c-d|/2 = (a-c)/2+(b-d)/2+ |a-b|/2  + |c-d|/2$. We now see that
\begin{align*}
    \max_{i \in \cL(\T^{\ell(\ell(u))})} \rz_i^{< u} - \min_{i \in \cL(\T^{\ell(\ell(u))})} \rz_i^{< u} &= \max_{i \in \cL(\T^{\ell(\ell(u))})} \rz_i^{< \ell(\ell(u))} - \min_{i \in \cL(\T^{\ell(\ell(u))})} \rz_i^{< \ell(\ell(u))} \\
    &= \mx^{<\ell(\ell(u))}(\rz)-\mn^{<\ell(\ell(u))}(\rz).
\end{align*}
and
\begin{align*}
    \max_{i \in \cL(\T^{\ell(r(u))})} \rz_i^{< u} - \min_{i \in \cL(\T^{\ell(r(u))})} \rz_i^{< u} &= \max_{i \in \cL(\T^{\ell(r(u))})} \rz_i^{< \ell(r(u))} - \min_{i \in \cL(\T^{\ell(r(u))})} \rz_i^{< \ell(r(u))}\\
    &= \mx^{<\ell(r(u))}(\rz)-\mn^{<\ell(r(u))}(\rz).
\end{align*}
Finally, we also have
\begin{align*}
    &\left|\max_{i \in \cL(\T^{\ell(\ell(u))})} \rz_i^{< u} - \max_{i \in \cL(\T^{\ell(r(u))})} \rz_i^{< u}\right| = \\
     &\left|\max_{i \in \cL(\T^{\ell(\ell(u))})} \rz_i^{< \ell(\ell(u))} - \max_{i \in \cL(\T^{\ell(r(u))})} \rz_i^{< \ell(r(u))} - \rbeta_{r(u)}\right| =\\
     &\left| \mx^{<\ell(\ell(u))}(\rz) - \mx^{<\ell(r(u))}(\rz) - \rbeta_{r(u)}\right|.
\end{align*}
and
\begin{align*}
    &\left|\min_{i \in \cL(\T^{\ell(\ell(u))})} \rz_i^{< u} - \min_{i \in \cL(\T^{\ell(r(u))})} \rz_i^{< u}\right| = 
     \left| \mn^{<\ell(\ell(u))}(\rz) - \mn^{<\ell(r(u))}(\rz) - \rbeta_{r(u)}\right|.
\end{align*}
We conclude
\begin{align*}
        &\mx^{< u}(\rz) - \mn^{< u}(\rz) \geq\\
        &\frac{\mx^{<\ell(\ell(u))}(\rz)-\mn^{<\ell(\ell(u))}(\rz)}{2} + \frac{\mx^{<\ell(r(u))}(\rz)-\mn^{<\ell(r(u))}(\rz)}{2} +\\
        &\frac{\left| \mx^{<\ell(\ell(u))}(\rz) - \mx^{<\ell(r(u))}(\rz) - \rbeta_{r(u)}\right|}{2} + \frac{\left| \mn^{<\ell(\ell(u))}(\rz) - \mn^{<\ell(r(u))}(\rz) - \rbeta_{r(u)}\right|}{2}.
\end{align*}
Now observe that $\mx^{<\ell(\ell(u))}(\rz), \mx^{<\ell(r(u))}(\rz),\mn^{<\ell(\ell(u))}(\rz)$ and $\mn^{<\ell(r(u))}(\rz)$ are all determined from $\Delta^{<\ell(\ell(u))}(\rz), \Delta^{<\ell(r(u))}(\rz)$. Hence
\begin{align*}
    &\E\left[\left| \mx^{<\ell(\ell(u))}(\rz) - \mx^{<\ell(r(u))}(\rz) - \rbeta_{r(u)}\right| \mid \Delta^{<\ell(\ell(u))}(\rz), \Delta^{<\ell(r(u))}(\rz)\right] \geq \sigma_{u,\ell}(\Delta^{<\ell(\ell(u))}(\rz), \Delta^{<\ell(r(u))}(\rz))
\end{align*}
and
\begin{align*}
    &\E\left[\left| \mn^{<\ell(\ell(u))}(\rz) - \mn^{<\ell(r(u))}(\rz) - \rbeta_{r(u)}\right| \mid \Delta^{<\ell(\ell(u))}(\rz), \Delta^{<\ell(r(u))}(\rz) \right] \geq \sigma_{u,\ell}(\Delta^{<\ell(\ell(u))}(\rz), \Delta^{<\ell(r(u))}(\rz)).
\end{align*}
Thus we have
\begin{align*}
        &\E[\mx^{< u}(\rz) - \mn^{< u}(\rz) \mid \Delta^{<\ell(\ell(u))}(\rz), \Delta^{<\ell(r(u))}(\rz)] \geq\\
        &\frac{\E[\mx^{<\ell(\ell(u))}(\rz)-\mn^{<\ell(\ell(u))}(\rz) \mid \Delta^{<\ell(\ell(u))}(\rz), \Delta^{<\ell(r(u))}(\rz)]}{2} +\\ &\frac{\E[\mx^{<\ell(r(u))}(\rz)-\mn^{<\ell(r(u))}(\rz)\mid \Delta^{<\ell(\ell(u))}(\rz), \Delta^{<\ell(r(u))}(\rz)]}{2} +\\
        &\frac{\sigma_{u,\ell}(\Delta^{<\ell(\ell(u))}(\rz), \Delta^{<\ell(r(u))}(\rz))}{2} + \frac{\sigma_{u,\ell}(\Delta^{<\ell(\ell(u))}(\rz), \Delta^{<\ell(r(u))}(\rz))}{2}.
\end{align*}
By the law of total expectation, we conclude
\begin{align*}
&\E[\mx^{< u}(\rz) - \mn^{< u}(\rz)] \geq \\
&\frac{\E[\mx^{<\ell(\ell(u))}(\rz)-\mn^{<\ell(\ell(u))}(\rz)]}{2} + \frac{\E[\mx^{<\ell(r(u))}(\rz)-\mn^{<\ell(r(u))}(\rz)]}{2} + \sigma_{u,\ell}.
\end{align*}
Similarly for the grandchildren of $u$ that are right children of their parents, namely $r(\ell(u))$ and $r(r(u))$, we have
\begin{align*}
    &\mx^{< u}(\rz) - \mn^{< u}(\rz) \geq\\
    &\max\{\max_{i \in \cL(\T^{r(\ell(u))})} \rz_i^{< u},\max_{i \in \cL(\T^{r(r(u))})} \rz_i^{< u} \} - \min\{\min_{i \in \cL(\T^{r(\ell(u))})} \rz_i^{< u},\min_{i \in \cL(\T^{r(r(u))})} \rz_i^{< u} \}.
\end{align*}
Using the same formula for $\max\{a,b\}-\min\{c,d\}$, we notice that
\begin{align*}
    &\left|\max_{i \in \cL(\T^{r(\ell(u))})} \rz_i^{< u} - \max_{i \in \cL(\T^{r(r(u))})} \rz_i^{< u}\right| = \\
     &\left|\max_{i \in \cL(\T^{r(\ell(u))})} \rz_i^{< r(\ell(u))} + \rbeta_{r(\ell(u))} - \max_{i \in \cL(\T^{r(r(u))})} \rz_i^{< r(r(u))} - \rbeta_{r(u)} - \rbeta_{r(r(u))}\right| =\\
     &\left| \mx^{<r(\ell(u))}(\rz) - \mx^{<r(r(u))}(\rz)+ \rbeta_{r(\ell(u))} - \rbeta_{r(u)} - \rbeta_{r(r(u))}\right|.
\end{align*}
and
\begin{align*}
    &\left|\min_{i \in \cL(\T^{r(\ell(u))})} \rz_i^{< u} - \min_{i \in \cL(\T^{r(r(u))})} \rz_i^{< u}\right| = 
     \left| \mn^{<r(\ell(u))}(\rz) - \mn^{<r(r(u))}(\rz)+ \rbeta_{r(\ell(u))} - \rbeta_{r(u)} - \rbeta_{r(r(u))}\right|.
\end{align*}
We notice that $\mx^{<r(\ell(u))}(\rz), \mx^{<r(r(u))}(\rz),\mn^{<r(\ell(u))}(\rz)$ and $\mn^{<r(r(u))}(\rz)$ are all determined from $\Delta^{<r(\ell(u))}(\rz), \Delta^{<r(r(u))}(\rz)$. Thus by the same arguments as above, we conclude
\begin{align*}
&\E[\mx^{< u}(\rz) - \mn^{< u}(\rz)] \geq \\
&\frac{\E[\mx^{<r(\ell(u))}(\rz)-\mn^{<r(\ell(u))}(\rz)]}{2} + \frac{\E[\mx^{<r(r(u))}(\rz)-\mn^{<r(r(u))}(\rz)]}{2} + \sigma_{u,r}.
\end{align*}
By averaging, we finally conclude
\begin{align*}
&\E[\mx^{< u}(\rz) - \mn^{< u}(\rz)] \geq \sigma_{u,\ell}/2 + \sigma_{u,r}/2 + \sum_{f,g \in \{\ell,r\}^2} \E[\mx^{< f(g(u))}(\rz) - \mn^{<f(g(u))}(\rz)]/4.
\end{align*}
\end{proof}
Recall that we assume $n$ is a power of $16$ and thus $h = \log_2 n$ is even. If we say that the root $\rho$ is at depth $0$, then the leaves are at depth $h$. If we now recursively apply Lemma~\ref{lem:noiseAcc} on the even levels of $\T$, starting from the root $\rho$ and using that $\mx^{<u}(\rz)=\mn^{<u}(\rz)=0$ when $u$ is a leaf, we conclude
\begin{align*}
\E[\mx^{<\rho}(\rz)-\mn^{<\rho}(\rz)] &\geq  
\sum_{d=0}^{h/2-1} \sum_{u \in \T : d(u)=2d} 2^{-2d-1}(\sigma_{u,\ell} + \sigma_{u,r}),
\end{align*}
where $d(u)$ is the depth of $u$. Finally, observing that the root $\rho$ satisfies $\rho \in \Anc(i)$ for all $i$ we have
\begin{align*}
    \E[\|\rz\|_\infty] &\geq \E[\mx_i \rz_i - \mn_i \rz_i]/2 \\
    &= \E[\mx^{<\rho}(\rz) + \rbeta_\rho - (\mn^{<\rho}(\rz) + \rbeta_\rho)]/2 \\
    &= \E[\mx^{<\rho}(\rz) - \mn^{<\rho}(\rz)]/2 \\
    &\geq \sum_{d=0}^{h/2-1} \sum_{u \in \T : d(u)=2d} 2^{-2d-2}(\sigma_{u,\ell} + \sigma_{u,r}).
\end{align*}
We now rewrite this into an alternative, but more convenient representation. Let $\ri$ be a uniform random leaf index. Then for a node $u$, we have $\Pr[u \in \Anc(\ri)] = 2^{-d(u)}$. Letting $\pAnc(i)$ be the \emph{proper} ancestors of the $i$'th leaf, i.e.\ $\Anc(i)$ but without the $i$'th leaf itself, we have
\begin{align}
        \E[\|\rz\|_\infty] &\geq \sum_{d=0}^{h/2-1} \sum_{u \in \T : d(u)=2d} \Pr[u \in \Anc(\ri)](\sigma_{u,\ell} + \sigma_{u,r})/4 \nonumber\\
        &= \sum_{d=0}^{h/2-1} \sum_{u \in \T : d(u)=2d}\E[\one\{u \in \Anc(\ri)\}](\sigma_{u,\ell} + \sigma_{u,r})/4  \nonumber\\
        &= \frac{1}{4} \cdot \E_{\ri \in \cL(\T)}\left[\sum_{u \in \pAnc(\ri) : d(u) \equiv 0 \bmod 2 } \sigma_{u,\ell} +\sigma_{u,r}\right].
\end{align}
What remains is thus to lower bound the expected sum of $\sigma_{u,\ell}$ and $\sigma_{u,r}$ on the even levels of a uniformly random root-to-leaf path.

\paragraph{Restricting to a Subtree.}
We now restrict attention to a subtree of $\T$. Recall $h$ is even, so there are $h/2$ odd depths $1,3,\dots,h-1$. For a string $s \in \{\ell,r\}^{h/2}$, let $\T^s$ be the subset of $\T$ containing all nodes $u$ for which the path from the root to $u$ satisfies that in every odd depth $2j-1$, the path to $u$ descends to the left child if $s_j=\ell$ and to the right child if $s_j=r$. For even levels, we allow descending into both children. Let $\cL(\T^s)$ be the indices in the leaves of $\T^s$.

Now let $\rs$ be uniform random in $\{\ell ,r\}^{h/2}$ and let $\ri$ be uniform random in $\cL(\T^\rs)$. Then the distribution of $\ri$ is uniform among all leaves of $\T$. By linearity of expectation, there must exist a choice of $s \in \{\ell,r\}^{h/2}$ such that
\begin{align*}
\E[\|\rz\|_\infty] &\geq \frac{1}{4} \cdot \E_{\ri \in \cL(\T^{s})}\left[\sum_{u \in \pAnc(\ri) : d(u) \equiv 0 \bmod 2 } \sigma_{u,\ell} +\sigma_{u,r}\right] \\
&= \frac{1}{4} \cdot \sum_{d=0}^{h/2-1} \sum_{u \in \T^{s} : d(u)=2d} \Pr_{\ri \in \cL(\T^{s})}[u \in \Anc(\ri)](\sigma_{u,\ell} + \sigma_{u,r}) \\
&=   \frac{1}{4} \cdot \sum_{d=0}^{h/2-1} \sum_{u \in \T^{s} : d(u)=2d} 2^{-d}(\sigma_{u,\ell} + \sigma_{u,r}).
\end{align*}
Fix such an $s \in \{\ell ,r \}^{h/2}$. For a character $a \in \{\ell,r\}$, let $\tilde{a}$ denote the opposite character, i.e.\ $\tilde{\ell}=r$ and $\tilde{r}=\ell$. Then, we further have
\begin{align}
\label{eq:needsum}
\E[\|\rz\|_\infty] &\geq  \frac{1}{4} \cdot \sum_{d=0}^{h/2-1} \sum_{u \in \T^{s} : d(u)=2d} 2^{-d} \cdot \sigma_{u,\widetilde{s_{d+1}}}.
\end{align}
We will thus show that this sum is large for any $s$.

\paragraph{From Weighted Sum to Max Bound.}
The bound in~\eqref{eq:needsum} requires us to bound a weighted sum of $\sigma_{u,\widetilde{s_{d+1}}}$ values. We reduce this to the simpler case of merely showing that at least one term in the sum is large.

More concretely, consider choosing $\cJ \subseteq \{0,\dots,h/2-1\}$ as a subset of $h/4$ out of the $h/2$ choices for $d$ in the outer sum in~\eqref{eq:needsum}. Next, for each $d \in \cJ$ with $d>0$, choose a subset $\cJ_d$ of $2^{d-1}$ out of the $2^d$ nodes $u \in \T^s$ with $d(u)=2d$. If $d=0$ is in $\cJ$, then that corresponds to the root and we choose $\cJ_0 = \{\rho\}$, which is the only node at depth $0$. Assume we can show that for \emph{any} such choice of $\cJ$ and $\{\cJ_d\}_{d \in \cJ}$, we have
\[
\max_{d \in \cJ} \max_{u \in \cJ_d} \sigma_{u,\widetilde{s_{d+1}} } \geq \sigma,
\]
then we claim that 
\begin{align*}
    \E[\|\rz\|_\infty] \geq \frac{h \sigma}{32}.
\end{align*}
To see this, define for each $d \in \{0,\dots,h/2-1\}$ the value 
\[
m_d := \median_{u \in \T^s : d(u)=2d}\{ \sigma_{u,\widetilde{s_{d+1}}}\},
\]
where the median, for an even number of elements $n_d$, is defined as the $n_d/2$'th smallest element and for a set of size $1$ is defined as the single element. Now let $\cJ \subseteq \{0,\dots,h/2-1\}$ be the indices $d$ with the $h/4$ smallest values of $m_d$. For each $d \in \cJ$, let $\cJ_d$ contain the $2^{d-1}$ nodes $u \in \T^s$ with $d(u)=2d$ having the smallest values of $\sigma_{u,\widetilde{s_{d+1}}}$. By assumption, there is a $d^\star \in \cJ$ and $u \in \cJ_{d^\star}$ with $\sigma_{u,\widetilde{s_{d^\star+1}}} \geq \sigma$. By definition of $\cJ_{d^\star}$, this implies $m_{d^\star} \geq \sigma$. By definition of $\cJ$, it follows that for every $d \notin \cJ$, we have $m_d \geq m_{d^\star} \geq \sigma$. Finally, for every $d \notin \cJ$, it follows that at least half the nodes $u \in \T^s$ with $d(u)=2d$ have $\sigma_{u,\widetilde{s_{d+1}}} \geq \sigma$. Since there are $2^d$ nodes $u$ in $\T^s$ with $d(u)=2d$, we conclude from~\eqref{eq:needsum} that
\begin{align}
\label{eq:lbhsigma}
\E[\|\rz\|_\infty] &\geq \frac{1}{4} \cdot \sum_{d=0}^{h/2-1} \sum_{u \in \T^{s} : d(u)=2d} 2^{-d} \cdot \sigma_{u,\widetilde{s_{d+1}}} \nonumber\\
&\geq \frac{1}{4} \cdot \sum_{d \notin \cJ} \sum_{u \in \T^{s} : d(u)=2d} 2^{-d} \cdot \sigma_{u,\widetilde{s_{d+1}}}  \nonumber\\
    &\geq \frac{1}{8} \sum_{d \notin \cJ} \sigma \nonumber\\
    &= \frac{h \sigma}{32}.
\end{align}
We can thus focus on showing that for any $\cJ \subseteq
\{0,\dots,h/2-1\}$ with $|\cJ|=h/4$ and any $\{\cJ_d\}_{d \in \cJ}$,
with $\cJ_d$ consisting of half the nodes of $\T^s$ at depth $2d$,
there is at least one $d \in \cJ$ and one $u \in \cJ_d$ such that
$\sigma_{u,\widetilde{s_{d+1}}}$ is large.

\paragraph{Noise Reduction.}
The first step in bounding $\sigma = \max_{d \in \cJ} \max_{u \in \cJ_d} \sigma_{u, \widetilde{s_{d+1}}}$ for any choice of $\cJ$ and $\{\cJ_d\}_{d \in \cJ}$,
is to design an alternative mechanism whose noise is bounded in terms
of $\sigma_{u,\widetilde{s_{d+1}}}$.

Consider an arbitrary $s \in \{\ell,r\}^{h/2}$. We
will design a mechanism $\widehat{\rM}_s$ that is
$(3\eps,3\delta)$-DP when restricted to inputs $x \in \{0,1\}^n$ with
$x_i=0$ for $i \notin \cL(\T^s)$. The mechanism $\widehat{\rM}_s$ does
not itself output a noisy version of the prefix sums $Ax$. Instead, it
outputs a vector $\widehat{\rM}_s(x)=\bq$ with one coordinate $\bq_u$ for each node $u \in \T^s$ with
an even depth $d(u) < h$.

In more detail, the mechanism $\widehat{\rM}_s$ draws three i.i.d.\ copies $\rz^{(1)}, \rz^{(2)}, \rz^{(3)}$ of the
noise $\rz$ of the mechanism $\rM_o(x) = Ax+\rz$. Let $\rbeta^{(k)}_u$
for $u \in \Vecd(\T)$ be such that $\rz^{(k)} = \sum_{u \in \Vecd(\T)}
\rbeta^{(k)}_u \chi^u$. Then, for any $d \in
\{0,\dots,h/2-1\}$ and node $u \in \T^s$ with depth $2d$, if
$s_{d+1}=r$, let
\begin{align}
  \raq_u := \median_{k=1,2,3} \{ \rbeta^{(k)}_{r(u)} -
  m^u_\ell(\Delta^{<\ell(\ell(u))}(\rz^{(k)}),
  \Delta^{<\ell(r(u))}(\rz^{(k)})) \},
  \label{eq:q_right}
\end{align}
and let the mechanism $\widehat{\rM}_s$ set
\[
  \bq_u := \raq_u + a_{r(u)}.
\]
Otherwise, if $s_{d+1}=\ell$, let
\begin{align}
  \raq_u := \median_{k=1,2,3} \{ \rbeta^{(k)}_{r(\ell(u))} -\rbeta^{(k)}_{r(u)}-\rbeta^{(k)}_{r(r(u))}-
  m^u_r(\Delta^{<r(\ell(u))}(\rz^{(k)}),
          \Delta^{<r(r(u))}(\rz^{(k)})) \} 
          \label{eq:q_left}
\end{align}
and let the mechanism set
\begin{align*}
  \bq_u &:= \raq_u + a_{r(\ell(u))} - a_{r(u)} - a_{r(r(u))}.
\end{align*}

\begin{lemma}
The mechanism $\widehat{\rM}_s$ is $(3\eps,3\delta)$-DP when restricted
to inputs $x \in \{0,1\}^n$ with
$x_i=0$ for $i \notin \cL(\T^s)$.
\end{lemma}

\begin{proof}
Let $x^{(0)},x^{(1)} \in \{0,1\}^n$ be an arbitrary neighboring pair
with $x^{(0)}_i=x^{(1)}_i=0$ for $i \notin \cL(\T^s)$. Let $i \in \cL(\T^s)$ be
the coordinate on which $0=x_i^{(0)} \neq x^{(1)}_i=1$. Given an $x \in \{x^{(0)},x^{(1)}\}$, we argue that we can obtain
$\widehat{\rM}_s(x)$ as post-processing of $\rM_o'(x)=(Ax+\rz^{(1)},
Ax+\rz^{(2)}, Ax+\rz^{(3)})$. By standard composition, we have
that the mechanism $\rM_o'$ is $(3\eps,3\delta)$-DP and thus the same
holds for $\widehat{\rM}_s$ when restricted to inputs $x$ with $x_i =0$
for $i \notin \cL(\T^s)$. Note that our post-processing function depends on the pair $x^{(0)}, x^{(1)}$. This is sufficient because the DP inequality is checked separately for every neighboring pair.

Now to obtain $\widehat{\rM}_s(x)$ from $\rM_o'(x)$, first compute for
every node $u \in \Vecd(\T)$ the value
\[
  a'_u = \sum_{j \neq i : u \in
    \I(j)} x^{(0)}_j = \sum_{j \neq i : u \in
    \I(j)} x_j.
\]
Next, for $k=1,2,3$, write $Ax+\rz^{(k)}$ in the basis $\chi^u$ as
$Ax+\rz^{(k)} = \sum_{u \in \Vecd(\T)} (a_u + \rbeta^{(k)}_u)\chi^u$. For
each $u$, subtract $a'_u$ from $(a_u + \rbeta^{(k)}_u)$ to obtain the value
\[
\rbeta^{(k)}_u + \one\{ u \in \I(i)\}x_i.
\]
Now consider any $d
\in \{0,\dots,h/2-1\}$ and node $u \in \T^s$ with depth $2d$.

If
$s_{d+1}=r$, then any leaf with an index $i \in \cL(\T^s)$ falls in
either the subtree rooted at $r(\ell(u))$ or $r(r(u))$. By definition
of $\I(i)$, we get that $\I(i)$ is disjoint from
$\Vecd(\T^{\ell(\ell(u))}) \cup \Vecd(\T^{\ell(r(u))}) \setminus
\{\ell(\ell(u)), \ell(r(u))\}$. Hence for $v \in \Vecd(\T^{\ell(\ell(u))}) \cup \Vecd(\T^{\ell(r(u))}) \setminus
\{\ell(\ell(u)), \ell(r(u))\}$, we have
\[
\rbeta^{(k)}_v + \one\{ v \in \I(i)\}x_i = \rbeta^{(k)}_v.
\]
This allows us to obtain $\Delta^{<\ell(\ell(u))}(\rz^{(k)})$ and
$\Delta^{<\ell(r(u))}(\rz^{(k)})$. Finally, since we have $a_v +
\rbeta^{(k)}_v$ for every $v \in \Vecd(\T)$, we can compute
\begin{align*}
&\median_{k=1,2,3} \{ a_{r(u)} +
\rbeta^{(k)}_{r(u)} -
                 m_\ell^u(\Delta^{<\ell(\ell(u))}(\rz^{(k)}),\Delta^{<\ell(r(u))}(\rz^{(k)}))\}
                 =\\
&\median_{k=1,2,3} \{
\rbeta^{(k)}_{r(u)} -
                 m_\ell^u(\Delta^{<\ell(\ell(u))}(\rz^{(k)}),\Delta^{<\ell(r(u))}(\rz^{(k)}))\}
  + a_{r(u)}
                       = \\
  &\bq_u.
\end{align*}
Symmetrically, if $s_{d+1}=\ell$, we have that $\I(i)$ is disjoint from $\Vecd(\T^{r(\ell(u))}) \cup \Vecd(\T^{r(r(u))}) \setminus
\{r(\ell(u)), r(r(u))\}$. We can thus compute $\Delta^{<r(\ell(u))}(\rz^{(k)})$ and
$\Delta^{<r(r(u))}(\rz^{(k)})$. Again using that we know $a_v +
\rbeta^{(k)}_v$ for every $v \in \Vecd(\T)$, we can compute
\begin{align*}
&\median_{k=1,2,3} \{ a_{r(\ell(u))}-a_{r(u)}-a_{r(r(u))} + \rbeta^{(k)}_{r(\ell(u))}-\rbeta^{(k)}_{r(u)}-\rbeta^{(k)}_{r(r(u))}
 -m_r^u(\Delta^{<r(\ell(u))}(\rz^{(k)}),\Delta^{<r(r(u))}(\rz^{(k)}))\}
                 =\\
&\median_{k=1,2,3} \{ \rbeta^{(k)}_{r(\ell(u))}-\rbeta^{(k)}_{r(u)}-\rbeta^{(k)}_{r(r(u))}
 -m_r^u(\Delta^{<r(\ell(u))}(\rz^{(k)}),\Delta^{<r(r(u))}(\rz^{(k)}))\}+ a_{r(\ell(u))}-a_{r(u)}-a_{r(r(u))}
                       =\\
  &\bq_u.
\end{align*}
\end{proof}

Next we argue that $\bq_u$ has noise bounded roughly by
$\sigma_{u,\widetilde{s_{d+1}}}^2$.

\begin{lemma}
\label{lem:noisebound}
For any $u \in \T^s$ at an even depth $2d < h$, we have if $s_{d+1} = r$
\[
\E[\median_{k=1,2,3} \{
\rbeta^{(k)}_{r(u)} -
                 m_\ell^u(\Delta^{<\ell(\ell(u))}(\rz^{(k)}),\Delta^{<\ell(r(u))}(\rz^{(k)}))\}^2]
                 \leq 3 \cdot \sigma_{u,\widetilde{s_{d+1}}}^2, 
\]
and if $s_{d+1} = \ell$
\[
 \E[\median_{k=1,2,3} \{ \rbeta^{(k)}_{r(\ell(u))}-\rbeta^{(k)}_{r(u)}-\rbeta^{(k)}_{r(r(u))}
 -m_r^u(\Delta^{<r(\ell(u))}(\rz^{(k)}),\Delta^{<r(r(u))}(\rz^{(k)}))\}^2]
 \leq 3 \cdot \sigma_{u,\widetilde{s_{d+1}}}^2.
\]
\end{lemma}
Lemma~\ref{lem:noisebound} follows almost immediately from Corollary~\ref{cor:median}
\begin{proof}[Proof of Lemma~\ref{lem:noisebound}]
  Using Corollary~\ref{cor:median} we have for $u \in \T^s$ at depth $2d$ if $s_{d+1} = r$
  \begin{align*}
        &\E[\median_{k=1,2,3} \{
\rbeta^{(k)}_{r(u)} -
                 m_\ell^u(\Delta^{<\ell(\ell(u))}(\rz^{(k)}),\Delta^{<\ell(r(u))}(\rz^{(k)}))\}^2]
          \leq \\
    &3 \cdot \E[|
\rbeta^{(1)}_{r(u)} -
                 m_\ell^u(\Delta^{<\ell(\ell(u))}(\rz^{(1)}),\Delta^{<\ell(r(u))}(\rz^{(1)}))|]^2
      =\\
    &3 \cdot \sigma^2_{u,\ell} = 3 \cdot \sigma_{u,\widetilde{s_{d+1}}}^2.
  \end{align*}
  and if $s_{d+1} = \ell$
  \begin{align*}
    &\E[\median_{k=1,2,3} \{ \rbeta^{(k)}_{r(\ell(u))}-\rbeta^{(k)}_{r(u)}-\rbeta^{(k)}_{r(r(u))}
 -m_r^u(\Delta^{<r(\ell(u))}(\rz^{(k)}),\Delta^{<r(r(u))}(\rz^{(k)}))\}^2]
    \leq \\
    &3 \cdot \E[|\rbeta^{(1)}_{r(\ell(u))}-\rbeta^{(1)}_{r(u)}-\rbeta^{(1)}_{r(r(u))}
 -m_r^u(\Delta^{<r(\ell(u))}(\rz^{(1)}),\Delta^{<r(r(u))}(\rz^{(1)}))|]^2
    =\\
    &3 \cdot \sigma_{u,r}^2 = 3 \cdot \sigma_{u,\widetilde{s_{d+1}}}^2.
  \end{align*}
\end{proof}

\paragraph{Privacy Attacks on $\widehat{\rM}_s$.}
Based on our analysis of the residual noise above, we now construct an adversary that deduces the value of $\rx_i$ for a fixed $i \in \cL(\T^s)$ from the output of $\widehat{\rM}_s$ via linear measurements, when $\sigma$ is too small. Specifically, we show that for any $\cJ \subseteq \{0,\dots,h/2-1\}$ with $|\cJ|=h/4$ and any $\{\cJ_d\}_{d \in \cJ}$, with $\cJ_d$ consisting of half the nodes of $\T^s$ at depth $2d$ for $d>0$ and $\cJ_0 = \{\rho\}$ if $0\in \cJ$, there is at least one $d \in \cJ$ and one $u \in \cJ_d$ such that $\sigma_{u,\widetilde{s_{d+1}}}$ is large. Inserting this in~\eqref{eq:lbhsigma} implies that the $\ell_\infty$ error of the original mechanism $\rM$ is large. 

For fixed $s \in \{\ell, r\}^{h/2}$, $i \in \cL(\T^s)$, $\cJ \subseteq \{0,\dots,h/2-1\}$ with $|\cJ|=h/4$ and any $\{\cJ_d\}_{d \in \cJ}$ with $\cJ_d$ consisting of half the nodes of $\T^s$ at depth $2d$ for $d>0$ and $\cJ_0 = \{\rho\}$ if $0 \in \cJ$, we design an adversary that attempts to predict $\rx_i$ from $\widehat{\rM}_s(\rx)$, while knowing the values of the coordinates $\rx_j$ for $j\neq i$. Since $\widehat{\rM}_s(\rx)$ is $(6\eps,3e\delta)$-DP for $\rx$ that is zero for $j \notin \cL(\T^s)$, it must be the case that for any neighboring pair of inputs $x',x'' \in \{0,1\}^n$ with $x_i' = 0$, $x_i'' = 1$, $x_j' = x_j'' = 0$ for $j\notin \cL(\T^s)$, and $x_j' = x_j''$ for all $j\neq i$, if $\rx$ is chosen uniformly among $\{x',x''\}$, then no adversary can guess $\rx_i$ from $\widehat{\rM}_s(\rx)$ with high probability, even when given knowledge of the pair $\{x',x''\}$. Let $\rx$ be uniform in $\{x',x''\}$. We consider an adversary that estimates $\rx_i$ as in~\Cref{alg:adversary}.

\begin{algorithm}[ht]
    \caption{Adversary for $\widehat{\rM}_s$}
    \label{alg:adversary}
    \KwIn{$s \in \{\ell, r\}^{h/2}$, $i \in \cL(\T^s)$, $\cJ$, $\{\cJ_d\}_{d\in \cJ}$, $x'$, $\widehat{\rM}_s(\rx)$}
    \KwOut{Estimate $\hat{\rx}_i $}
    Initialize $\rc^i_u = \mu^i_u = \psi^i_u = 0$ for all $u \in \T^s$ with even depth $d(u) < h$, and set $\rc^i_\Lambda = \mu^i_\Lambda = 0$ and $\psi^i_{\Lambda} = 2^{-h/4}$\;
    \For{$u \in \T^s$ with even depth $2d<h$}{
        \eIf{$i \in \cL(\T^{\ell(u)})$}{
            $\psi^i_u = -2^{-(h/2-d)/2}$\;
        }{
            $\psi^i_u = 2^{-(h/2-d)/2}$\;
        }
    }
    \For{$d \in \cJ$ and $u \in \cJ_d$}{
        \eIf{$s_{d+1} = r$}{
            $\rc_u^i = 2^{-d/2}\left(\widehat{\rM}_s(\rx)_{u} - \sum_{j \neq i : r(u) \in \I(j)} x_j'\right)$\;
            $\mu^i_u = 2^{-d/2}\one \{i \in \cL(\T^{r(\ell(u))})\}$\;
        }{
            $\rc_u^i = 2^{-d/2}\left(\widehat{\rM}_s(\rx)_{u} - \sum_{j \neq i: r(\ell(u)) \in \I(j)} x_j' +\sum_{j \neq i: r(u) \in \I(j)} x_j' +\sum_{j \neq i: r(r(u)) \in \I(j)} x_j'\right)$\;
            $\mu_u^i = -2^{-d/2}\one\{i \in \cL(\T^{\ell(r(u))})\}$\;
        }
    }
    Compute $\hat{\rx}_i =  \frac{\langle \psi^i, \rc^i \rangle}{\langle \psi^i, \mu^i \rangle}$\;
    \Return $\hat{\rx}_i$\;
\end{algorithm}

In essence, the adversary performs a privacy attack via linear measurements $\psi^i$. The vectors $\psi^i$ are inspired by the Haar wavelet basis and are defined independently of $\rx$. The key step is to show that the vector $\rc^i$ computed by the adversary can be decomposed as
\[
\rc^i = \rx_i \mu^i + \rr,
\]
where $\mu^i$ is a deterministic vector depending only on $i$, $s$, $\cJ$, $\{\cJ_d\}_{d\in\cJ}$, and $\rr$ is a noise vector whose $\ell_2$ norm is controlled by $\sigma$. Given this decomposition, the adversary computes
\[
\hat{\rx}_i = \frac{\langle \psi^i, \rc^i \rangle}{\langle \psi^i, \mu^i \rangle} = \rx_i + \frac{\langle \psi^i, \rr \rangle}{\langle \psi^i, \mu^i \rangle},
\]
which is a private estimate of $\rx_i$ with additive noise $\langle \psi^i, \rr\rangle / \langle \psi^i, \mu^i\rangle$. We now establish the decomposition $\rc^i = \rx_i\mu^i + \rr$.

For each $u \in
\Vecd(\T)$, let $\ra_u = \sum_{i : u \in \I(i)} \rx_i$ so that $A\rx =
\sum_{u \in \Vecd(\T)} \ra_u \chi^u$. For every $d \in \cJ$ and every $u \in \cJ_d$, if $s_{d+1}=r$, the
adversary subtracts 
\[
   \sum_{j \neq i : r(u) \in \I(j)} \rx_j
\]
from $\widehat{\rM}_s(\rx)_u$ to obtain
\begin{align*}
  &\widehat{\rM}_s(\rx)_u - \sum_{j \neq i : r(u) \in \I(j)} \rx_j= \raq_u + \ra_{r(u)} - \sum_{j \neq i : r(u) \in \I(j)} \rx_j = \raq_u+ \one\{r(u) \in \I(i)\} \rx_i,
\end{align*}
where $\raq_u$ is defined as in~\eqref{eq:q_right}.
We observe that for $s_{d+1}=r$ we have that $i$ is in either the subtree
rooted at $r(\ell(u))$ or $r(r(u))$ and thus $\one\{r(u) \in \I(i)\}
=\one\{i \in \cL(\T^{r(\ell(u))})\}$. The adversary has thus computed
\[
  \raq_u + \one\{i \in \cL(\T^{r(\ell(u))})\}\rx_i.
\]
Now multiply this by $2^{-d/2}$ to obtain
\[
\rc_u^i = 2^{-d/2}\left(\raq_u + \one\{i \in \cL(\T^{r(\ell(u))})\}\rx_i\right).
\]
If instead $s_{d+1}=\ell$, the adversary subtracts
\[
  \sum_{j \neq i: r(\ell(u)) \in \I(j)} \rx_j - \sum_{j \neq i: r(u) \in \I(j)} \rx_j -\sum_{j \neq i: r(r(u)) \in \I(j)} \rx_j 
\]
from $\widehat{\rM}_s(\rx)_u$ to obtain
\begin{align*}
&\widehat{\rM}_s(\rx)_u - \left(
    \sum_{j \neq i: r(\ell(u)) \in \I(j)} \rx_j - \sum_{j \neq i: r(u)
    \in \I(j)} \rx_j -\sum_{j \neq i: r(r(u)) \in \I(j)} \rx_j \right)=\\
  &\raq_u + \ra_{r(\ell(u))} -\ra_{r(u)} - \ra_{r(r(u))} - \left(
    \sum_{j \neq i: r(\ell(u)) \in \I(j)} \rx_j - \sum_{j \neq i: r(u)
    \in \I(j)} \rx_j -\sum_{j \neq i: r(r(u)) \in \I(j)} \rx_j \right)
    =\\
  & \raq_u + \left(\one\{r(\ell(u)) \in \I(i)\} - \one\{r(u) \in \I(i)\}- \one\{r(r(u)) \in \I(i)\}   \right)\rx_i,
\end{align*}
where $\raq_u$ is defined in~\eqref{eq:q_left}.
Since for $s_{d+1}=\ell$ we have that $i$ is either in the
subtree rooted at $\ell(\ell(u))$ or at $\ell(r(u))$, we conclude by the
definition of $\I(i)$ that if $i \in \cL(\T^{\ell(\ell(u))})$, we have
\begin{align*}
  &\one\{r(\ell(u)) \in \I(i)\} - \one\{r(u) \in \I(i)\}-
    \one\{r(r(u)) \in \I(i)\}   = 1-1-0 = 0.
\end{align*}
and if $i \in \cL(\T^{\ell(r(u))})$ we have
\begin{align*}
  &\one\{r(\ell(u)) \in \I(i)\} - \one\{r(u) \in \I(i)\}-
    \one\{r(r(u)) \in \I(i)\}   = 0-0-1 = -1.
\end{align*}
The adversary has thus computed
\[
  \raq_u - \one\{i \in \cL(\T^{\ell(r(u))})\}\rx_i.
\]
Now multiply this by $2^{-d/2}$ to obtain
\[
\rc_u^i = 2^{-d/2}\left(\raq_u - \one\{i \in \cL(\T^{\ell(r(u))})\}\rx_i \right).
\]

Now define a vector $\rr$ with one coordinate $\rr_u$ for each $u \in \T^s$ at an even depth less than $h$ and an additional dummy coordinate $\rr_\Lambda = 0$. Define $\rr_u = 0$ for $u \notin \cup_{d \in \cJ} \cJ_d$ and define $\rr_u = 2^{-d/2} \raq_u$ otherwise.
By the above arguments, $\rc^i = \rx_i \mu^i + \rr$, establishing the promised decomposition. For $i \in \cL(\T^s)$, the adversary computes the estimate
\[
  \hat{\rx}_i = \frac{\langle \psi^i, \rc^i\rangle}{\langle \psi^i, \mu^i \rangle}=\frac{\langle \psi^i, \rr + \rx_i \mu^i\rangle}{\langle \psi^i, \mu^i \rangle} = \rx_i + \frac{\langle \psi^i, \rr\rangle}{\langle \psi^i, \mu^i \rangle}.
\]

Since $\hat{\rx}_i$ is a deterministic function of $\widehat{\rM}_s(\rx)$ and $x'$, it is $(6\eps, 3e\delta)$-differentially private by the postprocessing property of differential privacy. Moreover, $\hat{\rx}_i$ is an estimate of the binary variable $\rx_i$ with additive noise $\frac{\langle \psi^i, \rr\rangle}{\langle \psi^i, \mu^i \rangle}$, which is independent of $\rx_i$. This is because $\rr$ only depends on $s$, $\cJ$, $\{\cJ_d\}_{d\in \cJ}$, and the noise vectors $\rz^{(1)}, \rz^{(2)}, \rz^{(3)}$ (which are independent of $\rx$ because $\rM_o$ is oblivious). We observe that none of $\psi^i$, $\mu^i$, $\rr$ depend on $x'$: $\psi^i$ depends on $i$ and $s$, $\mu^i$ depends on $i, s, \cJ, \{\cJ_d\}_{d\in \cJ}$ and $\rr$ depends on $s, \cJ, \{\cJ_d\}_{d \in \cJ}$ and the noise vectors $\rz^{(1)}, \rz^{(2)}, \rz^{(3)}$. Applying Lemma~\ref{lem:releaseVar}, we conclude that for fixed $s, i, \cJ, \{\cJ_d\}_{d \in \cJ}$ as defined above

\begin{align*}
    \frac{\E[\langle \psi^i,\rr \rangle^2]}{\langle\psi^i, \mu^i \rangle^2} = \Omega( \min\{3^{-2} \eps^{-2}, 3^{-2}\delta^{-2}\}).
\end{align*}
Therefore, for fixed $s, \cJ, \{\cJ_d\}_{d \in \cJ}$
\begin{align}
\label{eq:connecting-inequality}
    \sum_{i \in \cL(\T^s)}\E[\langle \psi^i,\rr \rangle^2] = \Omega\left( \sum_{i \in \cL(\T^s)}\langle\psi^i, \mu^i \rangle^2 \min\{3^{-2} \eps^{-2}, 3^{-2}\delta^{-2}\}\right).
\end{align}

We now want to upper bound $\sum_{i \in \cL(\T^s)}\E[\langle \psi^i,\rr \rangle^2]$ and lower bound $\sum_{i \in \cL(\T^s)}\langle\psi^i, \mu^i \rangle^2$, which will give us a lower bound on $\sigma$. 

We start by upper bounding $\sum_{i \in \cL(\T^s)}\E[\langle \psi^i,\rr \rangle^2]$. Observe that the $\psi^i$ vectors are orthogonal and unit length.
\begin{lemma}
\label{lem:psiprop}
For $i \in \cL(\T^s)$ we have $\|\psi^i\|^2_2 = 1$ and 
for $i \neq j$ with $i,j \in \cL(\T^s)$, we have $\langle \psi^i, \psi^j \rangle = 0$.
\end{lemma}
This lemma follows directly from the definition of $\psi^i$ and the proof has been deferred to Appendix~\ref{sec:psiprop}. It follows from orthogonality and unit length that $\sum_{i\in \cL(\T^s)} \E[\langle \psi^i ,\rr \rangle^2] = \E[\sum_{i\in \cL(\T^s)} \langle \psi^i ,\rr \rangle^2] \leq \E[\|\rr\|_2^2] $. Now recall that $\sigma = \max_{d \in \cJ} \max_{u \in \cJ_d} \sigma_{u, \widetilde{s_{d+1}}}$. We observe that
\begin{align}
\label{eq:ub}
    \E[\|\rr\|_2^2] &= \sum_{d \in \cJ} \sum_{u \in \cJ_d} \E[\rr_u^2] \nonumber\\
    &\leq \sum_{d \in \cJ} \sum_{u \in \cJ_d} 3 \cdot 2^{-d} \cdot \sigma^2_{u, \widetilde{s_{d+1}}} \nonumber\\
    &\leq 3 \cdot \sum_{d \in \cJ} \sum_{u \in \cJ_d} 2^{-d} \cdot \sigma^2 \nonumber\\
    &= 3 \cdot \sum_{d \in \cJ} |\cJ_d| 2^{-d} \cdot \sigma^2 \nonumber\\
    &\leq 3 \cdot \sum_{d \in \cJ} 2^d\cdot 2^{-d} \cdot \sigma^2 \nonumber\\
    &= 3 \cdot |\cJ| \cdot \sigma^2 \nonumber\\
    &= (3/4) \cdot h \cdot \sigma^2.
\end{align}

Next, we lower bound $\sum_{i \in \cL(\T^s)}\langle\psi^i, \mu^i \rangle^2$. Define for convenience the notation $\Anc(i,d) \in \T$ as the ancestor of the $i$'th leaf at depth $2d$. Let us next consider
\begin{align}
\label{eq:ipsq}
    &\sum_{i \in \cL(\T^s)} \langle \psi^i, \mu^i \rangle^2 = \nonumber \\
    &\sum_{i \in \cL(\T^s)} \left(\sum_{d=0}^{h/2-1} \mu_{\Anc(i,d)}^i 2^{-(h/2-d)/2} (-1)^{\one\{i \in \cL(\T^{\ell(\Anc(i,d))})\}}\right)^2.
\end{align}
Let $u = \Anc(i,d)$ for short. Then
\begin{align*}
    &\mu^i_{u} (-1)^{\one\{i \in \cL(\T^{\ell(u)})\}} =\\
    &2^{-d/2} (-1)^{\one\{s_{d+1}=\ell\}} \one\{i \in \cL(\T^{s_{d+1}(\widetilde{s_{d+1}}(u))})\}(-1)^{\one\{i \in \cL(\T^{\ell(u)})\}}.
\end{align*}
Now consider first the case $s_{d+1}=\ell$. Then $\one\{i \in \cL(\T^{s_{d+1}(\widetilde{s_{d+1}}(u))})\} = 0$ if $i \in \cL(\T^{\ell(u)})$. If instead $i \in \cL(\T^{r(u)})$, we get 
\[
(-1)^{\one\{s_{d+1}=\ell\}} (-1)^{\one\{i \in \cL(\T^{\ell(u)})\}}= -1 \cdot 1 = -1.
\]
Similarly, if $s_{d+1}=r$, then we need only consider $i \in \cL(\T^{\ell(u)})$. Here we have
\[
(-1)^{\one\{s_{d+1}=\ell\}} (-1)^{\one\{i \in \cL(\T^{\ell(u)})\}}= 1 \cdot (-1) = -1.
\]
Thus in all circumstances, we may simplify
\begin{align*}
    &\mu^i_{u} (-1)^{\one\{i \in \cL(\T^{\ell(u)})\}} = -2^{-d/2} \one\{i \in \cL(\T^{s_{d+1}(\widetilde{s_{d+1}}(u))})\}.
\end{align*}
Inserting this in~\eqref{eq:ipsq}, we see that
\begin{align*}
    &\sum_{i \in \cL(\T^s)} \langle \psi^i, \mu^i \rangle^2 =  \\
    &\sum_{i \in \cL(\T^s)} \left(\sum_{d=0}^{h/2-1} -2^{-h/4}\one\{i \in \cL(\T^{s_{d+1}(\widetilde{s_{d+1}}(\Anc(i,d)))})\} \right)^2 =\\
    &2^{-h/2} \sum_{i \in \cL(\T^s)}  \left(\sum_{d=0}^{h/2-1} \one\{i \in \cL(\T^{s_{d+1}(\widetilde{s_{d+1}}(\Anc(i,d)))})\} \right)^2 = \\
    &2^{-h/2} \sum_{i \in \cL(\T^s)}  \sum_{d=0}^{h/2-1} \sum_{d'=0}^{h/2-1} \one\{i \in \cL(\T^{s_{d+1}(\widetilde{s_{d+1}}(\Anc(i,d)))})\}\one\{i \in \cL(\T^{s_{d'+1}(\widetilde{s_{d'+1}}(\Anc(i,d')))})\}  = \\
    &2^{-h/2} \sum_{d=0}^{h/2-1} \sum_{d'=0}^{h/2-1} \sum_{u \in \cJ_d} \sum_{v \in \cJ_{d'}} \sum_{i \in \cL(\T^s)}\one\{i \in \cL(\T^{s_{d+1}(\widetilde{s_{d+1}}(u))})\}\one\{i \in \cL(\T^{s_{d'+1}(\widetilde{s_{d'+1}}(v))})\}.
\end{align*}
Now define vectors $\phi^u$ for each $u \in \cup_{d \in \cJ} \cJ_d$. If $u$ is in depth $2d$, we let $\phi^u$ be the indicator vector with one coordinate for each $i \in \cL(\T^s)$ taking the value $\one\{i \in \cL(\T^{s_{d+1}(\widetilde{s_{d+1}}(u))})\}$. Then by the above we have
\begin{align*}
    &\sum_{i \in \cL(\T^s)} \langle \psi^i, \mu^i \rangle^2 =  2^{-h/2} \left\|\sum_{u \in \cup_{d \in \cJ} \cJ_d} \phi^u \right\|^2_2.
\end{align*}
Applying Cauchy-Schwartz, we first see that
\begin{align}
\label{eq:step1}
    2^{-h/2} \left\|\sum_{u \in \cup_{d \in \cJ} \cJ_d} \phi^u \right\|^2_2 \geq 
    &2^{-h/2} \left\|\sum_{u \in \cup_{d \in \cJ} \cJ_d} \phi^u \right\|^2_1 \cdot |\cL(\T^s)|^{-1}.
\end{align}
Now using that the vectors $\phi^u$ are indicator vectors we have that~\eqref{eq:step1} equals
\begin{align*}
    &2^{-h} \left(\sum_{u \in \cup_{d \in \cJ} \cJ_d} \|\phi^u\|_1  \right)^2 =
    2^{-h} \left(\sum_{d \in \cJ} \sum_{u \in \cJ_d} \|\phi^u\|_1  \right)^2 =
    2^{-h} \left(\sum_{d \in \cJ} \sum_{u \in \cJ_d} 2^{h/2-d-1}  \right)^2 =\\
    &2^{-h} \left(\sum_{d \in \cJ} 2^{d-1} 2^{h/2-d-1}  \right)^2 =
    2^{-h} \left(h \cdot 2^{h/2-4}  \right)^2 =
    2^{-8} h^2.
\end{align*}
Therefore, 
\begin{align}
    \sum_{i \in \cL(\T^s)} \langle \psi^i, \mu^i \rangle^2\geq 2^{-8} h^2.
    \label{eq:lb}
\end{align}

Using the upper bound of $\E[\|\rr\|^2]$~\eqref{eq:ub} and the lower bound of $\sum_{i \in \cL(\T^s)}\langle \psi^i, \mu^i\rangle^2$~\eqref{eq:lb} in inequality~\eqref{eq:connecting-inequality}, it follows that 
\[
 \frac{3 \cdot h\cdot \sigma^2}{4} \geq \Omega(h^2 \cdot 2^{-8} \min\{6^{-2} \eps^{-2}, (3e)^{-2}\delta^{-2}\}). 
\]
This gives a lower bound of
\[
\sigma^2 = \Omega(h \min\{\eps^{-2},\delta^{-2}\}) \Rightarrow \sigma = \Omega(\sqrt{h} \min\{\eps^{-1},\delta^{-1}\}).
\]
Inserting this in~\eqref{eq:lbhsigma} we conclude
\[
\max_{x \in \{0,1\}^n }\E[\|\rM(x)-Ax\|_\infty] \geq \E[\|\rz\|_\infty] = \Omega\left(\frac{h^{3/2}}{\max\{\eps,\delta\}} \right) = \Omega\left(\frac{\log^{3/2} n}{\max\{\eps,\delta\}} \right).
\]
This completes the proof of Theorem~\ref{thm:main}.

\section{Formal Lower Bound Extension to Pure Differential Privacy}
\label{sec:pure-dp}

We now prove a stronger lower bound for pure differential privacy. We use
the same construction as in Section~\ref{sec:formal-proof}, replacing only the privacy attack on $\widehat{\rM}_s$. This new attack, inspired by the randomized tree decoder of Ghazi et al.~\cite[Section~5.1]{GhaziEtAl2023Trees} for private tree aggregation, yields a packing argument for proving the lower bound.

\begin{restatethm}{\ref{thm:pure-dp}}
For any $n^{-1+\Omega(1)}<\eps<1$, it holds that any $\eps$-DP
mechanism $\rM$ for continual counting has error
\[
  \max_{x\in\{0,1\}^n}\E[\|\rM(x)-Ax\|_\infty]
  =\Omega\left(\frac{\log^2 n}{\eps}\right),
\]
where $A$ is the $n\times n$ prefix-sum matrix.
\end{restatethm}

We first consider an oblivious $(\eps,\delta)$-DP mechanism
$\rM_o(x)=Ax+\rz$. We allow a small positive $\delta$ because
Lemma~\ref{lem:oblivious} introduces such a term even when the original
mechanism is pure DP. Recall that $n$ is a power of $16$ and
$h=\log_2 n$. By the argument leading to~\eqref{eq:lbhsigma}, it
suffices to prove the following bound on the maximum $\sigma$ value.

\begin{lemma}
\label{lem:pure-selected-noise}
Suppose $0<\eps\leq2$ and $0\leq\delta\leq\eps/\sqrt n$.
For any $s\in\{\ell,r\}^{h/2}$, any
$\cJ\subseteq\{0,\dots,h/2-1\}$ with $|\cJ|=h/4$, and any
$\{\cJ_d\}_{d\in\cJ}$ with $\cJ_d$ consisting of half the nodes of
$\T^s$ at depth $2d$ for $d>0$ and $\cJ_0=\{\rho\}$ if $0\in\cJ$,
we have
\[
  \sigma:=\max_{d\in\cJ}\max_{u\in\cJ_d}
  \sigma_{u,\widetilde{s_{d+1}}}
  =\Omega\left(\frac{h}{\eps}\right)
\]
for sufficiently large $n$.
\end{lemma}

\begin{proof}
Fix $s$, $\cJ$, and $\{\cJ_d\}_{d\in\cJ}$ as above. We may assume
$\sigma<\infty$. Recall that $\widehat{\rM}_s(x)=\bq$ is
$(3\eps,3\delta)$-DP on binary inputs with $x_i=0$ for
$i\notin\cL(\T^s)$, and that Lemma~\ref{lem:noisebound} gives
\begin{equation}
  \E[\raq_u^2]\leq3\sigma^2
  \qquad\text{for every }d\in\cJ\text{ and }u\in\cJ_d.
  \label{eq:pure-second-moment}
\end{equation}

\paragraph{The One-Term and Three-Term Cases.}
Let $u\in\T^s$ be at depth $2d<h$. For an input supported on
$\cL(\T^s)$, the definition of the coefficients $a_u:= \sum_{i : u \in \I(i)} x_i,$ gives
\begin{align}
  a_{r(u)}
  &=\sum_{i\in\cL(\T^s)\cap\cL(\T^{r(\ell(u))})}x_i
  &&\text{if }s_{d+1}=r,\label{eq:pure-single-signal}\\
  a_{r(\ell(u))}-a_{r(u)}-a_{r(r(u))}
  &=-\sum_{i\in\cL(\T^s)\cap\cL(\T^{\ell(r(u))})}x_i
  &&\text{if }s_{d+1}=\ell.
  \label{eq:pure-three-signal}
\end{align}
To see this directly, let $t_1,t_2,t_3,t_4$ be the leftmost leaf
indices in the four grandchild subtrees of $u$, from left to right.
The coefficient of a right child is the difference between the prefix
sums at the leftmost leaves of that child and its parent. Thus
$a_{r(u)}=(Ax)_{t_3}-(Ax)_{t_1}$. When $s_{d+1}=r$, the only
possibly nonzero coordinates in this interval belong to
$\cL(\T^{r(\ell(u))})$, giving~\eqref{eq:pure-single-signal}.
For the other case, the three terms telescope:
\begin{align*}
  a_{r(\ell(u))}-a_{r(u)}-a_{r(r(u))}
  &=\bigl((Ax)_{t_2}-(Ax)_{t_1}\bigr)
    -\bigl((Ax)_{t_3}-(Ax)_{t_1}\bigr)
    -\bigl((Ax)_{t_4}-(Ax)_{t_3}\bigr)\\
  &=(Ax)_{t_2}-(Ax)_{t_4}
   =-\sum_{i=t_2+1}^{t_4}x_i.
\end{align*}
When $s_{d+1}=\ell$, the only possibly nonzero coordinates in this
last interval belong to $\cL(\T^{\ell(r(u))})$, giving
\eqref{eq:pure-three-signal}. In particular, in either case
\begin{equation}
  (-1)^{\one\{s_{d+1}=\ell\}}\bq_u
  =\sum_{i\in\cL(\T^s)\cap
      \cL(\T^{s_{d+1}(\widetilde{s_{d+1}}(u))})}x_i
    +(-1)^{\one\{s_{d+1}=\ell\}}\raq_u.
  \label{eq:pure-signed-count}
\end{equation}

We use the same additive-noise formula to define
$\widehat{\rM}_s$, but now on integer inputs supported on $\cL(\T^s)$, rather than only binary inputs.
The integer-input extension justified in Appendix~\ref{app:oblivious}
applies here as well: the signal in~\eqref{eq:pure-signed-count} is
linear in $x$, and the distribution of $\raq$ is independent of $x$.
Therefore, for any $w\in \mathbb{Z}^n$ and any $i \in \cL(\T^s)$, the pair of output distributions of $\widehat{\rM}_s$ on neighboring inputs $w$ and $w+e_i$ is a deterministic shift, depending only on $w$, of the pair of output distributions on inputs $0$ and $e_i$.
Since a deterministic shift applied to both members of a pair of distributions preserves the DP guarantee between them, the extended mechanism
is $(3\eps,3\delta)$-DP for integer inputs at $\ell_1$ distance $1$. Moreover, 
 the same noise bound~\eqref{eq:pure-second-moment} continues to hold, because the distribution of $\raq$ does not depend on the input. We can thus
apply group privacy along inputs $0,e_i,\dots,De_i$ for any positive integer
$D$. These are inputs to the extended oblivious mechanism; the
original mechanism is evaluated only on binary inputs, as explained
in the appendix.

\paragraph{Selecting Leaves.}
If we want to write a privacy attack for $\widehat{\rM}_s$, the nodes for which we have small noise guarantees when $\sigma$ is small for the given $\cJ$ and $\{\cJ_d\}_{d\in \cJ}$ are the nodes in $\bigcup_{d\in\cJ}\cJ_d$. For a uniformly random $\ri\in\cL(\T^s)$, we have
\[
  \E_{\ri}\left[
    \left|\pAnc(\ri)\cap\bigcup_{d\in\cJ}\cJ_d\right|
  \right] = \sum_{d\in \cJ} \sum_{u \in \cJ_d} \Pr[u \in \pAnc(i)]
  =\sum_{d\in\cJ}|\cJ_d|2^{-d}\geq\frac h8,
\]
because $|\cJ_d| \geq 2^{d-1}$.
The number of ancestors in $\bigcup_{d\in\cJ}\cJ_d$ is at most $|\cJ|=h/4$.
By Markov's inequality applied to
$h/4-|\pAnc(\ri)\cap\bigcup_{d\in\cJ}\cJ_d|$, 
\begin{align*}
    \Pr\left[\frac{h}{4}- \left|\pAnc(\ri)\cap\bigcup_{d\in\cJ}\cJ_d\right|\geq \frac{3h}{16}\right] \leq \frac{2}{3}.
\end{align*}
Therefore, at least a third,
and hence at least a quarter, of the leaves $i\in\cL(\T^s)$ satisfy
\begin{equation}
  \left|\pAnc(i)\cap\bigcup_{d\in\cJ}\cJ_d\right|\geq\frac h{16}.
  \label{eq:pure-many-selected}
\end{equation}
We only attack input coordinates corresponding to these leaves.

\paragraph{Privacy Attack on $\widehat{\rM}_s$.}
Let $D=\max\{1,\lceil16\sigma\rceil\}$. For the packing argument, we consider the following collection of inputs
\[\mathcal{C}:= \{0\}\cup \big\{De_i\}_{i \in \cL(\T^s)}.\] 
We construct a decoder, described in Algorithm~\ref{alg:decoder}, that given $\widehat{\rM}_s(x)=\bq$, returns an index $i \in \cL(\T^s)$.
The decoder aims to recover $i$  when the input is $De_i$, with high probability, for every $i$ that satisfies~\eqref{eq:pure-many-selected}. Denote
its output by $\operatorname{Dec}(\widehat{\rM}_s(x))$.

\begin{algorithm}[ht]
    \caption{Decoder for $\widehat{\rM}_s$}
    \label{alg:decoder}
    \KwIn{$s \in \{\ell, r\}^{h/2}$, $\cJ$, $\{\cJ_d\}_{d\in \cJ}$, $x'$, $\widehat{\rM}_s(\rx) = \bq$}
    \KwOut{$i \in \cL(\T^s)$}
    $u \gets \rho$\;
    \For{$d = 0$ \KwTo $h/2-1$}{
        \eIf {$u\notin\bigcup_{d'\in\cJ}\cJ_{d'}$}{
            $u \gets s_{d+1}(\ell(u))$ or $s_{d+1}(r(u))$, uniformly at random\;
        }
        {
            \eIf{$(-1)^{\one\{s_{d+1}=\ell\}}\bq_u\geq D/2$}{
                $u\gets s_{d+1}(\widetilde{s_{d+1}}(u))$ with probability $3/4$ or $s_{d+1}(s_{d+1}(u))$ with probability $1/4$\;
            }
            {
                $u\gets s_{d+1}(\widetilde{s_{d+1}}(u))$ with probability $1/4$ or $s_{d+1}(s_{d+1}(u))$ with probability $3/4$\;
            }
        }
    }
    \Return $i$ that corresponds to $u$\;
\end{algorithm}

Fix an $i$ satisfying~\eqref{eq:pure-many-selected} and consider
the input $De_i$. By~\eqref{eq:pure-signed-count}, at an
ancestor $u$ of $i$ that is in $\{\cJ_d\}_{d \in \cJ}$ the count being estimated is either $D$ or $0$,
depending on which of the two choices ($s_{d+1}(\widetilde{s_{d+1}}(u))$ or $s_{d+1}(s_{d+1}(u))$) leads to $i$. Hence, the decoder
assigns probability $3/4$ to the wrong choice only if
$|\raq_u|\geq D/2$. By~\eqref{eq:pure-second-moment} and Markov's
inequality,
\[
  \Pr[|\raq_u|\geq D/2]
  \leq\frac{12\sigma^2}{D^2}\leq\frac3{64}<\frac1{16}.
\]
For a fixed noise vector $\raq$,  let $B_i$ be the number of ancestors of $i$ in $ \cup_{d\in \cJ}\cJ_d$ at which the
decoder would assign probability $3/4$ to the wrong choice, whether or not it visits that node. Then,
\begin{align*}
    \E[B_i]&\leq \sum_{u \in \pAnc(i)\cap\bigcup_{d\in\cJ}\cJ_d} \Pr\left[\text{the decoder assigns probability }3/4\text{ to the wrong choice at }u\right]\\
    &\leq \sum_{u \in \pAnc(i)\cap\bigcup_{d\in\cJ}\cJ_d} \Pr\left[|\raq_u| \geq D/2\right] \\
    &\leq \frac1{16}
    \left|\pAnc(i)\cap\bigcup_{d\in\cJ}\cJ_d\right|
\end{align*}

Conditioned on $\raq$, the decoder's path is a sequence of $h/2$ independent random choices. Therefore, the probability that the decoder returns $i$ is a product over the $h/2$ rounds of the probability that the decoder moves toward $i$ at that round. For each round $d$, let $u$ be $i$'s ancestor visited at that round. If $u \notin \bigcup_{d'\in\cJ}\cJ_{d'}$, then the decoder picks the correct option with probability $1/2$. If $u \in \bigcup_{d'\in\cJ}\cJ_{d'}$, the decoder picks the correct option with probability $3/4$ if it has assigned probability $3/4$ to the correct choice based on $\raq_u$, and with probability $1/4$ if instead it has assigned probability $3/4$ to the wrong choice. Conditioned on $\raq$, the probability that the decoder returns $i$ is
\begin{align*}
  \Pr[\operatorname{Dec}(\widehat{\rM}_s(De_i))=i\mid \raq] & = \left(\frac{1}{2}\right)^{h/2- \left|\pAnc(i)\cap\bigcup_{d\in\cJ}\cJ_d\right|} \left(\frac{3}{4}\right)^{\left|\pAnc(i)\cap\bigcup_{d\in\cJ}\cJ_d\right|-B_i} \left(\frac{1}{4}\right)^{B_i}\\
  &= 2^{-h/2}(3/2)^{|\pAnc(i)\cap\bigcup_{d\in\cJ}\cJ_d|}3^{-B_i}.
\end{align*}
Since $t\mapsto3^{-t}$ is convex, Jensen's
inequality now gives
\begin{align}
  \Pr[\operatorname{Dec}(\widehat{\rM}_s(De_i))=i]
  & = \E\left[\Pr[\operatorname{Dec}(\widehat{\rM}_s(De_i))=i\mid \raq]\right]\geq2^{-h/2}(3/2)^{|\pAnc(i)\cap\bigcup_{d\in\cJ}\cJ_d|}
       3^{-\E[B_i]}\nonumber\\
  &\geq  2^{-h/2}(3/2)^{|\pAnc(i)\cap\bigcup_{d\in\cJ}\cJ_d|}3^{-|\pAnc(i)\cap\bigcup_{d\in\cJ}\cJ_d|/16}\nonumber\\
  &\geq2^{-h/2}\exp\left(
       \frac14\left|\pAnc(i)\cap\bigcup_{d\in\cJ}\cJ_d\right|
       \right)
   \geq2^{-h/2}e^{h/64},
  \label{eq:pure-decoding-success}
\end{align}
where we used $\ln(3/2)-(\ln3)/16>1/4$ and
\eqref{eq:pure-many-selected}. No independence between the
coordinates of $\raq$ is needed.

On the other hand, group privacy between $De_i$ and $0$ gives
\begin{align*}
  \Pr[\operatorname{Dec}(\widehat{\rM}_s(De_i))=i]
  &\leq e^{3\eps D}
       \Pr[\operatorname{Dec}(\widehat{\rM}_s(0))=i]
       +3\delta\sum_{j=0}^{D-1}e^{3\eps j}\\
  &\leq e^{3\eps D}\left(
       \Pr[\operatorname{Dec}(\widehat{\rM}_s(0))=i]+3\delta D
       \right).
\end{align*}
Let $S_{\cL}$ be the set of leaves that satisfy~\eqref{eq:pure-many-selected}. Summing the probabilities of correct decoding over $i \in S_{\cL}$, we get
\[
 2^{-h/2}e^{h/64}|S_{\cL}|\leq \sum_{i \in S_{\cL}}\Pr[\operatorname{Dec}(\widehat{\rM}_s(De_i))=i] \leq e^{3\eps D}\left(\sum_{i \in S_{\cL}}
       \Pr[\operatorname{Dec}(\widehat{\rM}_s(0))=i]+3\delta D |S_{\cL}|
       \right).
\]
Since the decoder outputs a single leaf on any given run, the events $\{\operatorname{Dec}(\widehat{\rM}_s(0))=i\}$
for these leaves $i$ are disjoint. Therefore, 
\[
 e^{h/64} \leq e^{3\eps D}\left(\frac{2^{h/2}}{|S_{\cL}|}+3\delta D2^{h/2}\right).
\]
There are at least $2^{h/2}/4$ such leaves. 
This implies that
\[
  \frac14e^{h/64}
  \leq e^{3\eps D}\left(1+3\delta D2^{h/2}\right).
\]
Since $\delta2^{h/2}\leq\eps$, $\ln(1+t)\leq t$ for $t\geq0$,
and $D\leq16\sigma+1$, we conclude
\[
  \frac h{64}-\ln4
  \leq3\eps D+\ln(1+3\eps D)
  \leq6\eps D
  \leq96\eps\sigma+6\eps.
\]
Thus $\sigma=\Omega(h/\eps)$ for sufficiently large $h$, as claimed.
\end{proof}

\begin{proof}[Proof of Theorem~\ref{thm:pure-dp}]
Apply Lemma~\ref{lem:oblivious} to the original pure $\eps$-DP
mechanism with $\delta'=1/n^2$. This gives an oblivious
$(2\eps,\delta')$-DP mechanism which, by the integer-input extension
in Appendix~\ref{app:oblivious}, is defined on
\[
  m=\left\lfloor
    \frac{n}{2\lceil2\ln(n)/\eps\rceil}
  \right\rfloor
\]
integer coordinates, with expected $\ell_\infty$ error at most the
worst-case error of the original mechanism on Boolean inputs.
Since $\eps>n^{-1+\Omega(1)}$, we have
$m=n^{\Omega(1)}$. As before, we may restrict to a power-of-$16$
prefix, losing only a constant factor in the length. Moreover,
$\delta'=1/n^2\leq2\eps/\sqrt m$. Applying
Lemma~\ref{lem:pure-selected-noise} with privacy parameters
$(2\eps,\delta')$ and using the reduction already established
in~\eqref{eq:lbhsigma}, we obtain
\[
  \max_{x\in\{0,1\}^n}\E[\|\rM(x)-Ax\|_\infty]
  =\Omega\left(\frac{\log^2 m}{\eps}\right)
  =\Omega\left(\frac{\log^2 n}{\eps}\right).
\]
This completes the proof.
\end{proof}

\section*{AI Usage Disclosure}
All results in this paper were obtained solely by the authors. An earlier manuscript contained only the approximate-DP result and was also written without AI. We were in the process of writing the pure-DP result when we became aware of the alternative proof of~\cite{HarrisonLeeman2026}. We thus gave ChatGPT 6.0 Astra our unpolished draft and asked it to make the proof and notation consistent with our manuscript for the approximate-DP setting and combine the two.

\bibliography{refs}

@inproceedings{mediantrick,
  author       = {Kasper Green Larsen and
                  Rasmus Pagh and
                  Jakub Tetek},
  title        = {CountSketches, Feature Hashing and the Median of Three},
  booktitle    = {{ICML}},
  series       = {Proceedings of Machine Learning Research},
  volume       = {139},
  pages        = {6011--6020},
  publisher    = {{PMLR}},
  year         = {2021}
}

@article{GengViswanath2016ApproxDPNoise,
  author  = {Quan Geng and Pramod Viswanath},
  title   = {Optimal Noise Adding Mechanisms for Approximate Differential Privacy},
  journal = {IEEE Transactions on Information Theory},
  volume  = {62},
  number  = {2},
  pages   = {952--969},
  year    = {2016}
}

@inproceedings{UpperBound1,
  author       = {Cynthia Dwork and
                  Moni Naor and
                  Toniann Pitassi and
                  Guy N. Rothblum},
  editor       = {Leonard J. Schulman},
  title        = {Differential privacy under continual observation},
  booktitle    = {Proceedings of the 42nd {ACM} Symposium on Theory of Computing, {STOC}
                  2010, Cambridge, Massachusetts, USA, 5-8 June 2010},
  pages        = {715--724},
  publisher    = {{ACM}},
  year         = {2010}
}

@inproceedings{UpperBound2,
  author       = {T.{-}H. Hubert Chan and
                  Elaine Shi and
                  Dawn Song},
  editor       = {Samson Abramsky and
                  Cyril Gavoille and
                  Claude Kirchner and
                  Friedhelm Meyer auf der Heide and
                  Paul G. Spirakis},
  title        = {Private and Continual Release of Statistics},
  booktitle    = {Automata, Languages and Programming, 37th International Colloquium,
                  {ICALP} 2010, Bordeaux, France, July 6-10, 2010, Proceedings, Part
                  {II}},
  series       = {Lecture Notes in Computer Science},
  pages        = {405--417},
  publisher    = {Springer},
  year         = {2010}
}

@inproceedings{l2LowerBound,
  author       = {Monika Henzinger and
                  Jalaj Upadhyay and
                  Sarvagya Upadhyay},
  editor       = {Nikhil Bansal and
                  Viswanath Nagarajan},
  title        = {Almost Tight Error Bounds on Differentially Private Continual Counting},
  booktitle    = {Proceedings of the 2023 {ACM-SIAM} Symposium on Discrete Algorithms,
                  {SODA} 2023, Florence, Italy, January 22-25, 2023},
  pages        = {5003--5039},
  publisher    = {{SIAM}},
  year         = {2023}
}

@inproceedings{SparseLowerBound,
  author       = {Edith Cohen and
                  Xin Lyu and
                  Jelani Nelson and
                  Tam{\'{a}}s Sarl{\'{o}}s and
                  Uri Stemmer},
  editor       = {Shipra Agrawal and
                  Aaron Roth},
  title        = {Lower Bounds for Differential Privacy Under Continual Observation
                  and Online Threshold Queries},
  booktitle    = {The Thirty Seventh Annual Conference on Learning Theory, June 30 -
                  July 3, 2023, Edmonton, Canada},
  series       = {Proceedings of Machine Learning Research},
  pages        = {1200--1222},
  publisher    = {{PMLR}},
  year         = {2024}
}

@inproceedings{UpperBoundApprox,
  author       = {Palak Jain and
                  Sofya Raskhodnikova and
                  Satchit Sivakumar and
                  Adam D. Smith},
  editor       = {Andreas Krause and
                  Emma Brunskill and
                  Kyunghyun Cho and
                  Barbara Engelhardt and
                  Sivan Sabato and
                  Jonathan Scarlett},
  title        = {The Price of Differential Privacy under Continual Observation},
  booktitle    = {International Conference on Machine Learning, {ICML} 2023, 23-29 July
                  2023, Honolulu, Hawaii, {USA}},
  series       = {Proceedings of Machine Learning Research},
  pages        = {14654--14678},
  publisher    = {{PMLR}},
  year         = {2023},
  url          = {https://proceedings.mlr.press/v202/jain23b.html}
}

@inproceedings{ObliviousLinear,
author = {Bhaskara, Aditya and Dadush, Daniel and Krishnaswamy, Ravishankar and Talwar, Kunal},
title = {Unconditional differentially private mechanisms for linear queries},
year = {2012},
isbn = {9781450312455},
publisher = {Association for Computing Machinery},
address = {New York, NY, USA},
booktitle = {Proceedings of the Forty-Fourth Annual ACM Symposium on Theory of Computing},
pages = {1269–1284},
numpages = {16},
location = {New York, New York, USA},
series = {STOC '12}
}

@inproceedings{PureDPDefinition,
  author       = {Cynthia Dwork and
                  Frank McSherry and
                  Kobbi Nissim and
                  Adam D. Smith},
  editor       = {Shai Halevi and
                  Tal Rabin},
  title        = {Calibrating Noise to Sensitivity in Private Data Analysis},
  booktitle    = {Theory of Cryptography, Third Theory of Cryptography Conference, {TCC}
                  2006, New York, NY, USA, March 4-7, 2006, Proceedings},
  series       = {Lecture Notes in Computer Science},
  pages        = {265--284},
  publisher    = {Springer},
  year         = {2006}
}

@inproceedings{ConstantUB,
  author       = {Hendrik Fichtenberger and
                  Monika Henzinger and
                  Jalaj Upadhyay},
  editor       = {Andreas Krause and
                  Emma Brunskill and
                  Kyunghyun Cho and
                  Barbara Engelhardt and
                  Sivan Sabato and
                  Jonathan Scarlett},
  title        = {Constant Matters: Fine-grained Error Bound on Differentially Private
                  Continual Observation},
  booktitle    = {International Conference on Machine Learning, {ICML} 2023, 23-29 July
                  2023, Honolulu, Hawaii, {USA}},
  series       = {Proceedings of Machine Learning Research},
  pages        = {10072--10092},
  publisher    = {{PMLR}},
  year         = {2023}
}

@inproceedings{PrivateSCO,
  author       = {Yuxuan Han and
                  Zhicong Liang and
                  Zhipeng Liang and
                  Yang Wang and
                  Yuan Yao and
                  Jiheng Zhang},
  editor       = {Kamalika Chaudhuri and
                  Stefanie Jegelka and
                  Le Song and
                  Csaba Szepesv{\'{a}}ri and
                  Gang Niu and
                  Sivan Sabato},
  title        = {Private Streaming {SCO} in $\ell_p$ geometry
                  with Applications in High Dimensional Online Decision Making},
  booktitle    = {International Conference on Machine Learning, {ICML} 2022, 17-23 July
                  2022, Baltimore, Maryland, {USA}},
  series       = {Proceedings of Machine Learning Research},
  pages        = {8249--8279},
  publisher    = {{PMLR}},
  year         = {2022},
  url          = {https://proceedings.mlr.press/v162/han22d.html}
}

@inproceedings{DPFTRL,
  author       = {Abhradeep Guha Thakurta and
                  Adam D. Smith},
  editor       = {Christopher J. C. Burges and
                  L{\'{e}}on Bottou and
                  Zoubin Ghahramani and
                  Kilian Q. Weinberger},
  title        = {(Nearly) Optimal Algorithms for Private Online Learning in Full-information
                  and Bandit Settings},
  booktitle    = {Advances in Neural Information Processing Systems 26: 27th Annual
                  Conference on Neural Information Processing Systems 2013. Proceedings
                  of a meeting held December 5-8, 2013, Lake Tahoe, Nevada, United States},
  pages        = {2733--2741},
  year         = {2013},
  url          = {https://proceedings.neurips.cc/paper/2013/hash/c850371fda6892fbfd1c5a5b457e5777-Abstract.html}
}

@inproceedings{PrivateLearningwDPFTRL,
  author       = {Peter Kairouz and
                  Brendan McMahan and
                  Shuang Song and
                  Om Thakkar and
                  Abhradeep Thakurta and
                  Zheng Xu},
  editor       = {Marina Meila and
                  Tong Zhang},
  title        = {Practical and Private (Deep) Learning Without Sampling or Shuffling},
  booktitle    = {Proceedings of the 38th International Conference on Machine Learning,
                  {ICML} 2021, 18-24 July 2021, Virtual Event},
  series       = {Proceedings of Machine Learning Research},
  pages        = {5213--5225},
  publisher    = {{PMLR}},
  year         = {2021},
  url          = {http://proceedings.mlr.press/v139/kairouz21b.html}
}

@inproceedings{GradualPrivacy,
  author       = {Joel Daniel Andersson and
                  Monika Henzinger and
                  Rasmus Pagh and
                  Teresa Anna Steiner and
                  Jalaj Upadhyay},
  editor       = {Amir Globersons and
                  Lester Mackey and
                  Danielle Belgrave and
                  Angela Fan and
                  Ulrich Paquet and
                  Jakub M. Tomczak and
                  Cheng Zhang},
  title        = {Continual Counting with Gradual Privacy Expiration},
  booktitle    = {Advances in Neural Information Processing Systems 37: Annual Conference
                  on Neural Information Processing Systems 2024, NeurIPS 2024, Vancouver,
                  BC, Canada, December 10 - 15, 2024},
  year         = {2024},
  url          = {http://papers.nips.cc/paper\_files/paper/2024/hash/11faf17bf7e5412d9cded369f97db23d-Abstract-Conference.html}
}

@inproceedings{FrameworkPrivateSums,
  author       = {Monika Henzinger and
                  Jalaj Upadhyay and
                  Sarvagya Upadhyay},
  editor       = {David P. Woodruff},
  title        = {A Unifying Framework for Differentially Private Sums under Continual
                  Observation},
  booktitle    = {Proceedings of the 2024 {ACM-SIAM} Symposium on Discrete Algorithms,
                  {SODA} 2024, Alexandria, VA, USA, January 7-10, 2024},
  pages        = {995--1018},
  publisher    = {{SIAM}},
  year         = {2024}
}

@article{ImprovedPure,
  author       = {Joel Daniel Andersson and
                  Rasmus Pagh and
                  Sahel Torkamani},
  title        = {Improved Counting under Continual Observation with Pure Differential
                  Privacy},
  journal      = {CoRR},
  volume       = {abs/2408.07021},
  year         = {2024},
  url          = {https://doi.org/10.48550/arXiv.2408.07021},
  doi          = {10.48550/ARXIV.2408.07021},
  eprinttype   = {arXiv},
  eprint       = {2408.07021}
}

@inproceedings{GeometrySparse,
  author       = {Aleksandar Nikolov and
                  Kunal Talwar and
                  Li Zhang},
  editor       = {Dan Boneh and
                  Tim Roughgarden and
                  Joan Feigenbaum},
  title        = {The geometry of differential privacy: the sparse and approximate cases},
  booktitle    = {Symposium on Theory of Computing Conference, STOC'13, Palo Alto, CA,
                  USA, June 1-4, 2013},
  pages        = {351--360},
  publisher    = {{ACM}},
  year         = {2013}
}

@inproceedings{ContinualGraphs,
  author       = {Hendrik Fichtenberger and
                  Monika Henzinger and
                  Lara Ost},
  editor       = {Petra Mutzel and
                  Rasmus Pagh and
                  Grzegorz Herman},
  title        = {Differentially Private Algorithms for Graphs Under Continual Observation},
  booktitle    = {29th Annual European Symposium on Algorithms, {ESA} 2021, Lisbon,
                  Portugal (Virtual Conference), September 6-8, 2021},
  series       = {LIPIcs},
  volume       = {204},
  pages        = {42:1--42:16},
  publisher    = {Schloss Dagstuhl - Leibniz-Zentrum f{\"{u}}r Informatik},
  year         = {2021}
}

@inproceedings{CorrelatedNoise,
  author       = {Christopher A. Choquette{-}Choo and
                  Krishnamurthy Dj Dvijotham and
                  Krishna Pillutla and
                  Arun Ganesh and
                  Thomas Steinke and
                  Abhradeep Guha Thakurta},
  title        = {Correlated Noise Provably Beats Independent Noise for Differentially
                  Private Learning},
  booktitle    = {The Twelfth International Conference on Learning Representations,
                  {ICLR} 2024, Vienna, Austria, May 7-11, 2024},
  publisher    = {OpenReview.net},
  year         = {2024}
}

@article{HereditaryDiscrepancy,
    author = {Jir{\'{\i}} Matousek and
                  Aleksandar Nikolov and
                  Kunal Talwar},
    title = {Factorization Norms and Hereditary Discrepancy},
    journal = {International Mathematics Research Notices},
    volume = {2020},
    number = {3},
    pages = {751-780},
    year = {2020},
    month = {02},
    issn = {1073-7928},
    doi = {10.1093/imrn/rny033},
    url = {https://doi.org/10.1093/imrn/rny033},
    eprint = {https://academic.oup.com/imrn/article-pdf/2020/3/751/32219164/rny033.pdf},
}

@article{EfficientBinary,
  title={Efficient use of differentially private binary trees},
  author={Honaker, James},
  journal={Theory and Practice of Differential Privacy (TPDP 2015), London, UK},
  volume={2},
  pages={26--27},
  year={2015}
}

@inproceedings{RectangleQueries,
  author       = {Cynthia Dwork and
                  Moni Naor and
                  Omer Reingold and
                  Guy N. Rothblum},
  editor       = {Tetsu Iwata and
                  Jung Hee Cheon},
  title        = {Pure Differential Privacy for Rectangle Queries via Private Partitions},
  booktitle    = {Advances in Cryptology - {ASIACRYPT} 2015 - 21st International Conference
                  on the Theory and Application of Cryptology and Information Security,
                  Auckland, New Zealand, November 29 - December 3, 2015, Proceedings,
                  Part {II}},
  series       = {Lecture Notes in Computer Science},
  volume       = {9453},
  pages        = {735--751},
  publisher    = {Springer},
  year         = {2015}
}

@inproceedings{CountingDistinct,
  author       = {Palak Jain and
                  Iden Kalemaj and
                  Sofya Raskhodnikova and
                  Satchit Sivakumar and
                  Adam Smith},
  editor       = {Alice Oh and
                  Tristan Naumann and
                  Amir Globerson and
                  Kate Saenko and
                  Moritz Hardt and
                  Sergey Levine},
  title        = {Counting Distinct Elements in the Turnstile Model with Differential
                  Privacy under Continual Observation},
  booktitle    = {Advances in Neural Information Processing Systems 36: Annual Conference
                  on Neural Information Processing Systems 2023, NeurIPS 2023, New Orleans,
                  LA, USA, December 10 - 16, 2023},
  year         = {2023}
}

@inproceedings{DistinctExtension,
  author       = {Monika Henzinger and
                  A. R. Sricharan and
                  Teresa Anna Steiner},
  editor       = {Amit Kumar and
                  Noga Ron{-}Zewi},
  title        = {Private Counting of Distinct Elements in the Turnstile Model and Extensions},
  booktitle    = {Approximation, Randomization, and Combinatorial Optimization. Algorithms
                  and Techniques, {APPROX/RANDOM} 2024, London School of Economics,
                  London, UK, August 28-30, 2024},
  series       = {LIPIcs},
  volume       = {317},
  pages        = {40:1--40:21},
  publisher    = {Schloss Dagstuhl - Leibniz-Zentrum f{\"{u}}r Informatik},
  year         = {2024}
}

@inproceedings{EfficientDistinct,
  author       = {Rachel Cummings and
                  Alessandro Epasto and
                  Jieming Mao and
                  Tamalika Mukherjee and
                  Tingting Ou and
                  Peilin Zhong},
  editor       = {Aarti Singh and
                  Maryam Fazel and
                  Daniel Hsu and
                  Simon Lacoste{-}Julien and
                  Felix Berkenkamp and
                  Tegan Maharaj and
                  Kiri Wagstaff and
                  Jerry Zhu},
  title        = {Differentially Private Space-Efficient Algorithms for Counting Distinct
                  Elements in the Turnstile Model},
  booktitle    = {Forty-second International Conference on Machine Learning, {ICML}
                  2025, Vancouver, BC, Canada, July 13-19, 2025},
  series       = {Proceedings of Machine Learning Research},
  volume       = {267},
  publisher    = {{PMLR} / OpenReview.net},
  year         = {2025}
}

@article{Graphs,
  author       = {Sofya Raskhodnikova and
                  Teresa Anna Steiner},
  title        = {Fully Dynamic Algorithms for Graph Databases with Edge Differential
                  Privacy},
  journal      = {Proc. {ACM} Manag. Data},
  volume       = {3},
  number       = {2},
  pages        = {99:1--99:28},
  year         = {2025}
}

@inproceedings{DPOnline,
  author       = {Prateek Jain and
                  Pravesh Kothari and
                  Abhradeep Thakurta},
  editor       = {Shie Mannor and
                  Nathan Srebro and
                  Robert C. Williamson},
  title        = {Differentially Private Online Learning},
  booktitle    = {{COLT} 2012 - The 25th Annual Conference on Learning Theory, June
                  25-27, 2012, Edinburgh, Scotland},
  series       = {{JMLR} Proceedings},
  volume       = {23},
  pages        = {24.1--24.34},
  publisher    = {JMLR.org},
  year      = {2012}
  }

@inproceedings{DPFTRL2,
  author       = {Naman Agarwal and
                  Karan Singh},
  editor       = {Doina Precup and
                  Yee Whye Teh},
  title        = {The Price of Differential Privacy for Online Learning},
  booktitle    = {Proceedings of the 34th International Conference on Machine Learning,
                  {ICML} 2017, Sydney, NSW, Australia, 6-11 August 2017},
  series       = {Proceedings of Machine Learning Research},
  volume       = {70},
  pages        = {32--40},
  publisher    = {{PMLR}},
  year         = {2017}
}

@inproceedings{PrivateSCOl1,
  author       = {Hilal Asi and
                  Vitaly Feldman and
                  Tomer Koren and
                  Kunal Talwar},
  editor       = {Marina Meila and
                  Tong Zhang},
  title        = {Private Stochastic Convex Optimization: Optimal Rates in {L1} Geometry},
  booktitle    = {Proceedings of the 38th International Conference on Machine Learning,
                  {ICML} 2021, 18-24 July 2021, Virtual Event},
  series       = {Proceedings of Machine Learning Research},
  volume       = {139},
  pages        = {393--403},
  publisher    = {{PMLR}},
  year         = {2021},
  url          = {http://proceedings.mlr.press/v139/asi21b.html}
}

@inproceedings{GhaziEtAl2023Trees,
  author        = {Badih Ghazi and Pritish Kamath and Ravi Kumar and Pasin Manurangsi and Kewen Wu},
  title         = {On Differentially Private Counting on Trees},
  booktitle     = {50th International Colloquium on Automata, Languages, and Programming (ICALP 2023)},
  editor        = {Kousha Etessami and Uriel Feige and Gabriele Puppis},
  series        = {Leibniz International Proceedings in Informatics (LIPIcs)},
  volume        = {261},
  pages         = {66:1--66:18},
  publisher     = {Schloss Dagstuhl -- Leibniz-Zentrum f{\"u}r Informatik},
  year          = {2023},
  doi           = {10.4230/LIPIcs.ICALP.2023.66},
  url           = {https://drops.dagstuhl.de/entities/document/10.4230/LIPIcs.ICALP.2023.66}
}

@misc{HarrisonLeeman2026,
  author        = {Charlie Harrison and Ethan Leeman},
  title         = {Tight Lower Bounds for Differentially Private Continual Counting},
  year          = {2026},
  eprint        = {2609.17650},
  archivePrefix = {arXiv},
  primaryClass  = {cs.DS},
  note          = {\href{https://arxiv.org/abs/2609.17650}{arXiv:2609.17650}},
  url           = {https://arxiv.org/abs/2609.17650}
}
\bibliographystyle{alpha} 

\appendix
\section{Reduction to Oblivious Mechanisms}
\label{app:oblivious}
\begin{proof}[Proof of Lemma~\ref{lem:oblivious}]
Assume $\rM$ is an $(\eps,\delta)$-DP mechanism (possibly with data dependent noise) for prefix sums on inputs $x \in \{0,1\}^n$ and define 
\[
\err_\infty(\rM) = \max_{x \in \{0,1\}^n} \E[\|\rM(x)-Ax\|_\infty].
\]
Let $R \in [n]$ be an integer such that $2R$ divides $n$ and let $m = n/(2R)$ for short. To avoid confusion, for $k \in \mathbb{N}$ let $A^{(k)} \in \{0,1\}^{k\times k}$ denote the lower-triangular all-ones matrix, so $A^{(k)}z$ gives the prefix sums of $z \in \R^k$. To be consistent with the rest of the paper we write $A := A^{(n)}$. We first argue that we can construct an $(\eps,\delta)$-DP mechanism $\rM'$ for prefix sums on inputs $x \in \{-R,\dots,R\}^{m}$ with $\max_{x \in \{-R,\dots,R\}^{m}} \E[\|\rM'(x)-A^{(m)}x\|_\infty] \leq \err_\infty(\rM)$. Note that we define two inputs in $\{-R,\dots,R\}^{m}$ to be neighboring if they differ by $1$ in exactly one coordinate and are equal in all other coordinates.

The construction goes as follows. Given a vector $x \in \{-R,\dots,R\}^{m}$, map it to the vector $x' \in \{0,1\}^n$ where each $x_i$ is replaced by a consecutive block of $2R$ bits in $x'$. The block starts with $x_i+R$ many $1$'s, followed by $2R-(x_i+R) = R-x_i$ many $0$'s. Observe that two neighboring vectors in $\{-R,\dots,R\}^{m}$ map to neighboring vectors in $\{0,1\}^n$. We thus have that $\rM'$ is $(\eps,\delta)$-DP. To answer a prefix sum query for the $i$'th prefix, with $i \in [m]$, we use $\rM(x')$ and compute the $i \cdot 2R$'th prefix sum on $x'$ and subtract $i\cdot R$ from it. We can thus via post-processing of $\rM(x')$ obtain $\rM'(x)$ satisfying $\E[\|\rM'(x)-A^{(m)}x\|_\infty] \leq \err_\infty(\rM)$. 

From $\rM'$, we now wish to construct an oblivious mechanism $\widehat{\rM}$ for binary inputs $x \in \{0,1\}^{m}$. For this, consider the following discrete truncated Laplace $\rz$ with range $\{-R,\dots,R\}$:
\[
\Pr[\rz = k] = \frac{\exp(-\eps |k|)}{\sum_{j=-R}^R \exp(-\eps |j|)}.
\]
Notice that for $k \in \{-R+1,\dots,R\}$ we have
\[
\frac{\Pr[\rz = k]}{\Pr[\rz+1 = k]} = \frac{\exp(-\eps |k|)}{\exp(-\eps |k-1|)} \leq \exp(\eps).
\]
The same holds for $\Pr[\rz+1=k]/\Pr[\rz=k]$. For binary $x \in \{0,1\}$ and mechanism $x + \rz$, the boundary values $-R$ and $R$ contribute $\exp(-\eps R)/\sum_{j=-R}^R \exp(-\eps |j|) \leq \exp(-\eps R)$. Setting $R \geq \lceil \ln(1/\delta')/\eps \rceil$ thus suffices for $(\eps,\delta')$-DP of the mechanism $x + \rz$ on binary $x \in \{0,1\}$. We can now execute a variant of the reduction to the oblivious setting by~\cite{ObliviousLinear}. On an input $x \in \{0,1\}^{m}$, the mechanism $\widehat{\rM}$ first samples $\rz_1,\dots,\rz_{m}$ as discrete truncated Laplace variables with range $\{-R,\dots,R\}$. Let $\rz$ be the vector with coordinates $\rz_i$. The mechanism then outputs $\widehat{\rM}(x) = A^{(m)}x + \rM'(\rz) - A^{(m)}\rz$.

The mechanism $\widehat{\rM}$ is clearly oblivious as the noise $\rM'(\rz)-A^{(m)}\rz$ is independent of $x$. Moreover, since $\max_{x \in \{-R,\dots,R\}^m} \E[\|\rM'(x)-A^{(m)}x\|_\infty] \leq \err_\infty(\rM)$, we also have $\E[\|\rM'(\rz)-A^{(m)}\rz\|_\infty] \leq \err_\infty(\rM)$. 

Finally, for privacy, consider two neighboring $x,x' \in \{0,1\}^m$ and assume without loss of generality that $x_i=0$ and $x'_i = 1$. Consider any possible output $y \in \R^n$ and let $\rho(\rM',x,y)$ denote the probability density function of $\rM'$ returning $y$ when the input is $x$. Define $\rho(\widehat{\rM},x,y)$ similarly but for $\widehat{\rM}$. Then
\begin{align*}
\rho(\widehat{\rM},x,y) &= \sum_{z \in \{-R,\dots,R\}^m} \prod_{j=1}^m \Pr[\rz_j=z_j] \rho(\rM', z, y-Ax+Az).
\end{align*}
At the same time, we have
\begin{align*}
\rho(\widehat{\rM},x',y) &= \sum_{z \in \{-R,\dots,R\}^m} \prod_{i=1}^m \Pr[\rz_i=z_i] \rho(\rM', z, y-Ax'+Az) \\
&= \sum_{z \in \{-R,\dots,R\}^m} \prod_{i=1}^m \Pr[\rz_i=z_i] \rho(\rM', z, y-Ax-Ae_i+Az).
\end{align*}
Now for any measurable subset $S \subseteq \R^m$ and $z_i < R$, we have by privacy of $\rM'$ and $z,z+e_i$ being neighboring that
\[
\int_{y \in S} \rho(\rM',z,y-Ax+Az) dy \leq e^{\eps} \int_{y \in S} \rho(\rM', z+e_i, y-Ax+Az) dy + \delta.
\]
We therefore have
\begin{align*}
    &\int_{y \in S} \rho(\widehat{\rM}, x,y) dy \leq\\
    &\sum_{z \in \{-R,\dots,R\}^m} \prod_{j=1}^m \Pr[\rz_j=z_j] \left(e^{\eps} \int_{y \in S} \rho(\rM', z+e_i, y-Ax+Az) dy + \delta \right)
    \leq\\ &\sum_{z \in \{-R,\dots,R\}^m : z_i \neq R} e^\eps \Pr[\rz_i=z_i+1] \prod_{j \neq i} \Pr[\rz_j=z_j] \left(e^{\eps} \int_{y \in S} \rho(\rM', z+e_i, y-Ax+Az) dy + \delta \right)\\ &\qquad{}+\Pr[\rz_i = R] 
    \leq\\
    & \sum_{z \in \{-R,\dots,R\}^m : z_i \neq R} e^{2\eps}  \Pr[\rz_i=z_i+1] \prod_{j \neq i} \Pr[\rz_j=z_j] \left( \int_{y \in S} \rho(\rM', z+e_i, y-Ax+Az) dy  \right) + e^\eps \delta + \delta'.
\end{align*}
Now do the substitution $\tilde{z} = z +e_i$ to conclude
\begin{align*}
&\int_{y \in S} \rho(\widehat{\rM}, x,y) dy \leq\\ 
&\sum_{\tilde{z} \in \{-R,\dots,R\}^m} e^{2\eps}  \Pr[\rz_i=\tilde{z}_i] \prod_{j \neq i} \Pr[\rz_j=\tilde{z}_j] \left( \int_{y \in S} \rho(\rM', \tilde{z}, y-Ax+A\tilde{z}-Ae_i) dy  \right) + e^\eps \delta + \delta' =\\
& e^{2\eps} \int_{y \in S} \rho(\widehat{\rM}, x',y) dy  + e^\eps \delta + \delta'.
\end{align*}
A symmetric argument with the roles of $x$ and $x'$ swapped finally concludes that $\widehat{\rM}$ is $(2 \eps, e^\eps\delta + \delta')$-DP. 
\end{proof}

\paragraph{Integer Inputs.}
The construction also defines an oblivious mechanism for inputs
$x\in\mathbb{Z}^m$, with the same privacy and error guarantees. Indeed,
we use exactly the same formula
\[
  \widehat{\rM}(x)=A^{(m)}x+\rM'(\rz)-A^{(m)}\rz.
\]
The mechanism $\rM'$ is evaluated only at
$\rz\in\{-R,\dots,R\}^m$, which is encoded by the Boolean blocks
above. Thus every call to the original mechanism $\rM$ still has a
Boolean input of length $n$, regardless of the size of $x$.

For privacy, fix an integer $x$, an index $i$, and a measurable output
event $S$. Since the additive noise has the same distribution on every
input, the privacy inequality for the Boolean inputs $0,e_i$, applied
to the translated event $S-A^{(m)}x$, gives
\begin{align*}
  \Pr[\widehat{\rM}(x+e_i)\in S]
  &=\Pr[\widehat{\rM}(e_i)\in S-A^{(m)}x]\\
  &\leq e^{2\eps}\Pr[\widehat{\rM}(0)\in S-A^{(m)}x]
      +e^\eps\delta+\delta'\\
  &=e^{2\eps}\Pr[\widehat{\rM}(x)\in S]
      +e^\eps\delta+\delta'.
\end{align*}
The reverse inequality follows in the same way. Also,
\[
  \E[\|\widehat{\rM}(x)-A^{(m)}x\|_\infty]
  =\E[\|\rM'(\rz)-A^{(m)}\rz\|_\infty]
  \leq\err_\infty(\rM)
\]
for every integer $x$. This does not extend the original mechanism
to integer inputs; it constructs an auxiliary mechanism whose noise
is controlled by the error of the original Boolean-input mechanism.
The same translation argument applies whenever the signal is linear
in the input and the additive noise is independent of the input.

\section{Median-of-Three Trick}
\label{sec:median}
Here we show how to derive Corollary~\ref{cor:median} from the work of~\cite{mediantrick}. For the proof, we use the following auxiliary result by Larsen, Pagh and Tetek:
\begin{lemma}[\cite{mediantrick}]
\label{lem:mediantrick}
Let $f: \R^+ \to \R^+$ be a non-increasing function and let $t$ be a positive integer. Then
\[
\int_{x=0}^\infty f(x^{1/t})^t dx \leq \left( \int_{x=0}^\infty f(x) dx\right)^t.
\]
\end{lemma} 
\begin{proof}[Proof of Corollary~\ref{cor:median}]
We have
\begin{align*}
    \E[\median\{\rz_1,\rz_2,\rz_3\}^2] 
    &\leq \E[\min\{|\rz_1|,|\rz_2|\}^2]  + \E[\min\{|\rz_1|,|\rz_3|\}^2]  +\E[\min\{|\rz_2|,|\rz_3|\}^2]  \\
    &= 3 \cdot \E[\min\{\rz_1^2,\rz_2^2\}].
\end{align*}
Let $p_x = \Pr[|\rz_1| > x]$ and recall the layer-cake formula $\E[|\rz_1|] = \int_{x=0}^\infty p_x dx$. By independence of $\rz_1,\rz_2$, we get
\begin{align*}
\E[\min\{\rz_1^2,\rz_2^2\}] &\leq \int_{x=0}^\infty \Pr[\rz_1^2 > x \wedge \rz_2^2 > x ] dx \\
&= \int_{x=0}^\infty \Pr[\rz_1^2 > x ]^2 dx \\
&= \int_{x=0}^\infty \Pr[|\rz_1| > \sqrt{x} ]^2 dx 
\end{align*}
Letting $f(x) = \Pr[|\rz_1| > x ]$, we get from Lemma~\ref{lem:mediantrick} that
\begin{align*}
\E[\median\{\rz_1,\rz_2,\rz_3\}^2] &\leq 3\cdot \left( \int_{x=0}^\infty \Pr[|\rz_1] > x] dx \right)^2 \\
&= 3 \cdot \E[|\rz_1|]^2.
\end{align*}
\end{proof}

\section{Properties of Attack Vectors}
\label{sec:psiprop}
Here we prove Lemma~\ref{lem:psiprop}, stating that the $\psi^i$ vectors are orthogonal and unit length.
\begin{proof}[Proof of Lemma~\ref{lem:psiprop}]
Let $i \neq j$ with $i,j \in \cL(\T^s)$. Then
\begin{align*}
    \langle \psi^i, \psi^j \rangle &= 2^{-h/2} + \sum_{u : i,j \in \cL(\T^u)} 2^{-(h/2-d)} (-1)^{\one\{i \in \cL(\T^{\ell(u)})\} + \one\{j \in \cL(\T^{\ell(u)})\}}
\end{align*}
If $u$ denotes the lowest common ancestor of the $i$'th and $j$'th leaf, and $u$ is at depth $2d$, then $u$ contributes $-2^{-(h/2-d)}$ to the inner product, the $k$'th ancestor of $u$ contributes $2^{-k} 2^{-(h/2-d)}$, and all other nodes contribute $0$. Finally, the dummy coordinate contributes $2^{-h/2}$. Hence the inner product precisely equals 
\[
-2^{-(h/2-d)} + 2^{-h/2} + \sum_{k=1}^d 2^{-k} 2^{-(h/2-d)} = -2^{-(h/2-d)} + 2^{-h/2} + 2^{-(h/2-d)} - 2^{-h/2} = 0.
\]
Thus the $\psi^i$ vectors are orthogonal. We also have
\begin{align*}
    \|\psi^i\|_2^2 &= 2^{-h/2} + \sum_{d=0}^{h/2-1} 2^{-(h/2-d)} = 1.
\end{align*}
\end{proof}

\end{document}